\documentclass{article}

\usepackage{amsmath,amssymb,amsthm}
\usepackage{bm}
\usepackage{xcolor}
\usepackage{hyperref}
\usepackage{graphicx}
\usepackage{mathtools}
\usepackage{caption}
\usepackage{subcaption}
\usepackage{placeins}
\usepackage[english]{babel}
\usepackage{algorithm}
\usepackage{algpseudocode}
\usepackage[utf8]{inputenc}

\usepackage{parskip}

\usepackage{booktabs}
\usepackage{multirow}

\usepackage[
  top=2.5cm,
  bottom=3.5cm,
  left=3.5cm,
  right=3cm
]{geometry}

\theoremstyle{plain}

\newtheorem{lemma}{Lemma}
\newtheorem{proposition}{Proposition}

\newtheorem*{remark}{Remark}

\title{Geometry-Aware Diffusion Approximate Posterior Sampling for Sparse-View and Limited-Angle CT}

\author{\normalsize{Honglei Brinkmann, Carola-Bibiane Schönlieb, Ander Biguri} \\
\footnotesize{Department of Applied Mathematics and Theoretical Physics, University of Cambridge, Cambridge, UK.}
}
\begin{document}
\maketitle

%%%%%%%%%%%%%%%%%%%%%%%%%%%%%%%%%%%%%%%%%%%%%%%%%%%%%%%%%%%%%%%%%%%%%%%%%%%%%%%% Abstract %%%%%%%%%%%%%%%%%%%%%%%%%%

\begin{abstract}
Sparse-view computed tomography (CT) reduces radiation dose and acquisition time and may mitigate motion artifacts. However, angular undersampling provides insufficient information to determine the image uniquely and stably. Limited-angle CT, arising from restricted angular coverage, produces strongly directional information loss associated with the missing angular range. In both settings, image directions may be strongly observed, weakly constrained, or unobservable, leading to severe ill-posedness and reconstruction ambiguity. Existing diffusion-based approaches incorporate measurement information through likelihood guidance, data-consistency operations, or range--null-space corrections. However, they do not generally use the continuously varying measurement sensitivity of the acquisition to jointly shape both reconstruction updates and stochastic exploration.

We propose a geometry-aware diffusion-guided stochastic reconstruction framework for sparse-view and limited-angle CT. Its central component is a regularized noise-weighted pullback metric constructed from the CT forward operator and measurement-noise covariance. This metric continuously adapts both the measurement-aware update and stochastic exploration according to directional measurement sensitivity, suppressing changes along strongly constrained directions while permitting greater exploration along weakly constrained and unobservable directions. We complement this geometry-aware update with a regularized data-consistency correction and approximately null-space-restricted stochastic perturbations, implemented matrix-free using forward and backprojection operations together with conjugate-gradient solves. Experiments on sparse-view, noisy, and limited-angle CT demonstrate competitive reconstruction quality, strong measurement consistency, and spatially resolved empirical uncertainty estimates.

\end{abstract}

%%%%%%%%%%%%%%%%%%%%%%%%%%%%%%%%%%%%%%%%%%%%%%%%%%%%%%%%%%%%%%%%%%%%%%%%%%%% Introduction %%%%%%%%%%%%%%%%%%%%%%%%%%

\section{Introduction}

%%% General setting for sparse view and limited angle CT reconstruction %%%

 Reducing radiation dose and acquisition time remains a central objective in CT and motivates sparse-view imaging, in which projections are acquired at only a small number of angles. Limited-angle CT arises from a different acquisition constraint: projections are unavailable over a contiguous angular range because of physical or geometric restrictions. Although the two settings produce different patterns of information loss, both leave many image directions unobservable or weakly constrained. Depending on the acquisition geometry and discretization, the corresponding forward operator may be rank deficient or severely ill-conditioned, resulting in an exact null space or a collection of directions that have only a small effect on the measurements. 

From an inverse-problem perspective, this loss of information creates intrinsic ambiguity. Different images may produce identical noiseless measurements or measurement differences too small to distinguish from noise. The measured sinogram therefore cannot uniquely and stably determine the unknown image, and prior information or regularization is required to obtain a plausible reconstruction \cite{Tikhonov1977IllPosed,Bertero1998Inverse}. 

%%% Classical CT reconstrution methods %%%

Classical CT reconstruction methods include analytical methods such as filtered back projection (FBP)~\cite{Kak1988CT} and iterative methods such as ART~\cite{Gordon1970ART}, SART~\cite{Andersen1984SART}, and statistical model-based iterative reconstruction~\cite{Fessler2000Statistical,Thibault2007StatisticalCT}. Under incomplete acquisition, iterative methods commonly combine a data-fidelity term with regularization based on assumptions such as image smoothness, sparsity, or piecewise constancy. Representative examples include total-variation and other sparsity-promoting formulations~\cite{Sidky2008TVCT,Donoho2006CS}. Although these methods improve stability, their performance depends on the selected regularizer, and overly restrictive regularization may remove fine anatomical structures, particularly under severe angular undersampling~\cite{Zhang2018Regularization,Chen2014EdgePreservingTV}. Moreover, these methods usually produce a single reconstruction and do not quantify the uncertainty in image components that are weakly constrained or unobservable from the measurements~\cite{Kaipio2005Statistical,Stuart2010BayesianInverse}.

\iffalse
Although these methods improve reconstruction stability and reduce
artifacts, their performance depends strongly on the selected regularizer. Handcrafted regularizers may oversmooth anatomical structures or remove fine details, particularly under severe angular undersampling~\cite{Zhang2018Regularization,Chen2014EdgePreservingTV}. Moreover, these methods usually produce a single reconstruction and do not quantify the uncertainty in image components that are weakly constrained or unobservable from the measurements~\cite{Kaipio2005Statistical,Stuart2010BayesianInverse}.
\fi

%%% Learning-based CT reconstrution methods %%%

\iffalse
Related model-based frameworks include plug-and-play
priors~\cite{Venkatakrishnan2013PnP} and regularization by
denoising~\cite{Romano2017RED}, which incorporate an external denoiser into the reconstruction algorithm. These frameworks are not inherently learning-based, but they can use pretrained neural denoisers as learned image priors.
\fi

Between classical regularization and fully learned reconstruction, plug-and-play priors~\cite{Venkatakrishnan2013PnP} and regularization by denoising~\cite{Romano2017RED} incorporate an external denoiser into an iterative reconstruction algorithm. These frameworks are not inherently learning-based, because the denoiser may be either handcrafted or learned.

More recently, learning-based methods have introduced data-driven image models into CT reconstruction. Representative approaches include supervised reconstruction networks~\cite{Jin2017FBPConvNet}, unrolled reconstruction networks~\cite{Adler2018LearnedPrimalDual}, deep image priors~\cite{Ulyanov2018DIP,Baguer2020DIPCT}, and implicit neural representations~\cite{Sitzmann2020SIREN,Song2023PINER}. The latter two
approaches optimize a network for the observed data and therefore do not necessarily require paired training images.

These methods can achieve strong reconstruction quality
\cite{DM4CT2024}. However, supervised methods require representative training data and may be sensitive to changes in anatomy, acquisition geometry, or noise level. Methods such as deep image priors and implicit neural representations do not require training on an external dataset. Instead, they represent the unknown image with a neural network whose parameters are fitted directly to each input sinogram through the CT forward model. Moreover, many of these approaches return a single reconstruction and do not directly characterize ambiguity caused by missing or noisy measurements.

%%% Diffusion-based reconstrution methods for General case %%%
% i moved some methods to ralated works. I shall just talk about the famous and general ones here
% I changed the perspective of linear and non linear to the way of update in diffusion model
% I think I should menstion but not make it too long

Diffusion and score-based models learn distributions of clean images and can therefore provide learned prior information for inverse problems. At each diffusion timestep, a pretrained diffusion model estimates the score of the noisy-image distribution obtained from clean CT images. This score provides a time-dependent learned prior direction that guides reconstruction toward plausible clean images. Measurement information can then be incorporated through conditional reverse transitions, likelihood-gradient guidance, or explicit data-consistency operations~\cite{Song2021ScoreSDE,Kawar2022DDRM,Chung2023DPS}. For example, DDRM constructs operator-conditioned reverse transitions for linear inverse problems~\cite{Kawar2022DDRM}; DPS incorporates the
measurement likelihood through gradient guidance~\cite{Chung2023DPS}; DDS incorporates conjugate-gradient data-consistency updates into diffusion sampling~\cite{Chung2024DDS}; and ReSample applies an optimization-based hard data-consistency step followed by stochastic resampling~\cite{Song2024ReSample}.

Diffusion models have also been applied to incomplete-view CT through sinogram restoration, image-domain degradation models, likelihood-guided sampling, and iterative data-consistency corrections~\cite{Yang2025CTSDM,Chen2025CvGDiff,Liu2024FourierDiffusion,Guan2024GMSD}.
A related family of methods exploits range--null-space structure, separating image components that affect the measured sinogram from components that are invisible to the acquisition. DDNM~\cite{Wang2023DDNM} uses a pseudoinverse to impose data consistency on the measurement-sensitive range-space component of each
diffusion estimate while retaining null-space content predicted by the diffusion model. RN-SDE~\cite{Guo2024RNSDE} and its noise-aware extension RNSD$^+$~\cite{Guo2025RNSDPlus} use a trained diffusion process to recover information missing from limited-angle measurements and repeatedly correct
the component that is visible in the acquired data. The conditional score-based null-space (CSN) model~\cite{Zhang2025CSN} similarly uses a score model to estimate missing image information while using an iterative least-squares step to enforce agreement with the measured projections. Like DDNM, these methods use range--null-space-based corrections that treat the data-visible and measurement-invisible components separately: measurement information corrects the range-space component, whereas the learned diffusion model estimates or retains the complementary null-space component.

In contrast, our method does not make a binary range--null-space separation its primary sampling mechanism. It uses a regularized noise-weighted pullback metric to continuously adapt both guided updates and stochastic exploration across image directions according to their measurement sensitivity. Approximate null-space exploration and regularized data-consistency correction are complementary components of this geometry-aware framework.

%%% What is still not fully explored %%%
Although these methods improve incomplete-view CT reconstruction, they do
not fully account for the fact that incomplete measurements constrain some
image directions strongly, other directions only weakly, and exact
null-space directions not at all. Existing methods use the forward operator
through likelihood guidance, data-consistency optimization, or separate
range--null-space corrections. However, they do not generally use the same continuously varying measurement sensitivity to govern both the measurement-aware update and
stochastic exploration.

%%% Our work %%%

In this work, we develop a geometry-aware diffusion-guided stochastic
reconstruction framework for sparse-view and limited-angle CT. Its central
component is a regularized noise-weighted pullback metric constructed from
the CT forward operator and measurement-noise covariance. The metric
continuously adapts the measurement-aware update and stochastic exploration
according to the sensitivity of each image direction to the acquired data:
directions that strongly affect the measured sinogram receive smaller
updates and less stochastic variation, whereas weakly constrained directions
receive relatively greater exploration. The framework additionally includes
approximate null-space exploration and regularized data-consistency
correction, both evaluated matrix-free through forward and backprojection
operations together with conjugate-gradient solves.

%%% Contribution %%%

\paragraph{Contributions.}
The main contributions of this work are as follows:
\begin{itemize}
\item We propose a geometry-aware diffusion-guided stochastic
reconstruction framework for sparse-view and limited-angle CT, combining a
learned image prior with a noise-weighted measurement likelihood. Its central
component is a regularized noise-weighted pullback metric that continuously
adapts both the measurement-aware update and stochastic exploration
according to directional measurement sensitivity.

\item We complement the whole-space metric update with approximate
null-space exploration and a regularized data-consistency correction.
All required inverse actions are evaluated matrix-free using forward and
backprojection operations together with conjugate-gradient solves, without
explicitly forming a Moore--Penrose pseudoinverse.

\item We demonstrate that the proposed framework achieves competitive
reconstruction quality, strong agreement with the measured data, and
spatially resolved empirical uncertainty estimates under sparse-view, noisy,
and limited-angle CT acquisition.
\end{itemize}

The remainder of this paper is organized as follows.
Section~\ref{related work} reviews diffusion-based inverse-problem methods
and incomplete-view CT approaches, positioning the proposed framework within
the existing literature. Section~\ref{method} first formulates the
incomplete-view CT posterior and its directional measurement geometry, and
then develops the proposed pullback-metric sampling framework, its
data-consistency correction, and approximate null-space exploration.
Section~\ref{expreiments} evaluates reconstruction quality, measurement
consistency, empirical sample variability, and the individual roles of the
geometry-aware components. Section~\ref{sec:lConclusion} concludes the paper.
The appendices~\ref{appendix} provide supporting derivations, implementation details, and
supplementary experimental results.

%%%%%%%%%%%%%%%%%%%%%%%%%%%%%%%%%%%%%%%%%%%%%%%%%%%%%%%%%%%%%%%%%%%%% Realted Work %%%%%%%%%%%%%%%%%%%%%%%%%%%%%%%%%

\section{Related Work}
\label{related work}

%%% 1. Diffusion and Score Based Priors for Inverse Problems %%%
% I think about the detailed methods for other works I shall move to the comparison of direct related methods rather than generaly case
% I think maybe I should listen to Claude to use a paragragh to seperate different categories;)

\paragraph{Diffusion priors for inverse problems.}
Diffusion and score-based generative models provide learned image priors for ill-posed inverse problems. The score-based SDE framework formulates generation as a reverse-time stochastic process and provides a general basis for incorporating measurement information during sampling~\cite{Song2021ScoreSDE}. DDRM constructs operator-conditioned reverse transitions for linear inverse problems~\cite{Kawar2022DDRM}, whereas DPS
uses approximate likelihood-gradient guidance for noisy linear and nonlinear inverse problems~\cite{Chung2023DPS}. MCG and $\Pi$GDM refine conditional-score guidance through manifold-constrained and pseudoinverse-inspired constructions, respectively~\cite{Chung2022MCG,Song2023PiGDM}. Other approaches combine pretrained diffusion models with optimization or consistency operations, including DDS,
ReSample, PSLD, and DMPlug~\cite{Chung2024DDS,Song2024ReSample,Rout2023PSLD,Wang2024DMPlug}.

%% 2. one geometry aware using mcmc 
\paragraph{Riemannian MCMC methods.}
Riemannian Langevin methods form a class of Markov chain Monte Carlo (MCMC) methods in which a metric determines how far the sampler moves and how much random perturbation it adds in different parameter-space directions~\cite{Girolami2011Riemannian}. Thus, they provide a general principle for using geometry to shape both deterministic and stochastic sampling components. Our framework uses this principle in a different setting: diffusion-guided reconstruction from incomplete CT measurements. Rather than constructing a general state-dependent posterior metric, we derive a noise-weighted pullback geometry from the CT forward operator and measurement-noise covariance. This geometry represents the directional sensitivity of the particular CT acquisition and is used to adapt the measurement-aware update and stochastic exploration during reconstruction.

%%% 2. Diffusion Models for incomplete CT %%%

\paragraph{Diffusion models for incomplete-view CT.}
Several diffusion models have been developed specifically for sparse-view and limited-angle CT. CT-SDM models view reduction through a deterministic sinogram-domain degradation process, whereas CvG-Diff learns to reverse image-domain artifacts caused by angular undersampling
~\cite{Yang2025CTSDM,Chen2025CvGDiff}. Other methods use sinogram-domain generation, spatial-frequency representations, projection conditioning, or physics-informed score models
~\cite{Liu2024FourierDiffusion,Guan2024GMSD,Peng2025ProjConsistent,Han2025PSDM}. These approaches show that diffusion priors can effectively reduce incomplete-view artifacts. They typically encode the acquisition geometry through a learned degradation, restoration, or conditioning process. In
contrast, our framework incorporates the CT forward operator and measurement-noise model directly into the reconstruction dynamics. The resulting noise-weighted pullback geometry continuously adapts both the measurement-aware update and stochastic exploration according to directional measurement sensitivity.

%%% 2. Anisotropic Diffusion Processes (Georgios.) I am not sure this would be very proper, since there are other works for the anisotropic diffusion processes, this part I have to check later %%%
\iffalse
A related but distinct line of work by Batzolis et al.~\cite{Batzolis2022NonUniform, Batzolis2021Conditional} introduces non-uniform diffusion models, in which different pixels or wavelet bands evolve under different forward diffusion SDEs. This yields multi-scale generative models with faster sampling and a principled multi-speed estimator for the conditional score. Our setting differs substantially: their non-uniformity acts on the forward diffusion process at training time, with a diagonal diffusion matrix of per-coordinate speeds, and is motivated by multi-scale efficiency for unconditional and conditional generation; we modify the posterior Langevin sampler at inference time, with a full non-diagonal noise covariance induced by the forward operator via the pullback metric, motivated by the geometry of ill-posed inverse problems. Both works share the broader observation that breaking the isotropic-Gaussian assumption can be beneficial in diffusion-style sampling, but the mechanisms, structures, and theoretical guarantees differ.
\fi

%%% 4. Hard Data Consistency and Range--null-space decompositon %%%
% I think I should talk about them in the same paragraph, because they are usually used jointly
% Most important, these two are two extra things we enhance our results but not main focus

\paragraph{Data consistency and range--null-space methods.}
Explicit data-consistency operations are widely used in diffusion-guided reconstruction. DDS combines diffusion sampling with conjugate-gradient data-consistency updates, while ReSample alternates diffusion generation, optimization-based consistency correction, and resampling
~\cite{Chung2024DDS,Song2024ReSample}. These operations can reduce disagreement with the acquired data. However, by themselves, they do not provide a common continuously varying noise-weighted
geometry for both the measurement-aware reconstruction update and stochastic exploration across strongly and weakly observed image directions.

A related family of methods uses range--null-space structure. For each diffusion estimate, DDNM uses the acquired measurements and a pseudoinverse to correct the image content that affects the measurements. It retains the remaining measurement-invisible content from the diffusion model~\cite{Wang2023DDNM}. RN-SDE uses a conditional mean-reverting diffusion process to reconstruct limited-angle CT images and repeatedly applies a range--null-space correction: the data-visible content is adjusted using the acquired measurements, while the missing content is inferred by the learned model~\cite{Guo2024RNSDE}. Its extension RNSD$^+$ addresses noisy limited-angle CT by treating missing views as a sinogram-inpainting problem and introducing a noise-aware
rectification mechanism~\cite{Guo2025RNSDPlus}. Similarly, CSN combines conditional score-based reconstruction with iterative least-squares data consistency so that missing information is estimated principally in the null space~\cite{Zhang2025CSN}. These methods conceptually treat data-visible and measurement-invisible components separately.

In contrast, our framework does not first decompose each diffusion estimate into range-space and null-space components in order to construct its primary sampling update. Instead, the main diffusion-guided update is defined over the whole image space and is continuously adapted by a regularized noise-weighted pullback metric. We then use complementary range--null-space-inspired operations: a regularized data-consistency
correction modifies the measurement-sensitive component, whereas an approximately null-space-restricted perturbation promotes additional exploration in weakly constrained and unobservable directions. Thus, range--null-space structure is used in supporting data-consistency and null-space-exploration steps rather than as a binary decomposition underlying the main update.

% and the resulting uncertainty remains unstructured and implicit.

%The measurement-aware or sinogram-conditioned diffusion methods, such as CT-SDM, condition diffusion model directly on measurement data,

%limiting the ability to model structured uncertainty modeling.

% under severe undersampling. Nevertheless, while data consistency constrains the solution in the row space, variability in the the null-space is implicitly governed by the diffusion prior. As a result, the posterior uncertainty is not explicitly modeled, but instead induced by the learned prior. 

%Despite these advances, existing diffusion-based CT methods primarily focus on improving reconstruction quality and enforcing data consistency, without explicitly modeling the geometric structure induced by the forward operator. In particular, the decomposition of the solution space into data-consistent (range-space) and null-space components is not incorporated into the sampling dynamics. As a result, the inherently anisotropic uncertainty arising from incomplete measurements in sparse-view and limited-angle CT is not explicitly characterized or exploited. In contrast, our work explicitly aligns diffusion sampling with the geometry of the measurement manifold, restricting stochastic exploration to null-space directions while enforcing strict data consistency throughout sampling.

%%%%%%%%%%%%%%%%%%%%%%%%%%%%%%%%%%%%%%%%%%%%%%%%%%%%%%%%%%%%%%%%%%%%%%% Method %%%%%%%%%%%%%%%%%%%%%
\section{Method}
\label{method}

In sparse-view and limited-angle CT, the measurements constrain some image directions more, while other directions are less constrained or unobservable. We therefore propose a diffusion-score-guided stochastic reconstruction framework that uses the CT forward operator and noise covariance to adapt both the direction and scale of each sampling update according to noise-weighted measurement sensitivity. The framework reduces the measurement residual through a regularized data-consistency correction. At the same time it concentrates stochastic exploration along weakly constrained and unobservable directions, where reconstruction ambiguity remains.

The proposed method is presented in six parts. We first introduce the forward model, the noise-weighted posterior, and the resulting anisotropy between strongly observed, weakly constrained, and unobservable image directions in Section~\ref{sub3.1}. Section~\ref{sub:pullbacl} uses this geometry to construct pullback-metric-preconditioned Langevin dynamics in place of standard isotropic sampling. Section~\ref{sub:direct_sum} describes the orthogonal decomposition of the image space into range-space (containing directions that affect the measurements) and null-space (containing directions that are invisible to the forward operator). Section~\ref{sub:project_consistency} introduces a regularized Riemannian data-consistency correction that avoids requiring an explicitly available generalized inverse. Section~\ref{sub:null_space_sampling} introduces approximately null-space-restricted stochastic perturbations that promote exploration along weakly constrained and unobservable image directions while suppressing changes in the measured sinogram. 

Finally, Section~\ref{sec:algorithm} combines all components into an iterative sampling procedure that (i) preconditions diffusion updates using the measurement-induced pullback metric, (ii) reduces the measurement residual through a regularized data-consistency correction, and (iii) controls stochastic exploration according to the measurement sensitivity of different image directions, thus supporting spatial uncertainty estimation.

%%%%%%%%%%%%%%%%%%%%% Problem Formulation %%%%%%%%%%%%%%%%%%%%%
%%%%%%%%%%%%%%%%%%%%%%%%%%%%%%%%%%%%%%%%%%%%%%%%%%%%%%%%%%%%%%%%
\subsection{Problem Formulation and Posterior Geometry}
\label{sub3.1}

In the context of CT reconstruction, our goal is to recover images from projection measurements acquired through a forward operator. The CT measurements constrain different image directions to different degrees: some directions strongly affect the measured sinogram, whereas others are weakly constrained or completely unobservable. We formulate sparse-view and limited-angle CT reconstruction as Bayesian inference. A posterior formulation combines the information provided by the measurements with a learned image prior and allows the remaining ambiguity to be represented by multiple reconstruction samples.

We consider sparse-view and limited-angle CT reconstruction under the linear forward model:
\begin{equation}
y \;=\; A_\Omega x + \varepsilon,
\end{equation}
where $x \in \mathbb{R}^n$ denotes the unknown image, $A_\Omega \in \mathbb{R}^{m \times n}$ is the forward operator restricted to a subset of view angles $\Omega$,
$y \in \mathbb{R}^m$ is the measured log-sinogram. The measurement error is modeled as
\begin{equation*}
    \varepsilon\sim\mathcal N(0,\Sigma),
\end{equation*}
where $\Sigma\in\mathbb R^{m\times m}$ is the measurement-noise covariance matrix and $\Sigma\succ0$.

\vspace{0.5em}
\begin{remark}{Origin of the Gaussian noise model.}\label{para:noise}
Under the idealized Beer--Lambert model, the expected photon count at detector bin $i$ is~\cite{taubmann2018ct,Whiting2006Sinogram}
\begin{equation}
    \lambda_i \;=\;I_0\, e^{-(A_\Omega x)_i},
\end{equation}
where $I_0$ is the incident photon count. The detected count is modeled as
\begin{equation*}
    Y_i\sim\operatorname{Poisson}(\lambda_i).
\end{equation*}

For sufficiently large transmitted counts, applying the logarithmic transform to the photon measurements gives the approximate post-log Gaussian model
\begin{equation}\label{eq:gaussian-noise-model}
    y \;\approx\; A_\Omega x + \varepsilon, 
    \qquad 
    \varepsilon \sim \mathcal{N}(0,\Sigma),
    \qquad
    \Sigma \;=\; \mathrm{diag}\!\left(\tfrac{1}{\lambda_i}\right)_{i=1}^m.
\end{equation}
Thus, highly attenuated rays have lower transmitted photon counts and larger log-domain noise variance. The ideal count-based plug-in precision is
$\widehat{\lambda}_i=I_0e^{-y_i}$. In practice, we use the stabilized, normalized, and clipped effective precision described in Appendix~\ref{app:gaussian-approx}. 
\end{remark}

Sparse-view CT acquires projections at fewer angular views while typically retaining the full angular coverage, whereas limited-angle CT restricts the angular range over which projections are acquired. Let $A_\Omega\in\mathbb R^{m\times n}$, where $n$ is the dimension of the image space and $m$ is the dimension of the measurement space. Depending on the acquisition geometry and discretization, $A_\Omega$ may be rank deficient or severely ill-conditioned.

The null-space of the forward operator is
\begin{equation}
    \mathcal N(A_\Omega)
    \;=\;
    \left\{
        v\in\mathbb R^n:A_\Omega v=0
    \right\}.
\end{equation}
By the rank--nullity theorem,
\begin{equation}
    \dim\mathcal N(A_\Omega)
    \;=\;
    n-\operatorname{rank}(A_\Omega).
\end{equation}
When $m<n$, we have $\operatorname{rank}(A_\Omega)\leq m<n$, so a nontrivial null-space is guaranteed. For any $v\in\mathcal N(A_\Omega)$,
\begin{equation}
    A_\Omega(x+v)
    \;=\;
    A_\Omega x.
\end{equation}
Thus, an image change along $v$ is completely invisible to the forward operator.

In the presence of measurement noise, however, directions outside the exact null space may produce sinogram changes that are difficult to distinguish from noise. Under the Gaussian measurement model, the likelihood is
\begin{equation}
    p(y\mid x)
    \propto
    \exp\left(
        -\frac{1}{2}
        \left\|
            \Sigma^{-1/2}(A_\Omega x-y)
        \right\|_2^2
    \right).
    \label{eq:likelihood}
\end{equation}
The covariance $\Sigma$ accounts for the different noise levels of the sinogram measurements.

Let $p_0(x)$ denote the learned distribution of clean CT images. Bayes' rule gives
\begin{equation}
    p(x\mid y)
    \propto
    p_0(x)
    \exp\left(
        -\frac{1}{2}
        \left\|
            \Sigma^{-1/2}(A_\Omega x-y)
        \right\|_2^2
    \right).
    \label{eq:posterior}
\end{equation}
The posterior combines the learned image prior with noise-weighted measurement information~\cite{Stuart2010BayesianInverse}. Posterior
sampling therefore uses the prior to explore image components that are unobservable or only weakly constrained by the measurements.

To quantify the measurement sensitivity along a unit-norm direction
$v\in\mathbb R^n$, $\|v\|_2=1$, consider the perturbed image $x+av$,
where $a\in\mathbb R$ controls the perturbation magnitude. The second derivative of the negative log-likelihood $ \Phi(x)
    :=
    \frac{1}{2}
    \left\|
        \Sigma^{-1/2}(A_\Omega x-y)
    \right\|_2^2$ with respect to $a$ is
\begin{equation}
    \frac{d^2}{da^2}\Phi(x+av)
     \;=\;
    v^\top A_\Omega^\top\Sigma^{-1}A_\Omega v
     \;=\;
    \left\|
        \Sigma^{-1/2}A_\Omega v
    \right\|_2^2.
    \label{eq:directional-data-curvature}
\end{equation}
This quantity is the likelihood curvature along $v$ and measures how
strongly the direction affects the measurements relative to the noise.

Directions with
$A_\Omega v\neq0$ but small
$\|\Sigma^{-1/2}A_\Omega v\|_2$ are mathematically observable but difficult to estimate reliably from noisy measurements. We refer to these directions as effective null-space directions. Therefore, the exact null-space describes completely unobservable directions, whereas the effective null-space additionally captures directions that are only weakly
observable relative to the noise.

Because the quantity in Equation~\eqref{eq:directional-data-curvature}
depends on the perturbation direction $v$, we can see that the likelihood induces an anisotropic geometry: strongly observed, weakly constrained, and unobservable directions receive different amounts of measurement information. The total posterior uncertainty is determined jointly by this measurement information and the learned prior, since the prior may still constrain directions that are invisible to the forward operator. 
Under a Laplace approximation around a posterior mode $x^\star$, the posterior is locally approximated by a Gaussian whose precision is the Hessian of the negative log-posterior:
\begin{equation}
    H_{\mathrm{post}}(x^\star)
    \;=\;
    H_{\mathrm{prior}}(x^\star)
    +
    A_\Omega^\top\Sigma^{-1}A_\Omega,
\end{equation}
where
$H_{\mathrm{prior}}(x^\star)
\;=\;-\nabla_x^2\log p_0(x^\star)$
is the local curvature contributed by the prior. For any
$v\in\mathcal N(A_\Omega)$,
\begin{equation}
    v^\top H_{\mathrm{post}}(x^\star)v
    \;=\;
    v^\top H_{\mathrm{prior}}(x^\star)v,
\end{equation}
because the measurement contribution vanishes along $v$. For a weakly constrained direction, the measurement contribution in Equation~\eqref{eq:directional-data-curvature} is nonzero but small. A formal statement and proof are provided in Appendix~\ref{app:posterior-geometry},
Proposition~\ref{prop:anisotropic-curvature}. This measurement-induced geometry motivates the geometry-aware sampling framework developed in the following subsection, which adapts the sampling dynamics to the noise-weighted sensitivity of different image directions.

%%%%%%%%%%%%% Geometry-Aware Diffusion Sampling via Pullback metric%%%%%%%%%%%%%%%%%%%%%%%%%%%%%%%%%%%%%%%%%%%%%%%%%

\subsection{Geometry-Aware Diffusion Sampling via a Pullback Metric}
\label{sub:pullbacl}

As discussed in Section~\ref{sub3.1}, the likelihood constrains different image directions to different degrees. It provides strong information along well-observed directions, limited information along weakly constrained directions, and no direct information along exact null-space directions. Standard isotropic Langevin dynamics~\cite{Langevin1908}, which combine gradient-based drift with Gaussian noise for posterior sampling~\cite{Roberts1996}, do not account for these differences because their stochastic term has the same variance in every image direction. We therefore construct a Riemannian metric from the CT forward operator and effective noise precision to adapt both the deterministic and stochastic updates to measurement sensitivity. 

%%%%%%%%%%%%%%%%%%%%% rewrite %%%%%%%%%%%%%%%%%%%%%%%%%%%%

To construct the metric, consider a small image perturbation
$\delta x\in\mathbb R^n$. Under the linear forward model, it produces the measurement perturbation
\begin{equation*}
    \delta y\;=\;A_\Omega\delta x.
\end{equation*}
The Gaussian likelihood in~\eqref{eq:likelihood} is
\begin{equation*}
    p(y\mid x)
    \;\propto \;
    \exp\left(
        -\frac{1}{2}
        (A_\Omega x-y)^\top
        \Sigma^{-1}
        (A_\Omega x-y)
    \right).
\end{equation*}
Here, $\Sigma$ denotes the measurement-noise covariance used during
reconstruction, estimated as described in Appendix~\ref{app:gaussian-approx}. Its inverse $\Sigma^{-1}\succ0$ is called the precision matrix. Because $\Sigma^{-1}$ is symmetric positive definite, it defines an inner product on the measurement space $\mathcal{Y}=(\mathbb R^m,\langle \cdot,\cdot\rangle_{\mathcal Y})$
\begin{equation}
    \langle u,v\rangle_{\mathcal Y}
    \;:=\;
    u^\top\Sigma^{-1}v,
    \qquad
    u,v\in\mathbb R^m.
\end{equation}
Under the conditionally independent noise model used in this work, $\Sigma^{-1}$ is diagonal, so the inner product can be written as
\begin{equation}
    \langle u,v\rangle_{\mathcal Y}
    \;=\;
    \sum_{i=1}^m
    (\Sigma^{-1})_{ii}u_iv_i.
\end{equation}
Each diagonal entry $(\Sigma^{-1})_{ii}$ represents the effective precision of measurement bin $i$. Measurements with higher precision receive greater weight, whereas noisier measurements with lower precision receive less weight. Thus, $\|\cdot\|_{\mathcal Y}$ measures a sinogram perturbation relative to the reliability of the measurements.

\paragraph{Pullback metric on image space.} 

The inner product $\langle \cdot, \cdot \rangle_\mathcal{Y}$ defines a geometry on the measurement space $\mathbb{R}^m$, but the unknown image $x$ lives in the image space $\mathbb{R}^n$. To transfer this geometry  to the image space, we use the concept of a \emph{pullback metric}~\cite{Lee2013,doCarmo1992},  a standard construction in differential geometry that transports a metric from one space to another via a smooth map. Given a small image 
perturbation $\delta x \in \mathbb{R}^n$, the forward operator $A_\Omega$  maps it to the corresponding measurement perturbation $\delta y = A_\Omega\,\delta x \in \mathbb{R}^m$. The significance of $\delta x$ in image space can then be assessed by measuring the significance of its image $\delta y$ in measurement space under $\langle \cdot, \cdot \rangle_\mathcal{Y}$. Pulling the measurement-space inner product back through $A_\Omega$ gives
\begin{equation}
    \langle\delta x_1,\delta x_2\rangle_G
    \;:=\;
    \langle A_\Omega\delta x_1,A_\Omega\delta x_2\rangle_{\mathcal Y}
    \;=\;
    \delta x_1^\top G\delta x_2,
    \qquad
    G=A_\Omega^\top\Sigma^{-1}A_\Omega.
    \label{eq:pullback_form}
\end{equation}
This bilinear form measures image perturbations according to how strongly and reliably they affect the measured sinogram. The matrix $G$ is symmetric positive semi-definite (a pullback pseudo-metric on the image space), not a genuine Riemannian metric. It is singular whenever
$\mathcal N(A_\Omega)$ is nontrivial, because
\begin{equation*}
    Gv\;=\;0
    \qquad
    \text{for every }v\in\mathcal N(A_\Omega).
\end{equation*}
We therefore define the regularized metric
\begin{equation}
\label{eq:Ggamma}
G_\gamma \;=\; A_\Omega^\top \Sigma^{-1} A_\Omega \;+\; \gamma I, \qquad \gamma > 0,
\end{equation}
The regularization shifts every eigenvalue of $G$ by $\gamma$, making $G_\gamma$ symmetric positive definite. Directions with data-dependent eigenvalues much larger than $\gamma$ are only slightly affected and retain their measurement-dependent weighting. Along exact null-space directions, the measurement-induced term vanishes because the measurements contribute no information. Only the isotropic Euclidean regularization $\gamma I$ remains. By contrast, weakly constrained
directions retain a small measurement-dependent contribution together with this regularization.

Having defined the measurement-induced metric $G_\gamma$, we now
incorporate it into posterior sampling. Langevin dynamics provide a natural framework because they combine gradient-based guidance toward a target density with stochastic exploration. We first consider an ideal continuous-time process with access to the exact posterior score and then derive the practical diffusion-guided approximation used in our algorithm.
%%%%%%%%%%%%%%%%%%%%%%%%%%%%%%%%%%%%%%%%%%%%%%%%%%%%%%%%%%%

\paragraph{Pullback-Preconditioned Langevin Dynamics.}

Let $X_\tau\in\mathbb R^n$ denote the random image state at continuous
sampling time $\tau$, and let $\pi(x)=p(x\mid y)$ be the target posterior
density. Standard isotropic overdamped Langevin dynamics combine
deterministic drift along the gradient of the log-target density with
isotropic Brownian noise~\cite{Roberts1996}:
\begin{equation}
    dX_\tau
    \; = \;
    \nabla\log\pi(X_\tau)\,d\tau
    +
    \sqrt{2}\,dW_\tau,
    \label{eq:standard_langevin}
\end{equation}
where $W_\tau$ is an $n$-dimensional standard Brownian motion. The first term is the deterministic drift, while the second is the stochastic term. Over a small time interval $\Delta\tau$, the stochastic term contributes a Gaussian increment with covariance $2\Delta\tau I_n$, where $I_n\in\mathbb R^{n\times n}$ is the identity matrix on the image space. It therefore applies the same variance in every image direction.

The process is constructed to have $\pi$ as an invariant density: if
$X_0$ has density $\pi$, then $X_\tau$ has the same density for every
$\tau\geq0$, as verified below.

To incorporate the measurement geometry, we precondition both the drift and
the stochastic term using $G_\gamma^{-1}$:
\begin{equation}
    dX_\tau
   \; = \;
    G_\gamma^{-1}\nabla\log\pi(X_\tau)\,d\tau
    +
    \sqrt{2G_\gamma^{-1}}\,dW_\tau.
    \label{eq:riemannian_langevin}
\end{equation}
The first term is the
geometry-preconditioned drift, while the second is the geometry-shaped
stochastic term. These dynamics suppress variation along directions with
high noise-weighted measurement sensitivity and permit larger updates along
weakly constrained and unobservable directions.

To ensure that the geometry-aware preconditioning changes the sampling
dynamics without changing the target posterior, we verify that $\pi$
remains an invariant density of the ideal continuous process.

Let $\rho_\tau(x)$ denote the probability density of $X_\tau$. For a
stochastic differential equation of the form
\begin{equation*}
    dX_\tau
   \; = \;
    b(X_\tau)\,d\tau
    +
    B\,dW_\tau,
\end{equation*}
where $b:\mathbb R^n\to\mathbb R^n$ is the drift function and
$B\in\mathbb R^{n\times n}$ is a constant diffusion matrix, the
corresponding Fokker--Planck equation is~\cite{Pavliotis2014Stochastic}
\begin{equation*}
    \partial_\tau\rho_\tau
    \; = \;
    -\nabla\!\cdot
    \left(
        b\,\rho_\tau
    \right)
    +
    \frac{1}{2}
    \nabla\!\cdot
    \left(
        BB^\top\nabla\rho_\tau
    \right).
\end{equation*}
For Equation~\eqref{eq:riemannian_langevin},
\begin{equation*}
    b(x)
    \; = \;
    G_\gamma^{-1}\nabla\log\pi(x),
    \qquad
    BB^\top
    \; = \;
    2G_\gamma^{-1}.
\end{equation*}
Substituting these expressions gives
\begin{align*}
    \partial_\tau\rho_\tau
    &\; = \;
    -\nabla\!\cdot
    \left(
        G_\gamma^{-1}\rho_\tau\nabla\log\pi
    \right)
    +
    \nabla\!\cdot
    \left(
        G_\gamma^{-1}\nabla\rho_\tau
    \right).
\end{align*}
To test stationarity, set $\rho_\tau=\pi$. Since
\begin{equation*}
    \nabla\log\pi
    \; = \;
    \frac{\nabla\pi}{\pi},
    \qquad\text{and}\qquad
    \pi\nabla\log\pi
    \; = \;
    \nabla\pi,
\end{equation*}
we obtain
\begin{align*}
    \left.
    \partial_\tau\rho_\tau
    \right|_{\rho_\tau=\pi}
    \; = \;
    -\nabla\!\cdot
    \left(
        G_\gamma^{-1}\nabla\pi
    \right)
    +
    \nabla\!\cdot
    \left(
        G_\gamma^{-1}\nabla\pi
    \right)
    \; = \;0.
\end{align*}
Therefore, if $X_0$ has density $\pi$, then $X_\tau$ has density $\pi$
for every $\tau\geq0$. Hence, $\pi$ is an invariant density of the ideal
continuous dynamics. A formal statement and the corresponding conditions are provided in Appendix~\ref{app:fp-proof}.

\paragraph{Diffusion-score approximation.}

The ideal dynamics in Equation~\eqref{eq:riemannian_langevin} require the exact posterior score. From the posterior factorization, this score has the
form
\begin{equation*}
    \nabla\log\pi(x)
   \; = \;
    \nabla\log p_0(x)
    -
    \nabla\Phi(x),
\end{equation*}
where the data-fidelity gradient $\nabla\Phi(x)$ can be evaluated from the known forward operator, measured sinogram, and fixed noise precision, whereas the clean-image prior score $\nabla\log p_0(x)$ is not available analytically. We use a denoising diffusion probabilistic model (DDPM)~\cite{Ho2020DDPM} pretrained on clean CT images. During pretraining, Gaussian noise is added to the images at different diffusion timesteps, and the model learns to predict the added noise. During reconstruction, the model parameters remain fixed, and the predicted noise is used to estimate the image-prior score that guides the reconstruction toward plausible CT
images.

Let $x^{(k)}$ denote the current reconstruction at sampling iteration $k$, and let $t_k$ be the corresponding DDPM timestep. For simplicity, all image variables below are expressed in the normalized intensity range used during diffusion-model training.

At each iteration, we treat $x^{(k)}$ as the current estimate of the clean image and apply the standard DDPM forward-noising equation:
\begin{equation}
    x_{t_k}^{(k)}
    \;=\;
    \sqrt{\bar{\alpha}_{t_k}}\,x^{(k)}
    +
    \sqrt{1-\bar{\alpha}_{t_k}}\,\epsilon^{(k)},
    \qquad
    \epsilon^{(k)}
    \sim
    \mathcal N(0,I_n),
    \label{eq:model-noised-state}
\end{equation}
where $\bar{\alpha}_{t_k}$ is the cumulative product of the DDPM noise schedule up to timestep $t_k$. The variable $x_{t_k}^{(k)}$ is the noisy image supplied to the pretrained diffusion model.

For an $\epsilon$-prediction DDPM, the network $\epsilon_\theta(x_{t_k}^{(k)},t_k)$ predicts the Gaussian noise contained in $x_{t_k}^{(k)}$. The corresponding score estimate is~\cite{Ho2020DDPM,Song2021ScoreSDE}
\begin{equation}
    s^{(k)}
    \;=\;
    -
    \frac{
        \epsilon_\theta(x_{t_k}^{(k)},t_k)
    }{
        \sqrt{1-\bar{\alpha}_{t_k}}
    },
    \label{eq:ddpm-score}
\end{equation}
which approximates
$\nabla_{x_{t_k}}\log p_{t_k}(x_{t_k}^{(k)})$. This score supplies the learned image-prior direction at noise level $t_k$.

The data-fidelity gradient is evaluated at the current reconstruction
$x^{(k)}$:
\begin{equation}
    g^{(k)}
    \;:=\;
    \nabla\Phi(x^{(k)})
    \;=\;
    A_\Omega^\top\Sigma^{-1}
    \left(
        A_\Omega x^{(k)}-y
    \right),
    \label{eq:data-gradient-iteration}
\end{equation}
where
\begin{equation}
    \Phi(x)
    \;=\;
    \frac{1}{2}
    (A_\Omega x-y)^\top
    \Sigma^{-1}
    (A_\Omega x-y).
\end{equation}
We combine the DDPM score and data-fidelity gradient as
\begin{equation}
    r^{(k)}
    \;=\;
    \lambda_k s^{(k)}
    -
    \mu_k g^{(k)},
    \label{eq:posterior-score}
\end{equation}
where $\lambda_k$ and $\mu_k$ control the relative strengths of the learned prior and measurement guidance. The DDPM score is evaluated at the noisy image $x_{t_k}^{(k)}$, whereas the data-fidelity gradient is evaluated at the current reconstruction $x^{(k)}$. Therefore, $r^{(k)}$ is used as a practical combined guidance direction rather than an exact diffusion-time posterior score.

\paragraph{Geometry-preconditioned update.}
The pullback-metric-preconditioned direction $u^{(k)}$ is obtained by solving
\begin{equation}
    G_\gamma u^{(k)}
    \;=\;
    r^{(k)},
    \qquad
    G_\gamma
    \;=\;
    A_\Omega^\top\Sigma^{-1}A_\Omega+\gamma I_n.
    \label{eq:drift-cg-system}
\end{equation}
The current reconstruction is then updated according to
\begin{equation}
    \tilde{x}^{(k)}
    \;=\;
    x^{(k)} 
    + \alpha_k u^{(k)},
    \qquad
    \tilde{x}^{(k)}
    \leftarrow
    \tilde{x}^{(k)}
    +
    \nu_G^{(k)},
    \qquad
    \nu_G^{(k)}
    \sim
    \mathcal N\!\left(0,2\alpha_kG_\gamma^{-1}\right),
    \label{eq:update-correct}
\end{equation}
where $\alpha_k>0$ is the sampling step size. The deterministic term $\alpha_k u^{(k)}$ combines the learned DDPM prior with the
noise-weighted data-fidelity gradient, while $\nu_G^{(k)}$ provides metric-shaped stochastic exploration.

Because the evolving reconstruction $x^{(k)}$ is not itself maintained as a DDPM noisy state, we apply the forward-noising equation before each score evaluation. This produces an input whose noise level is consistent with the assigned DDPM timestep $t_k$. The resulting score provides annealed prior guidance, ranging from coarse structure at high noise levels to finer
details at low noise levels. This procedure differs from a standard DDPM reverse trajectory, since a newly noised state is constructed from the current reconstruction at each iteration.

\paragraph{Matrix-free implementation via conjugate gradient.}
The preconditioned direction in~\eqref{eq:drift-cg-system} is computed by
solving a linear system of the general form
\begin{equation}
    G_\gamma u\;=\;r,
\end{equation}
where $r\in\mathbb R^n$ is known and
$u=G_\gamma^{-1}r\in\mathbb R^n$ is the unknown. We solve this via the conjugate gradient (CG) method~\cite{Hestenes1952CG,Hansen2018AIRToolsII} (Appendix~\ref{app:CG}), which avoids forming
$G_\gamma \in \mathbb{R}^{n \times n}$ explicitly. CG requires only matrix-vector products with $G_\gamma$. Expanding the definition $G_\gamma = A_\Omega^\top\Sigma^{-1}A_\Omega + \gamma I$
and applying it to a search direction $p$ gives
\begin{align}
  G_\gamma p
  &\;=\; \bigl(A_\Omega^\top\Sigma^{-1}A_\Omega + \gamma I\bigr)p \notag\\
  &\;=\; A_\Omega^\top\Sigma^{-1}A_\Omega p + \gamma p \notag\\
  &\;=\; A_\Omega^\top\bigl(\Sigma^{-1}(A_\Omega p)\bigr) + \gamma p.
  \label{eq:matvec}
\end{align}
Each CG iteration requires one forward projection and one backprojection, together with diagonal precision weighting. Since $G_\gamma\succ0$, the linear system has a unique solution and is suitable for conjugate gradient. In exact arithmetic, CG recovers this solution in at most $n$ iterations. In practice, it is terminated when the residual falls below a prescribed tolerance. The regularization parameter $\gamma$ bounds the smallest eigenvalue away from zero and improves numerical conditioning.

To sample
$\nu_G^{(k)}\sim
\mathcal N(0,2\alpha_kG_\gamma^{-1})$
matrix-free, draw independent standard Gaussian vectors
$\xi^{(k)}\sim\mathcal N(0,I_m)$ in measurement space and
$c^{(k)}\sim\mathcal N(0,I_n)$ in image space, and form
\begin{equation}
    w^{(k)}
    \;=\;
    A_\Omega^\top\Sigma^{-1/2}\xi^{(k)}
    +
    \sqrt{\gamma}\,c^{(k)}
    \sim
    \mathcal N(0,G_\gamma).
\end{equation}
Solving
$G_\gamma z^{(k)}=w^{(k)}$
and setting
$\nu_G^{(k)}=\sqrt{2\alpha_k}\,z^{(k)}$
gives
$\nu_G^{(k)}
\sim\mathcal N(0,2\alpha_kG_\gamma^{-1})$. This covariance is exact when the linear system is solved exactly. With a finite CG tolerance, the generated covariance approximates $2\alpha_kG_\gamma^{-1}$, with accuracy determined by the linear-solve residual. The full correctness proof is given in Appendix~\ref{app:fp-proof}, Lemma~\ref{lem:matrix-free-sampler}.

The update~\eqref{eq:update-correct} is anisotropic because the
measurement-induced metric preconditions both the deterministic direction and the stochastic covariance. It suppresses stochastic variation and prior-driven updates along strongly observed directions while allowing greater variation along weakly constrained and unobservable directions. The resulting variability is determined jointly by the measurement
geometry and the learned prior. The regularized data-consistency correction and additional approximately null-space-restricted perturbations are introduced in Sections~\ref{sub:project_consistency} and
\ref{sub:null_space_sampling}, respectively.

%%%%%% Orthogonal Direct Sum Decomposition %%%%%%%%%%%%%%
%%%%%%%%%%%%%%%%%%%%%%%%%%%%%%%%%%%%%%%%%%%%%%%%%%%%%%%%%

\subsection{Orthogonal Direct Sum Decomposition}
\label{sub:direct_sum}

To exploit the geometric structure of $A_\Omega$ explicitly, we need a precise definition of how image space splits into measurement-constrained and measurement-invisible directions. The orthogonal direct sum decomposition\cite{Schwab2019NullSpace,Chen2020DeepDecomposition,Chen2021Equivariant} provides this definition, and the two subspaces identified by it will correspond directly to the two components of the sampling process in Sections~\ref{sub:project_consistency} 
and~\ref{sub:null_space_sampling}.

The forward operator maps the image space to the measurement space:
\begin{equation}
    A_\Omega:\mathbb R^n\rightarrow\mathbb R^m.
\end{equation}
Accordingly, $\mathcal R(A_\Omega)\subseteq\mathbb R^m$ lies in the measurement space and describes sinograms that can be produced by the forward model. To describe measurement-sensitive directions in the image
space, we instead use $\mathcal R(A_\Omega^\top)\subseteq\mathbb R^n$.

\vspace{0.5em}
\begin{proposition}[Orthogonal decomposition]
\label{thm:orthogonal-decomposition}
For any $A_\Omega\in\mathbb R^{m\times n}$, the image space admits the orthogonal direct-sum decomposition
\begin{equation}
    \mathbb R^n
    \;=\;
    \mathcal R(A_\Omega^\top)
    \oplus
    \mathcal N(A_\Omega).
    \label{eq:orthogonal-direct-sum}
\end{equation}
Equivalently,
\begin{equation}
    \mathcal R(A_\Omega^\top)
    \;=\;
    \mathcal N(A_\Omega)^\perp.
\end{equation}
Here, $\mathcal R(A_\Omega^\top)$ contains image directions that can affect the measurements, whereas $\mathcal N(A_\Omega)$ contains directions that are invisible to the forward operator. A proof is provided in Appendix~\ref{app:orthogonal-decomposition}.
\end{proposition}

Consequently, every image $x\in\mathbb R^n$ has a unique decomposition
\begin{equation}
    x
    \;=\;
    x_R+x_N,
    \qquad
    x_R\in\mathcal R(A_\Omega^\top),
    \qquad
    x_N\in\mathcal N(A_\Omega).
    \label{eq:range-null-decomposition}
\end{equation}
Because $A_\Omega x_N=0$, the predicted sinogram depends only on the range-space component:
\begin{equation}
    A_\Omega x
    \;=\;
    A_\Omega x_R.
\end{equation}
If the equation $A_\Omega x=y$ is feasible, then every exactly measurement-consistent image can be written as
\begin{equation}
    x
    \;=\;
    x_{\mathrm{dc}}+v,
    \qquad
    x_{\mathrm{dc}}\in\mathcal R(A_\Omega^\top),
    \qquad
    v\in\mathcal N(A_\Omega),
    \label{eq:consistent-range-null}
\end{equation}
where $x_{\mathrm{dc}}$ is the unique range-space component satisfying $A_\Omega x_{\mathrm{dc}}=y$. The measurements therefore determine the range-space component under the exact feasible constraint, while the null-space component remains unidentifiable from the measurements alone.

Under the noisy likelihood used in this work, the distinction is not strictly binary. The likelihood provides no information along exact null-space directions, but it may provide only weak information along directions associated with small noise-weighted singular values. Moreover, the learned prior can influence both range-space and null-space components. The decomposition in~\eqref{eq:range-null-decomposition} therefore
describes exact algebraic observability, while the pullback metric also accounts for continuously varying measurement sensitivity.

This structure motivates the two subsequent components of the method. The regularized data-consistency correction in
Section~\ref{sub:project_consistency} acts along
$\mathcal R(A_\Omega^\top)$, while the stochastic perturbation in Section~\ref{sub:null_space_sampling} concentrates variation along weakly constrained and unobservable directions.
%%%%%%%%%%%%%%%%%%%%%%%%%%%%%%%%%%%%%%%%%%%%%%%%%%%%%%%%
%%%%%%%%%%%% Projection-Based Data Consistency %%%%%%%%%

\subsection{Data Consistency via Riemannian Gradient Projection}
\label{sub:project_consistency}

Although the geometry-aware update incorporates the measurements through the data-fidelity gradient, one Langevin update can move it away from measurement consistency. We therefore seek a correction $x^+$ that remains close to the intermediate estimate $\tilde{x}$ while reducing its mismatch with the measured sinogram $y$.

Define the data-fidelity functional and its Euclidean gradient:
\begin{equation}
  \Phi(x) \;=\; \tfrac{1}{2}(A_\Omega x - y)^{\top}\Sigma^{-1}(A_\Omega x - y),
  \qquad
  \nabla\Phi(x) \;=\; A_\Omega^{\top}\Sigma^{-1}(A_\Omega x - y).
  \label{eq:phi}
\end{equation}
Under the pullback metric $G_{\gamma_{\mathrm{proj}}} = A_\Omega^{\top}\Sigma^{-1}A_\Omega + \gamma_{\mathrm{proj}} I$,
the Riemannian gradient of $\Phi$ is
\begin{equation}
  \operatorname{grad}_{G_{\gamma_{\mathrm{proj}}}}\Phi(x)
  \;=\; {G^{-1}_{\gamma_{\mathrm{proj}}}}\nabla\Phi(x)
  \;=\; {G^{-1}_{\gamma_{\mathrm{proj}}}}A_\Omega^{\top}\Sigma^{-1}(A_\Omega x - y).
\label{eq:riemannian_grad}
\end{equation}
We use a separate regularization parameter $\gamma_{\mathrm{proj}}>0$ for the data-consistency correction.
This parameter controls the trade-off between measurement agreement and numerical stability. It is distinct from the pullback-metric parameter $\gamma$ used in the diffusion update and the null-space parameter $\gamma_{\mathrm{null}}$ used in Section~\ref{sub:null_space_sampling}.

For sparse-view and limited-angle CT, $A_\Omega$ may be rank deficient or
severely ill-conditioned. We therefore formulate data consistency directly as the regularized proximity problem
\begin{equation}
  x^+
  \;=\;\arg\min_x
  \left\{
  \tfrac{1}{2}\|x-\tilde{x}\|_2^2
  +\tfrac{\tau}{2}\|A_\Omega x-y\|_\Sigma^2
  \right\},
  \qquad
  \|v\|_\Sigma^2:=v^\top\Sigma^{-1}v,
  \label{eq:proximal_problem}
\end{equation}
where $\tau>0$ controls the strength of data consistency. 

To obtain the minimizer, we define
\begin{equation}
J(x)
\;=\;\tfrac{1}{2}\|x-\tilde{x}\|_2^2
+\tfrac{\tau}{2}
(A_\Omega x-y)^\top\Sigma^{-1}(A_\Omega x-y).
\end{equation}
Differentiating $J(x)$ with respect to $x$ gives
\begin{equation}
\nabla J(x)
\;=\;x-\tilde{x}+\tau A_\Omega^\top\Sigma^{-1}(A_\Omega x-y).
\end{equation}
At the minimum $x=x^+$, this gradient is zero. Therefore,
\begin{equation}
x^+-\tilde{x}+\tau A_\Omega^\top\Sigma^{-1}(A_\Omega x^+-y)
\;=\;0.
\end{equation}
Rearranging the terms gives
\begin{equation}
\left(
I_n+\tau A_\Omega^\top\Sigma^{-1}A_\Omega
\right)x^+\;=\;
\tilde{x}+
\tau A_\Omega^\top\Sigma^{-1}y.\end{equation}
Setting $\gamma_{\mathrm{proj}}=1/\tau$ and rearranging yields
\begin{equation}
x^+\;=\;\tilde{x}-
{G^{-1}_{\gamma_{\mathrm{proj}}}}A_\Omega^\top\Sigma^{-1}
(A_\Omega\tilde{x}-y),
\qquad
{G_{\gamma_{\mathrm{proj}}}}=A_\Omega^\top\Sigma^{-1}A_\Omega+\gamma_{\mathrm{proj}} I_n.
\label{eq:proj-closed-form_n}
\end{equation}
Equation~\eqref{eq:proj-closed-form_n} is also one unit
Riemannian-gradient step, 
\begin{equation}
x^+\;=\;\tilde{x}-\operatorname{grad}_{G_{\gamma_{\mathrm{proj}}}}\Phi(\tilde{x}),
\end{equation}
so the data-consistency correction and the Langevin drift use the same pullback geometry. Since ${G_{\gamma_{\mathrm{proj}}}}$ is symmetric positive definite for every $\gamma_{\mathrm{proj}}>0$, the correction is well-defined regardless of the rank of $A_\Omega$. For finite $\gamma_{\mathrm{proj}}$, it improves measurement consistency but does not generally enforce $A_\Omega x^+=y$ exactly. 

The term ${G^{-1}_{\gamma_{\mathrm{proj}}}}A_\Omega^{\top}\Sigma^{-1}(A_{\Omega} \tilde{x} - y)$ in the update~\eqref{eq:proj-closed-form_n} lies in $\mathcal{R}(A_\Omega^{\top})$ by construction, so the update modifies only the measurement-sensitive component of $\tilde{x}$ while leaving its null-space component intact. This is complementary to the geometry-aware diffusion update, in which the projection enforces near-hard measurement fidelity, while the diffusion dynamics govern prior-driven reconstruction and uncertainty exploration in the null-space.

\vspace{1em}
\begin{remark}[Regularized data consistency under noisy measurements]
For a measured sinogram $y$, define the exact consistency set
\begin{equation}
    \mathcal C_y
    \;:=\;
    \left\{
        x\in\mathbb R^n:A_\Omega x=y
    \right\}.
\end{equation}
This set is nonempty if and only if
$y\in\mathcal R(A_\Omega)$. When it is nonempty,
$\mathcal C_y$ is a closed affine subspace of $\mathbb R^n$,
rather than a curved manifold. If $y\notin\mathcal R(A_\Omega)$,
then no image can reproduce the measured sinogram exactly.

Under the noisy measurement model $y=A_\Omega x^\star+\varepsilon$, exact consistency may be infeasible
because the noise realization $\varepsilon$ need not lie in
$\mathcal R(A_\Omega)$. Even when the constraint is feasible, enforcing $A_\Omega x^+=y$ exactly fits the noisy measurement and may introduce noise-related features into the reconstruction.

The update in~\eqref{eq:proj-closed-form_n} is therefore implemented as a regularized data-consistency correction rather than an exact projection. It is well-defined for every
$\gamma_{\mathrm{proj}}>0$, regardless of the rank of $A_\Omega$ or the feasibility of the exact constraint. Smaller values of
$\gamma_{\mathrm{proj}}$ produce stronger agreement with the measurements along observable directions, whereas larger values give a weaker and more strongly regularized correction with improved numerical conditioning.
\end{remark}

%803

%%%%%%%%%%%%%%%%%%%%%%%%%%%%%%%%%%%%%%%%%%%%%%%%%%%%%

\subsection{Approximately Null-Space-Restricted Stochastic Sampling}
\label{sub:null_space_sampling}

The regularized data-consistency correction in Section~\ref{sub:project_consistency} acts along measurement-sensitive directions in $\mathcal R(A_\Omega^\top)$ and does not directly modify the exact null-space component. Reconstruction ambiguity may therefore remain along exact null-space directions and along directions that are only weakly constrained relative to the measurement noise. To explore this ambiguity, we introduce stochastic perturbations that retain more variation along weakly constrained and unobservable directions while suppressing measurement-sensitive components.

\paragraph{Exact null-space invariance.}
Suppose that $x$ satisfies $A_\Omega x=y$. For any
$v\in\mathcal N(A_\Omega)$,
\begin{equation}
    A_\Omega(x+v)
   \; = \;
    A_\Omega x+A_\Omega v
    \; = \;
    y.
\end{equation}
Thus, an exact null-space perturbation does not change the predicted sinogram. An exact orthogonal projection onto
$\mathcal N(A_\Omega)$ can be expressed using the Moore--Penrose
pseudoinverse of the discretized forward operator. Although this
pseudoinverse exists mathematically, accurately applying it can be
computationally demanding for large-scale and ill-conditioned CT systems. We therefore use a regularized approximation based on forward and backprojection operations.

\paragraph{Approximate null-space perturbation.}
At reconstruction iteration $k$, we draw
\begin{equation}
    \zeta^{(k)}
    \sim
    \mathcal N(0,I_n),
\end{equation}
where $I_n$ is the $n\times n$ identity matrix. The isotropic random vector $\zeta^{(k)}$ generally contains components along both measurement-sensitive and null-space directions.

We represent its measurement-sensitive component approximately as
$A_\Omega^\top q^{(k)}$, where $q^{(k)}\in\mathbb R^m$ is a measurement-space coefficient vector. We therefore seek a perturbation of the form
\begin{equation}
    v^{(k)}
    \; = \;
    \zeta^{(k)}-A_\Omega^\top q^{(k)}.
    \label{eq:null-perturb}
\end{equation}
Requiring exact null-space membership would give
\begin{equation}
\begin{aligned}
    A_\Omega v^{(k)}
    &\; = \;
    A_\Omega\zeta^{(k)}
    -
    A_\Omega A_\Omega^\top q^{(k)}\\
    &\; = \;
    0,
\end{aligned}
\end{equation}
and hence
\begin{equation}
    A_\Omega A_\Omega^\top q^{(k)}
    \; = \;
    A_\Omega\zeta^{(k)}.
\end{equation}
Because $A_\Omega A_\Omega^\top$ may be singular or severely
ill-conditioned, we replace this system with the noise-weighted
regularized system
\begin{equation}
    \left(
        A_\Omega A_\Omega^\top
        +
        \gamma_{\mathrm{null}}\Sigma
    \right)q^{(k)}
    \; = \;
    A_\Omega\zeta^{(k)},
    \qquad
    \gamma_{\mathrm{null}}>0.
    \label{eq:null-cg-system}
\end{equation}
Because $\Sigma\succ0$ and $\gamma_{\mathrm{null}}>0$, the matrix in \eqref{eq:null-cg-system} is symmetric positive definite, even when $A_\Omega$ is rank deficient. It can therefore be solved using conjugate gradient.

Combining~\eqref{eq:null-perturb} and~\eqref{eq:null-cg-system} gives
\begin{equation}
    v^{(k)}
    \; = \;
    \left[
        I_n
        -
        A_\Omega^\top
        \left(
            A_\Omega A_\Omega^\top
            +
            \gamma_{\mathrm{null}}\Sigma
        \right)^{-1}
        A_\Omega
    \right]\zeta^{(k)}.
    \label{eq:null-filter}
\end{equation}
The same perturbation can equivalently be written in image space as
\begin{equation}
    v^{(k)}
    \; = \;
    \gamma_{\mathrm{null}}
    \left(
        A_\Omega^\top\Sigma^{-1}A_\Omega
        +
        \gamma_{\mathrm{null}}I_n
    \right)^{-1}
    \zeta^{(k)}.
    \label{eq:null-filter-image}
\end{equation}

To explain the filtering effect, let $\phi_i$ be an eigenvector of
$A_\Omega^\top\Sigma^{-1}A_\Omega$ with eigenvalue $\mu_i\geq0$. Along
this direction,
\begin{equation}
    \gamma_{\mathrm{null}}
    \left(
        A_\Omega^\top\Sigma^{-1}A_\Omega
        +
        \gamma_{\mathrm{null}}I_n
    \right)^{-1}\phi_i
    \; = \;
    \frac{\gamma_{\mathrm{null}}}
    {\mu_i+\gamma_{\mathrm{null}}}\phi_i.
\end{equation}
If $\mu_i=0$, then $A_\Omega\phi_i=0$, so $\phi_i$ is an exact
null-space direction and its component is retained. If $\mu_i$ is large, the direction is strongly constrained by the noise-weighted measurements and its random component is strongly suppressed. If $\mu_i$ is small but nonzero, the direction is weakly constrained and its component is partially retained.

For finite $\gamma_{\mathrm{null}}>0$, the perturbation is not generally an exact null-space vector. Indeed, using
\eqref{eq:null-cg-system},
\begin{equation}
\begin{aligned}
    A_\Omega v^{(k)}
    &\; = \;
    A_\Omega\zeta^{(k)}
    -
    A_\Omega A_\Omega^\top q^{(k)}\\
    &\; = \;
    \gamma_{\mathrm{null}}\Sigma q^{(k)}.
\end{aligned}
\end{equation}
Thus, $v^{(k)}$ may change the predicted sinogram. The parameter
$\gamma_{\mathrm{null}}$ controls the regularized filtering of
measurement-sensitive components. As
$\gamma_{\mathrm{null}}\to0$, the spectral factors approach one for exact null-space directions and zero for directions with $\mu_i>0$. Therefore, the operator in~\eqref{eq:null-filter} approaches the orthogonal projector onto $\mathcal N(A_\Omega)$.
\vspace{0.5em}
\begin{remark}
The eigenvalue analysis is used only to explain the filtering effect. No eigenvectors or eigenvalues are computed in practice. Instead, $q^{(k)}$ is obtained by solving~\eqref{eq:null-cg-system} using conjugate gradient. Each iteration requires applications of the forward operator, the backprojection operator, and the diagonal noise covariance. No pseudoinverse or dense matrix is formed explicitly.
\end{remark}

\paragraph{Sampling update.}
The approximately null-space-restricted perturbation is added to the intermediate reconstruction according to
\begin{equation}
    \tilde{x}^{(k)}
    \leftarrow
    \tilde{x}^{(k)}
    +
    \beta_kv^{(k)},
    \label{eq:null-sampling-update}
\end{equation}
where $\beta_k\geq0$ controls the perturbation magnitude. Because
$v^{(k)}$ is only approximately null-space restricted, this update may change the predicted sinogram. The subsequent regularized data-consistency correction reduces the resulting measurement residual before producing $x^{(k+1)}$.
%%%%%%%%%%%%%%%%%%%%%%%%%%%%%%%%%%%%%%%%%%%%%%%%%%%%%%%%%%%%%%%%%%%%%% Overall Algorithm

\subsection{Overall Algorithm}
\label{sec:algorithm}

The components introduced above are combined into the iterative reconstruction procedure summarized in
Algorithm~\ref{alg:geo_diffusion}. Algorithm~\ref{alg:geo_diffusion} begins from an initial reconstruction $x^{(0)}=\operatorname{Initialize}(y)$ (Line~1). The initialization may be obtained using FBP, an iterative reconstruction method, an approximate
inverse, or another suitable procedure. The algorithm then performs $T$ reconstruction iterations (Lines~2--28), with one DDPM timestep $t_k$ assigned to each iteration $k$.

\textbf{DDPM score and geometry-preconditioned update
(Lines~3--9).} At iteration $k$, the algorithm first draws
$\epsilon^{(k)}\sim\mathcal N(0,I_n)$ (Line~3) and constructs a noisy version of the current reconstruction using the DDPM forward-noising equation (Line~4). The pretrained $\epsilon$-prediction network then provides the DDPM score estimate $s^{(k)}$ at timestep $t_k$ (Line~5). 

The data-fidelity gradient $g^{(k)}$ is evaluated at the current reconstruction $x^{(k)}$ (Line~6). The DDPM score and data-fidelity gradient are combined to form the guidance direction $r^{(k)}$ (Line~7). Because these two terms are evaluated at different states, $r^{(k)}$ is a practical combined guidance direction rather than an exact
diffusion-time posterior score. 

The pullback-metric-preconditioned direction $u^{(k)}$ is obtained using
a conjugate-gradient solve (Line~8). The resulting direction is applied with step size $\alpha_k$ to produce the intermediate estimate $\tilde{x}^{(k)}$ (Line~9).

\textbf{Metric-shaped Langevin noise (Lines~10--15).}
When Langevin noise is enabled (Line~10), the algorithm draws independent standard Gaussian vectors in the measurement and image spaces (Line~11).
These vectors are combined to form
$w^{(k)}\sim\mathcal N(0,G_\gamma)$ (Line~12). Solving
$G_\gamma z^{(k)}=w^{(k)}$ by conjugate gradient (Line~13) gives $z^{(k)}\sim\mathcal N(0,G_\gamma^{-1})$ when the system is solved exactly. The scaled perturbation
$\sqrt{2\alpha_k}\,z^{(k)}$ is then added to the intermediate estimate (Line~14). This produces stochastic variation whose covariance is adapted to the measurement-induced geometry.

\textbf{Approximately null-space-restricted perturbation
(Lines~16--21).} When null-space perturbation is enabled (Line~16), the algorithm draws an isotropic image-space vector $\zeta^{(k)}$ (Line~17). It then solves the
noise-weighted regularized measurement-space system for $q^{(k)}$ (Line~18). Subtracting the estimated measurement-sensitive component gives
\begin{equation}
    v^{(k)}
    \;=\;
    \zeta^{(k)}-A_\Omega^\top q^{(k)}
\end{equation}
(Line~19). The perturbation is added with magnitude $\beta_k$ (Line~20).

For finite $\gamma_{\mathrm{null}}>0$, $v^{(k)}$ is only approximately in $\mathcal N(A_\Omega)$. It suppresses components that strongly affect the measurements while retaining more variation along weakly constrained and
unobservable directions.

\textbf{Regularized data-consistency correction
(Lines~22--27).} The correction is applied according to the projection schedule (Line~22). At a scheduled iteration, the algorithm solves the regularized
measurement-space system for $\eta^{(k)}$ (Line~23) and backprojects this vector to update the reconstruction (Line~24). The correction acts in $\mathcal R(A_\Omega^\top)$ and reduces the measurement residual without directly modifying the exact null-space component.

Because $\gamma_{\mathrm{proj}}>0$, this operation is a regularized near-hard data-consistency correction rather than an exact orthogonal projection. When no correction is scheduled, the intermediate estimate is passed directly to the next iteration (Lines~25--26).

After completing all $T$ iterations, the algorithm returns the final reconstruction $x^{(T)}$ (Line~29).

The ideal continuous Langevin process in Section~\ref{sub:pullbacl} has the target posterior as an invariant distribution under the stated assumptions. Algorithm~\ref{alg:geo_diffusion} additionally uses a DDPM-based approximate guidance direction, finite step sizes, finite-tolerance CG solves, approximately null-space-restricted perturbations, and scheduled data-consistency corrections. It should therefore be interpreted as a practical geometry-aware stochastic reconstruction procedure rather than an exact discretization of the ideal
posterior dynamics.

\begin{algorithm}[htp]
\centering
\caption{Geometry-Aware Diffusion-Score-Guided Reconstruction}
\label{alg:geo_diffusion}
\begin{algorithmic}[1]
\Require Forward operator $A_\Omega$, measurement $y$, noise covariance
         $\Sigma$, pretrained DDPM noise predictor $\epsilon_\theta$,
         DDPM timesteps $\{t_k\}_{k=0}^{T-1}$ and cumulative coefficients
         $\{\bar{\alpha}_{t_k}\}_{k=0}^{T-1}$, regularisation parameters
         $\gamma$, $\gamma_{\mathrm{null}}$, $\gamma_{\mathrm{proj}}$,
         step sizes $\alpha_k$, weights $\lambda_k$, $\mu_k$, $\beta_k$,
         projection interval $k_{\mathrm{proj}}$, CG tolerance $\tau$

\State $x^{(0)}\gets\operatorname{Initialize}(y)$
       \Comment{initial reconstruction}

\For{$k=0,1,\ldots,T-1$}

    \State Sample $\epsilon^{(k)}\sim\mathcal N(0,I_n)$

    \State $x_{t_k}^{(k)}
    \gets
    \sqrt{\bar{\alpha}_{t_k}}\,x^{(k)}
    +
    \sqrt{1-\bar{\alpha}_{t_k}}\,\epsilon^{(k)}$
    \Comment{DDPM forward noising}

    \State $s^{(k)}
    \gets
    -\epsilon_\theta(x_{t_k}^{(k)},t_k)/
    \sqrt{1-\bar{\alpha}_{t_k}}$
    \Comment{DDPM score estimate}

    \State $g^{(k)}
    \gets
    A_\Omega^\top\Sigma^{-1}
    (A_\Omega x^{(k)}-y)$
    \Comment{data-fidelity gradient}

    \State $r^{(k)}
    \gets
    \lambda_k s^{(k)}-\mu_k g^{(k)}$
    \Comment{combined guidance direction}

    \State Solve
    $(A_\Omega^\top\Sigma^{-1}A_\Omega+\gamma I_n)
    u^{(k)}=r^{(k)}$
    \Comment{CG, tol.\ $\tau$}

    \State $\tilde{x}^{(k)}
    \gets
    x^{(k)}+\alpha_k u^{(k)}$
    \Comment{geometry-preconditioned update}

    \If{Langevin noise is enabled}

        \State Sample
        $\xi^{(k)}\sim\mathcal N(0,I_m)$ and
        $c^{(k)}\sim\mathcal N(0,I_n)$

        \State $w^{(k)}
        \gets
        A_\Omega^\top\Sigma^{-1/2}\xi^{(k)}
        +
        \sqrt{\gamma}\,c^{(k)}$
        \Comment{$w^{(k)}\sim\mathcal N(0,G_\gamma)$}

        \State Solve
        $(A_\Omega^\top\Sigma^{-1}A_\Omega+\gamma I_n)
        z^{(k)}=w^{(k)}$
        \Comment{CG, tol.\ $\tau$}

        \State $\tilde{x}^{(k)}
        \gets
        \tilde{x}^{(k)}
        +
        \sqrt{2\alpha_k}\,z^{(k)}$
        \Comment{metric-shaped noise}

    \EndIf

    \If{null-space perturbation is enabled}

        \State Sample
        $\zeta^{(k)}\sim\mathcal N(0,I_n)$

        \State Solve
        $(A_\Omega A_\Omega^\top
        +\gamma_{\mathrm{null}}\Sigma)
        q^{(k)}
        =
        A_\Omega\zeta^{(k)}$
        \Comment{CG, tol.\ $\tau$}

        \State $v^{(k)}
        \gets
        \zeta^{(k)}-A_\Omega^\top q^{(k)}$
        \Comment{approximate null-space perturbation}

        \State $\tilde{x}^{(k)}
        \gets
        \tilde{x}^{(k)}+\beta_kv^{(k)}$

    \EndIf

    \If{$k\bmod k_{\mathrm{proj}}=0$}

        \State Solve
        $(A_\Omega A_\Omega^\top
        +\gamma_{\mathrm{proj}}\Sigma)
        \eta^{(k)}
        =
        A_\Omega\tilde{x}^{(k)}-y$
        \Comment{CG, tol.\ $\tau$}

        \State $x^{(k+1)}
        \gets
        \tilde{x}^{(k)}-A_\Omega^\top\eta^{(k)}$
        \Comment{regularized data-consistency correction}

    \Else

        \State $x^{(k+1)}\gets\tilde{x}^{(k)}$

    \EndIf

\EndFor

\State \Return $x^{(T)}$
\end{algorithmic}
\end{algorithm}

\begin{table}[t]
\centering
\caption{Parameters of the geometry-aware reconstruction algorithm.}
\label{tab:hyperparameters}
\renewcommand{\arraystretch}{1.2}
\begin{tabular}{p{2.1cm} p{10.0cm} p{1.6cm}}
\toprule
\textbf{Symbol} & \textbf{Role} & \textbf{Status} \\
\midrule

$\gamma$
&
Regularization of the pullback metric
$G_\gamma=A_\Omega^\top\Sigma^{-1}A_\Omega+\gamma I_n$.
It controls the conditioning and anisotropy of the geometry-preconditioned
guidance and Langevin noise.
&
required
\\

$\gamma_{\mathrm{null}}$
&
Regularization of the approximately null-space-restricted perturbation.
It controls how strongly measurement-sensitive components of the
isotropic perturbation are suppressed.
&
required
\\

$\gamma_{\mathrm{proj}}$
&
Regularization of the data-consistency CG system. Smaller values generally
produce a stronger correction, whereas larger values improve numerical
stability and permit a larger measurement residual.
&
required
\\

$\alpha_k$
&
Step size of the geometry-preconditioned reconstruction update at
iteration $k$.
&
required
\\

$\lambda_k$
&
Weight of the DDPM score in the combined guidance direction.
&
required
\\

$\mu_k$
&
Weight of the data-fidelity gradient in the combined guidance direction.
&
required
\\

$\beta_k$
&
Magnitude of the approximately null-space-restricted perturbation.
Setting $\beta_k=0$ disables this component.
&
optional
\\

$t_k$
&
DDPM timestep assigned to reconstruction iteration $k$ by the diffusion
scheduler.
&
fixed
\\

$\bar{\alpha}_{t_k}$
&
Cumulative DDPM noise-schedule coefficient at timestep $t_k$.
&
fixed
\\

$T$
&
Total number of reconstruction iterations and selected DDPM timesteps.
&
required
\\

$k_{\mathrm{proj}}$
&
Interval between regularized data-consistency corrections. A correction
is applied when $k\bmod k_{\mathrm{proj}}=0$.
&
optional
\\

$\tau$
&
Convergence tolerance used by the conjugate-gradient solves.
&
optional
\\

$K_{\mathrm{CG}}$
&
Maximum number of iterations used by each conjugate-gradient solve.
&
optional
\\

\bottomrule
\end{tabular}
\end{table}
The sequences $\{\alpha_k\}$, $\{\lambda_k\}$, $\{\mu_k\}$, and
$\{\beta_k\}$ may be constant or scheduled over the reconstruction
iterations. In our implementation, each quantity is constant by default; a user-specified schedule can vary its value between the beginning and end of sampling. The exact values and schedules used in each experiment are reported in the implementation details.

%%%%%%%%%%%%%%%%%%%%%%%%%%%%%%%%%%%%%%%%%%%%%%%%%%%%%%%%%%%%
%%%%%% Experiments
%%%%%%%%%%%%%%%%%%%%%%%%%%%%%%%%%%%%%%%%%%%%%%%%%%%%%%%%%%%%
%\FloatBarrier
\section{Experiments}
\label{expreiments}

We evaluate the proposed geometry-aware approximate diffusion posterior sampling framework on sparse-view and limited-angle CT reconstruction. The experiments have three goals. First, we assess reconstruction quality by comparing against the DM4CT benchmark~\cite{DM4CT2024} using the standard distortion
metrics PSNR and SSIM, and additionally provide qualitative
comparisons across a range of acquisition degradations. Second, we evaluate whether the proposed method produces informative empirical sample variability rather than only a single point estimate, analysing the uncertainty structure of our method. Third, we verify the central geometry-aware claim: that posterior variability is suppressed in measurement space and concentrated in directions that are either invisible to the acquisition geometry or corrupted by measurement noise.

Experiments are conducted on the medical CT data from the 2016 Low Dose Grand Challenge, following the data preprocessing and evaluation protocol used in DM4CT. Unless otherwise stated, stochastic reconstructions are generated from the same measured sinogram. For noisy configurations, the measurement noise realization is fixed and only the sampler seed varies between samples, ensuring that observed variability reflects posterior uncertainty rather than changes in the measurement itself. 

%%%%%%%%%%%%%%%%%%%%%%%%%%%%%%%%%%%%%%%%%%%%%%%%%%%%%%%%%%%%%%%%%%%
%%%%%%%%%%% the resonstructed one 
%%%%%%%%%%%%%%%%%%%%%%%%%%%%%%%%%%%%%%%%%%%%%%%%%%%%%%%%%%%%%%%%%%%
\subsection{Reconstruction Performance}
\label{sub:reconstruction_performance}
We evaluate reconstruction performance from two complementary
perspectives. First, we compare against existing methods on the standard DM4CT benchmark, reporting PSNR and SSIM to enable direct quantitative comparison with the published benchmark results (Section~\ref{subsubsec:quantitative_benchmark}).
Second, we provide qualitative comparisons across a controlled set of acquisition degradations, to isolate how different kinds of ill-posedness affect reconstruction quality (Section~\ref{subsec:exp_visual_comparison}).
Together, these two evaluations establish that the proposed method
achieves competitive reconstruction quality across the range of CT acquisition conditions in DM4CT.

\subsubsection{Quantitative Benchmark Comparison}
\label{subsubsec:quantitative_benchmark}
%\paragraph{Experimental Settings}
%\label{Experimental Settings}

We follow the standard DM4CT benchmark configurations to enable direct comparison with existing CT reconstruction methods. The five configurations are: (i) 40 projections without noise, (ii) 20 projection angles with mild noise, (iii) 80 projections with higher noise, (iv) 80 projections with mild noise and ring artifacts, and (v) 40 projections with limited-angle views $[0,\frac{3}{4}\pi]$. All methods are evaluated under the same forward operator, preprocessing, and metric computation protocol as DM4CT for fair comparison.

We compare the proposed method with representative approaches from multiple categories as in the DM4CT paper:
\textbf{Classical methods:} filtered back projection (FBP)~\cite{Kak1988CT}, simultaneous iterative reconstruction technique (SIRT)~\cite{Gilbert1972SIRT}, and total variation (TV)-regularized reconstruction~\cite{Sidky2008TVCT}.
\textbf{Unsupervised methods:} deep image prior (DIP)~\cite{Ulyanov2018DIP} and implicit neural representations (INR)~\cite{Sitzmann2020SIREN}.
\textbf{Supervised methods:} a state-of-the-art supervised reconstruction network trained on paired data~\cite{Liang2021SwinIR}.
\textbf{Diffusion-based methods:} representative posterior sampling and data-consistency strategies, including DDS~\cite{Chung2024DDS}, DPS~\cite{Chung2023DPS}, PGDM~\cite{Song2023PiGDM}, and ReSample~\cite{Song2024ReSample}.

Table~\ref{tab:dm4ct_full} reports PSNR and SSIM for all methods included in the DM4CT benchmark~\cite{DM4CT2024} under all five standard configurations, evaluated on the full test set of 526 images. Baseline results are taken directly from DM4CT; results for the proposed method are produced under the same evaluation protocol using $N=10$ stochastic reconstructions averaged to form the sample-mean reconstruction. 

\begin{table*}[t]
\centering
\small
\setlength{\tabcolsep}{3.5pt}
\captionsetup{skip=7pt}
\caption{Full quantitative comparison (PSNR\,dB\,/\,SSIM) on the DM4CT benchmark test set.
Config~i: 40 views (no noise);
Config~ii: 20 views + mild noise;
Config~iii: 80 views + higher noise;
Config~iv: 80 views + noise + ring artifacts;
Config~v: 40 views, limited angle $[0,\tfrac{3}{4}\pi]$.
Baseline results are adopted from DM4CT.
\textbf{Bold}/\underline{underline}: best/second-best among diffusion-based methods;
$^{\dagger}$: overall best across all methods.}
\label{tab:dm4ct_full}
\begin{tabular}{l*{5}{c}}
\toprule
\textbf{Method}
  & \textbf{Config~i}
  & \textbf{Config~ii}
  & \textbf{Config~iii}
  & \textbf{Config~iv}
  & \textbf{Config~v} \\
\midrule
%--- Classical ---
FBP  & 26.98\,/\,0.69 & 9.89\,/\,0.03 & 12.78\,/\,0.09 & 14.50\,/\,0.13 & 24.14\,/\,0.64 \\
SIRT & 30.40\,/\,0.80 & 26.23\,/\,0.47 & 24.48\,/\,0.32 & 25.86\,/\,0.40 & 26.49\,/\,0.72 \\
\noalign{\hrule height 0.8pt}
%--- Iterative optimisation ---
ADMM-PDTV  & 30.56\,/\,0.79 & 25.12\,/\,0.36 & 18.20\,/\,0.10 & 20.11\,/\,0.15 & 29.69\,/\,0.80 \\
FISTA-SBTV & 30.57\,/\,0.82 & 26.28\,/\,0.70 & 25.39\,/\,0.67 & 27.84\,/\,0.74 & 27.68\,/\,0.76 \\
\noalign{\hrule height 0.8pt}
%--- Implicit neural representation ---
DIP        & 28.58\,/\,0.80 & 24.13\,/\,0.61 & 26.40\,/\,0.66 & 27.89\,/\,0.71 & 27.87\,/\,0.75 \\
INR        & 33.21\,/\,0.86
           & 26.15\,/\,0.76
           & 27.74\,/\,0.80
           & 29.50\,/\,0.74
           & 31.01\,/\,0.81 \\
R2Gaussian & 32.14\,/\,0.81
           & 24.90\,/\,0.70
           & 25.45\,/\,0.86$^{\dagger}$
           & 25.45\,/\,0.74
           & 28.26\,/\,0.78 \\
\noalign{\hrule height 0.8pt}
%--- Supervised ---
SwinIR     & 32.45\,/\,0.88
           & 29.92$^{\dagger}$\,/\,0.83$^{\dagger}$
           & 30.37$^{\dagger}$\,/\,0.84
           & 30.79$^{\dagger}$\,/\,0.85$^{\dagger}$
           & 28.93\,/\,0.81 \\
\noalign{\hrule height 0.8pt}
%--- Diffusion-based ---
MCG           & 30.00\,/\,0.79 & 27.50\,/\,0.68 & \textbf{28.90}\,/\,0.71 & 29.12\,/\,0.74 & 28.32\,/\,0.74 \\
DPS           & 30.75\,/\,0.79 & 27.09\,/\,0.73 & 27.81\,/\,0.74 & 28.28\,/\,0.75 & 27.47\,/\,0.72 \\
PSLD          & 26.03\,/\,0.75 & 25.12\,/\,0.73 & 25.77\,/\,0.74 & 26.03\,/\,0.75 & 25.66\,/\,0.72 \\
PGDM          & 30.26\,/\,0.80 & 27.81\,/\,0.70 & 28.44\,/\,0.66 & 29.11\,/\,0.72 & \underline{29.66}\,/\,\underline{0.77} \\
DDS           & 31.43\,/\,0.84 & 20.12\,/\,0.22 & 19.41\,/\,0.25 & 18.25\,/\,0.20 & 28.58\,/\,\underline{0.77} \\
ReSample      & \underline{32.03}\,/\,\underline{0.85}
              & \underline{27.92}\,/\,0.73
              & 28.67\,/\,0.73
              & \textbf{29.70}\,/\,0.76
              & 27.70\,/\,0.76 \\
DMPlug        & 25.77\,/\,0.71 & 25.78\,/\,0.71 & 25.81\,/\,0.71 & 25.70\,/\,0.68 & 20.70\,/\,0.50 \\
Reddiff       & 28.01\,/\,0.78 & 26.87\,/\,0.67 & 27.66\,/\,0.70 & 27.70\,/\,0.73 & 26.04\,/\,0.73 \\
HybridReg     & 27.63\,/\,0.78 & 26.68\,/\,0.67 & 27.40\,/\,0.71 & 27.44\,/\,0.73 & 25.74\,/\,0.73 \\
DiffStateGrad & 27.46\,/\,0.77
              & 26.97\,/\,\underline{0.76}
              & 27.35\,/\,\underline{0.77}
              & 27.36\,/\,\underline{0.77}
              & 24.29\,/\,0.70 \\
\noalign{\hrule height 0.8pt}
\textbf{Ours}
  & \textbf{34.04}$^{\dagger}$\,/\,\textbf{0.89}$^{\dagger}$
  & \textbf{28.98}\,/\,\textbf{0.81}
  & \underline{28.83}\,/\,\textbf{0.80}
  & \underline{29.42}\,/\,\textbf{0.81}
  & \textbf{33.82}$^{\dagger}$\,/\,\textbf{0.88}$^{\dagger}$ \\
\bottomrule
\end{tabular}
\end{table*}

\vspace{1em}

\paragraph{Evaluation protocol.}
Our method generates multiple samples from the posterior distribution $p(x \mid y)$, while some baseline methods produce a single deterministic reconstruction. For a fair quantitative comparison with deterministic benchmarks, we use the posterior mean reconstruction as the representative estimate. We generate $N$ reconstructions $\{x_j^{(i)}\}_{i=1}^{N}$ for test image $j$ and use their sample mean as the representative reconstruction:
\begin{equation}
    \bar{x}_j
    \;:=\;
    \frac{1}{N}\sum_{i=1}^{N}x_j^{(i)}
    \;\approx\;
    \mathbb{E}[x\mid y_j].
    \label{eq:posterior-mean-evaluation}
\end{equation}
The posterior mean is the minimum mean-squared-error estimator and provides
a single reconstruction for comparison with deterministic methods. PSNR
and SSIM are computed between $\bar{x}_j$ and the corresponding
ground-truth image $x_{\mathrm{gt},j}$. We report the dataset averages
\begin{equation}
    \operatorname{PSNR}_{\mathrm{dataset}}
    \;=\;
    \frac{1}{M}\sum_{j=1}^{M}
    \operatorname{PSNR}(\bar{x}_j,x_{\mathrm{gt},j}),
    \qquad
    \operatorname{SSIM}_{\mathrm{dataset}}
    \;=\;
    \frac{1}{M}\sum_{j=1}^{M}
    \operatorname{SSIM}(\bar{x}_j,x_{\mathrm{gt},j}),
    \label{eq:dataset-metrics}
\end{equation}
where $M=526$ is the number of test images, and $N=10$ is the number of samples.

For the noise-free and limited-angle settings (Configs~i and v), the proposed method achieves the best performance
among all compared methods. In Config~i, it obtains
34.04~dB/0.89, outperforming the strongest unsupervised baseline,
INR, by 0.83~dB, the strongest supervised restoration baseline,
SwinIR, by 1.59~dB, and the strongest diffusion-based posterior
baseline, ReSample, by 2.01~dB. The advantage is even larger in
Config~v, where the proposed method reaches 33.82~dB/0.88,
exceeding the next-best method, INR, by 2.81~dB and 0.07 in SSIM.

This increase in margin from Config~i to Config~v is consistent with the role of geometry-aware guidance. In the sparse-view setting, the null-space is induced by missing views but is less directionally structured, so the advantage of geometry-aware comes primarily from the pullback metric redirecting updates away from measurement-sensitive directions. In the limited-angle setting, the null-space is strongly anisotropic and aligned with the missing angular range. The pullback metric explicitly accounts for this geometry, providing a stronger advantage as the null-space becomes more geometrically structured. This is also reflected in the behavior of ReSample, which drops from 32.03~dB in
Config~i to 27.70~dB in Config~v, indicating that hard projection
alone is insufficient for structured null-spaces.

For the noisy sparse-view settings, the proposed method remains competitive and is particularly strong among posterior-sampling methods. In Config~ii, it achieves 28.98 dB/0.81 SSIM, outperforming DPS, DDS, PGDM, and ReSample. In Config~iii, the acquisition contains more projection views but a higher noise level. The increased angular coverage reduces the geometric
ambiguity and therefore limits the relative advantage of null-space-aware sampling. Reconstruction quality is instead determined mainly by the balance between noise suppression and measurement consistency. In addition, the effective measurement precision is based on a post-log Gaussian approximation that is most reliable at moderate-to-high photon counts. Under the stronger noise in Config~iii, deviations from this approximation may lead to less accurate likelihood weighting and reduce
reconstruction performance. Consequently, the proposed method obtains a slightly lower PSNR than MCG (28.83 versus 28.90~dB). However, it achieves substantially higher SSIM (0.80 versus 0.71), indicating better preservation of structural information.

Under ring artifacts (Config~iv), the proposed method obtains
29.42~dB/0.81, giving one of the strongest structural reconstructions among diffusion-based methods. Resample is slightly higher in PSNR, while the proposed method maintains a higher SSIM. One possible explanation is that the projection-consistency term forces the reconstruction to match the measured sinogram. When the sinogram contains ring artifacts, this can also force the reconstruction to keep some of these artifacts, which may lower the PSNR. This suggests that ring artifacts may require artifact-aware noise modelling or adaptive precision weighting, so that unreliable sinogram components are not treated with the same confidence as clean measurements.

Since posterior-sampling methods are stochastic, their performance can depend on the particular sample drawn from the posterior. To assess this variability, we additionally generate 20 stochastic reconstruction samples for DDS, DPS, PGDM, ReSample, and the proposed method on one representative test slice for each acquisition configuration. We select DDS~\cite{Chung2024DDS}, DPS~\cite{Chung2023DPS}, PGDM~\cite{Song2023PiGDM}, and ReSample~\cite{Song2024ReSample} for this analysis because they achieve competitive performance in the benchmark and can generate multiple stochastic reconstruction samples. They therefore provide appropriate baselines for assessing both reconstruction accuracy and sampling variability. The posterior mean and sample variability are reported in Appendix~\ref{app:single_image_comparison}. This analysis is not
intended to replace the benchmark-style comparison in Table~\ref{tab:dm4ct_full}, but to provide additional insight into the stability of the posterior samplers.

\iffalse

\paragraph{Geometry-aware validation.}
For the range-null decomposition and posterior variability diagnostics, we use three representative settings: 40 projections without noise, 40 projections with Poisson noise, and 20 projections with Poisson noise. For the pullback metric ablation, we use two noise-free settings: 40 and 20 projections with full angular range, chosen to isolate the individual effects of projection sparsity and angular range restriction on the geometric structure of the posterior.
\fi

\subsubsection{Qualitative Comparison across Acquisition Degradations}
\label{subsec:exp_visual_comparison}
We compare visually against four representative diffusion-based methods: 
DDS~\cite{Chung2024DDS}, DPS~\cite{Chung2023DPS}, PGDM~\cite{Song2023PiGDM}, and ReSample~\cite{Song2024ReSample}. These methods represent two broad approaches to data consistency. DDS, DPS, and PGDM use gradient- or optimization-guided updates. ReSample combines a
hard data-consistency correction with resampling. All four are top-performing diffusion-based methods in the DM4CT benchmark, making them strong and representative baselines. FBP is included as an analytical reference to the reconstruction difficulty of each setting rather than as a competitive baseline.

For comparison, we use a controlled configuration set 
that varies one factor at a time:
(i)~40 projections without noise,
(ii)~40 projections with mild Poisson noise,
(iii)~40 projections with limited angular range $[0,3\pi/4]$,
(iv)~40 projections with limited angular range $[0,\pi/2]$, and
(v)~10 projections with mild Poisson noise.
This isolates whether reconstruction degradation is caused by an enlarged null-space, missing directional information, or measurement noise. Figure~\ref{fig:method_comparison} presents the qualitative 
results for all five settings.

\begin{figure}[t]
\centering
\includegraphics[width=0.8\linewidth]{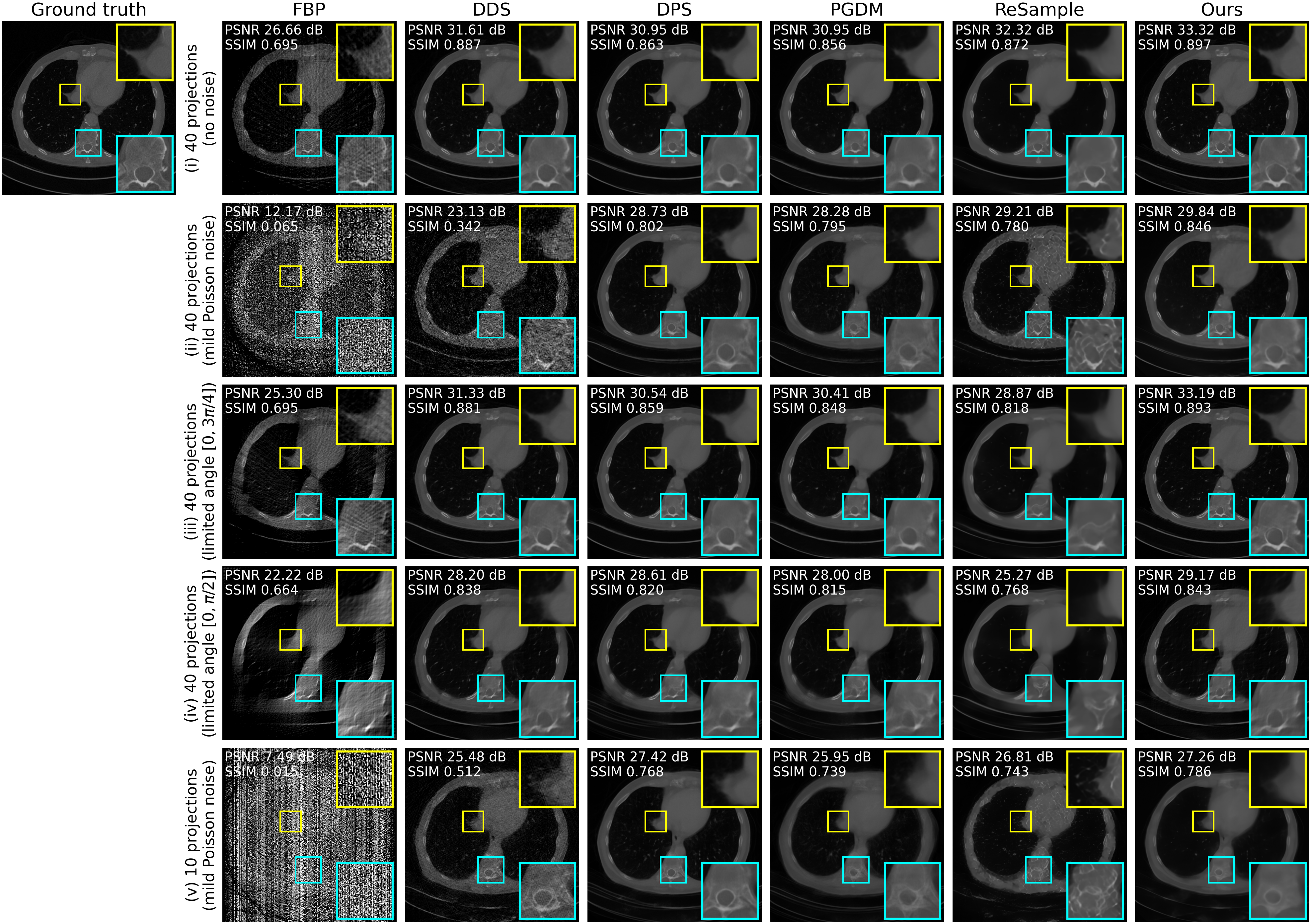}
\caption{Qualitative reconstruction comparison under five controlled 
acquisition-degradation settings. Columns show ground truth, FBP, DDS, DPS, PGDM, ReSample, and the proposed method (Ours). FBP is shown as an analytical reference illustrating acquisition ill-posedness; all diffusion-based methods show the sample mean over $N=20$ stochastic reconstructions. Yellow and cyan insets show zoomed regions of interest. PSNR and SSIM are computed on the posterior mean.}
\label{fig:method_comparison}
\end{figure}

As expected, FBP degrades rapidly as the inverse problem becomes more ill-posed. In the noise-free sparse-view setting~(i), streak artifacts are already visible, and the cyan zoom shows a grid-like interference pattern with reduced cortical bone definition. Under mild Poisson noise~(ii), noise amplification dominates the reconstruction, obscuring both the lung boundary and the vertebral structure. In the limited-angle settings~(iii--iv), FBP exhibits directional blurring and streaking aligned with the missing angular range; this effect becomes stronger as the angular coverage is reduced from $[0,3\pi/4]$ to $[0,\pi/2]$. The most
severe degradation occurs in the highly sparse noisy case~(v), where streak artifacts dominate the image and fine anatomy is largely lost.\footnote{In the benchmark-compatible implementation, the angular grid includes both $0$ and $\pi$. For parallel-beam CT, these two views are redundant up to detector reversal, effectively leaving one fewer distinct projection direction. We keep this geometry unchanged for all methods to ensure a fair comparison.}

The diffusion-based baselines show different behaviour across the five
acquisition settings. DDS performs well when the measurements are noise-free: in setting~(i), it recovers clean anatomy and sharp local structure (31.61~dB, SSIM~0.887), and it remains competitive in the limited-angle settings~(iii--iv). However, it is much less robust to measurement noise. In the mild Poisson case~(ii), its performance drops to 23.13~dB/0.342, with strong noise patterns in the lung region and a clear loss of vertebral structure in the cyan zoom. DPS and PGDM are more stable under noise, achieving 28.73~dB/0.802 and 28.28~dB/0.795 in setting~(ii), respectively. However, their zoomed regions show residual smoothing and less distinct boundary detail, especially in the limited-angle cases. ReSample performs strongly in the sparse-view settings, reaching 32.32~dB/0.872 in setting~(i) and 29.21~dB/0.780 in setting~(ii), but degrades more substantially when the angular range is restricted. In setting~(iv), it drops to 25.27~dB/0.768, and the zoomed regions show clear directional blurring. This suggests that hard projection is effective for sparse-view consistency, but is less sufficient when the null-space is structured by limited angular coverage.

The proposed method gives the strongest overall performance in the first four settings. It achieves 33.32~dB/0.897 in setting~(i), 29.84~dB/0.846 in setting~(ii), 33.19~dB/0.893 in setting~(iii), and 29.17~dB/0.843 in setting~(iv), giving the best PSNR and SSIM in each of these cases. These quantitative gains are consistent with the visual comparison: the proposed method better suppresses noise amplification, reduces directional artifacts, and preserves sharper anatomical boundaries in both zoomed regions. The advantage is particularly clear in the limited-angle setting~(iii), where the proposed method outperforms the strongest competing posterior baseline, DDS, by 1.86~dB and produces sharper local structure in the vertebral zoom. This trend remains visible in the more severe limited-angle setting~(iv). Although the margin is smaller, the proposed method still achieves the best PSNR and SSIM among the posterior methods, reaching 29.17~dB/0.843.

The most challenging setting~(v), with only 10 noisy projections, shows a different trade-off. DPS obtains a slightly higher PSNR than the proposed method (27.42~dB versus 27.26~dB), but the proposed method achieves the highest SSIM among the posterior methods (0.786). Visually, DPS retains slightly more pixelwise contrast, but also shows residual streak artifacts, whereas the proposed method produces a smoother and more coherent reconstruction. This indicates that, in the most severely ill-posed noisy setting, PSNR and SSIM emphasise different aspects of reconstruction quality. In setting~(v), the measurements are both extremely sparse and noisy. Under such conditions, hard or near-hard consistency, as used by ReSample and the proposed method, can preserve noise-related errors from the sinogram. This may limit PSNR, while the posterior mean still improves structural consistency, leading to high SSIM values.

Overall, the visual and quantitative results support the benefit of
geometry-aware sampling. The measurement-consistency step stabilises the range-space component of the reconstruction, while the pullback-metric-guided diffusion update improves the weakly constrained directions where the learned prior is most useful. This combination helps reduce streak artifacts, preserve anatomical boundaries, and produce stable posterior-mean reconstructions across sparse-view, noisy, and limited-angle settings.

\FloatBarrier
\subsection{Empirical Sample-Variability Analysis}
\label{subsec:exp_uncertainty}

A key output of the proposed framework is empirical sample variability
alongside the reconstruction. Figure~\ref{fig:ours_uncertainty} summarises
this behaviour across the five acquisition settings (as in
Section~\ref{subsec:exp_visual_comparison}). For each configuration, we show
the FBP reconstruction, the sample-mean reconstruction, the pixel-wise
sample standard deviation, and the absolute error
$|\bar{x}-x_{\mathrm{gt}}|$. This layout allows us to compare where the
generated reconstructions vary with where the sample mean makes
reconstruction errors.

\begin{figure}[t]
    \centering
    \includegraphics[width=0.85\linewidth]{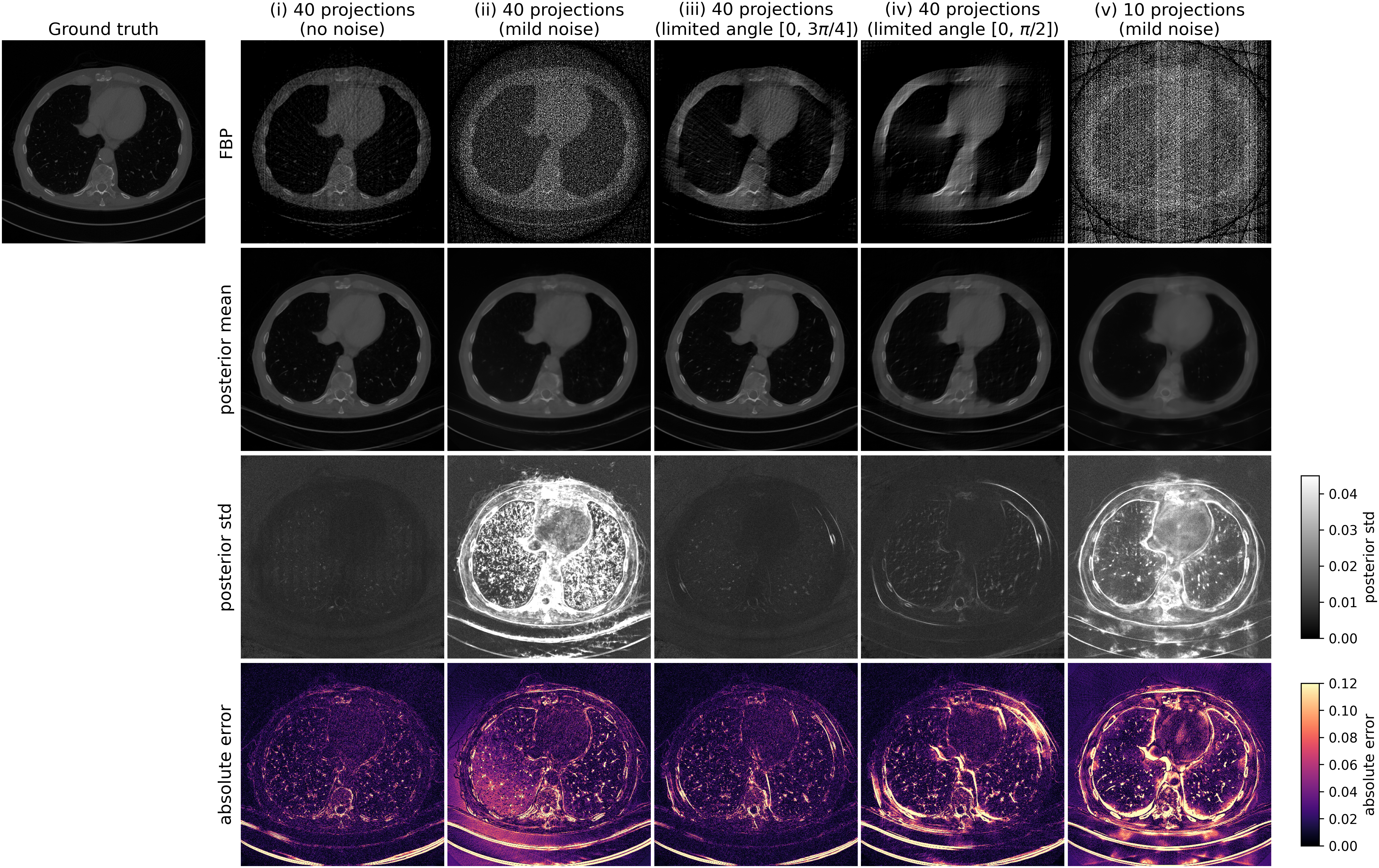}
    \caption{Characterisation of the proposed method across five
    acquisition-degradation settings. From top to bottom: FBP reconstruction
    (illustrating acquisition ill-posedness), sample-mean reconstruction,
    pixel-wise sample standard deviation (darker = lower sample variability),
    and absolute error $|\bar{x} - x_{\mathrm{gt}}|$. The ground-truth column
    is shown for reference.}
    \label{fig:ours_uncertainty}
\end{figure}

In the noise-free sparse-view setting~(i), the sample mean is close to the
ground truth and the pixel-wise sample standard deviation is low across most
of the image. The remaining sample variability is concentrated mainly around
anatomical boundaries and fine structures. The absolute-error map shows a
similar spatial pattern, with small errors in homogeneous regions and larger
errors near interfaces. This suggests that, when the measurements are
reliable, the generated reconstructions vary little in well-constrained
regions and more near ambiguous structures.

Adding Poisson noise in setting~(ii) substantially increases the pixel-wise
sample standard deviation. Unlike the noise-free case, the variability is no
longer confined mainly to edges, but spreads through the lung parenchyma and
soft tissue. This reflects the fact that the measurements themselves are
corrupted, so even geometrically observed directions are less reliable. The
absolute-error map is also more widespread, while the sample mean remains
visually much cleaner than the FBP reconstruction. This indicates that
averaging over stochastic reconstructions suppresses much of the noise, while
their variability still reflects the effect of noisy measurements.

The limited-angle settings show a different variability structure. In
setting~(iii), with angular range $[0,3\pi/4]$, the pixel-wise sample
standard deviation remains relatively low and is mainly visible around
boundaries. In setting~(iv), where the angular range is further reduced to
$[0,\pi/2]$, the variability becomes more pronounced and spatially
structured, especially near the body boundary and along directions affected
by the missing angular range. The corresponding error maps follow the same
trend. In setting~(iii), most errors appear near anatomical boundaries. In
setting~(iv), the errors become stronger and more directional. This suggests
that the sample variability is related to the missing-angle geometry, rather
than becoming uniformly larger across the image.

The most challenging case is setting~(v), which combines only 10 projections
with measurement noise. Here the pixel-wise sample standard deviation is high
and spatially structured, with strong variability around the body boundary,
lung region, and internal anatomical interfaces. The absolute-error map is
also brightest in these regions. Compared with setting~(ii), the variability
is less uniformly diffuse and more strongly shaped by the acquisition
geometry. This is consistent with the large null space induced by extreme
view sparsity. Many image structures are weakly constrained by the
measurements, while the data-consistency correction still promotes agreement
with the measured projection directions.

\begin{table}[t]
\centering
\caption{Mean pixel-wise sample standard deviation across $N=20$ stochastic
reconstructions under five CT acquisition configurations.}
\label{tab:posterior_uncertainty}
\begin{tabular}{lccccc}
\toprule
Method & Config i & Config ii & Config iii & Config iv & Config v \\
\midrule
DDS & 0.0194 & 0.0523 & 0.0209 & 0.0250 & 0.0459 \\
DPS & 0.0245 & 0.0328 & 0.0258 & 0.0418 & 0.0405 \\
PGDM & 0.0228 & 0.0331 & 0.0252 & 0.0355 & 0.0501 \\
ReSample & 0.0117 & 0.0182 & 0.0142 & 0.0215 & 0.0275 \\
Ours & 0.0079 & 0.0260 & 0.0081 & 0.0082 & 0.0205 \\
\bottomrule
\end{tabular}
\end{table}

The values in Table~\ref{tab:posterior_uncertainty} are consistent with these
observations. The proposed method gives the lowest mean pixel-wise sample
standard deviation in the more informative settings: 0.0079 in Config~i,
0.0081 in Config~iii, and 0.0082 in Config~iv. These low values indicate
that the generated reconstructions agree closely when the measurements
provide strong geometric constraints. In the noisy or highly undersampled
settings, the variability increases, reaching 0.0260 in Config~ii and
0.0205 in Config~v. This increase is expected, since measurement noise and
extreme view sparsity both enlarge the set of plausible reconstructions.

The corresponding spatial standard-deviation maps are shown in
Appendix~\ref{app:baseline_uncertainty},
Figure~\ref{fig:baseline_uncertainty}. Compared with the diffusion-based
baselines, the proposed method generally produces lower variability across
reconstructions. DDS, DPS, and PGDM have higher mean pixel-wise standard
deviations in most configurations. ReSample gives the closest competing
values: its variability is higher than that of the proposed method in
Configs~i, iii, iv, and v, but lower in Config~ii
(0.0182 versus 0.0260). Lower variability alone does not necessarily
indicate better uncertainty estimation, since it may also result from
insufficient exploration. In our results, however, the proposed method
combines competitive reconstruction quality with spatial variability that
corresponds to anatomical interfaces, noisy regions, and structures affected
by missing views. This correspondence suggests that the maps capture
measurement-related reconstruction ambiguity rather than only uncontrolled
sampling variation. They remain empirical variability diagnostics rather
than calibrated posterior-uncertainty estimates.

%%%%%%%%%%%%%%%%%%%%%% Geometry-aware testing %%%%%%%%%%%%%%%%%%%%

%\FloatBarrier
\subsection{Geometry-aware analysis}
\label{sub:Geometry-aware analysis}

The previous sections evaluated reconstruction quality using PSNR and SSIM
(Section~\ref{sub:reconstruction_performance}) and characterised empirical
sample variability through standard-deviation maps
(Section~\ref{subsec:exp_uncertainty}). While these results quantify
reconstruction fidelity and show where the generated reconstructions vary,
they do not fully explain how the acquisition geometry shapes the
reconstructions or their variability. To provide this geometric perspective,
we next perform a geometry-aware analysis using regularized range--null
component diagnostics.

First, we validate the regularized range--null component analysis on
representative reconstructions and confirm that the extracted null-like
component is effectively invisible to the measurement operator
(Section~\ref{sub:range_null_decomposition}). Second, we apply the
decomposition to stochastic reconstructions and examine whether their sample
variability is concentrated in approximately measurement-invisible directions,
as predicted by the geometry of the method
(Section~\ref{sub:Geometry-Aware Posterior}). Third, we isolate the role of
pullback-metric preconditioning by comparing the proposed update with a
Euclidean drift baseline, verifying that the geometry-aware update
preferentially moves along weakly constrained directions while maintaining
measurement consistency (Section~\ref{sub:pullback_ablation}). Finally, we
evaluate the effect of signal-dependent precision weighting derived from the
Poisson photon-counting noise model and demonstrate its benefit under
realistic CT noise (Section~\ref{sub:precision_ablation}).

%%%%%%%%%%%%%%%%%%%%%%%%%%%%%%%%%%%%%%%%%%%%%%%%%%%%%%
%%%%%%%%%%%%%%%%%%%%%%%%%%%%%%%%%%%%%%%%%%%%%%%%%%%%%%
\subsubsection{Regularized Range--Null Component Analysis}
\label{sub:range_null_decomposition}

Before analysing variability across stochastic reconstructions, we examine
how the sample-mean reconstruction separates into measurement-sensitive and
weakly constrained components. This analysis is motivated by the exact
orthogonal direct-sum structure in
Proposition~\ref{thm:orthogonal-decomposition}, but avoids explicitly
computing a pseudoinverse.

We define the regularized range-component operator
\begin{equation}
    P_{R,\gamma_{\mathrm{dec}}}
    \;:=\;
    A_\Omega^\top
    \left(
        A_\Omega A_\Omega^\top
        +
        \gamma_{\mathrm{dec}} I_m
    \right)^{-1}
    A_\Omega,
    \qquad
    \gamma_{\mathrm{dec}}>0,
    \label{eq:regularized-range-operator}
\end{equation}
and its complementary null-like component operator
\begin{equation}
    P_{N,\gamma_{\mathrm{dec}}}
    \;:=\;
    I_n-P_{R,\gamma_{\mathrm{dec}}}.
    \label{eq:regularized-null-operator}
\end{equation}
Because the exact range--null decomposition is determined by $A_\Omega$, we
use the unweighted Euclidean operator in this diagnostic. This isolates the
effect of the acquisition geometry from the noise weighting used during
reconstruction.

For a reconstruction $x$, we define
\begin{equation}
    x_{\mathrm{range}}
    \;:=\;
    P_{R,\gamma_{\mathrm{dec}}}x,
    \qquad
    x_{\mathrm{null}}
    \;:=\;
    P_{N,\gamma_{\mathrm{dec}}}x.
    \label{eq:exp_single_decomposition}
\end{equation}
By construction,
\begin{equation}
    x=x_{\mathrm{range}}+x_{\mathrm{null}}.
\end{equation}
Here, $\gamma_{\mathrm{dec}}>0$ is a diagnostic regularization parameter
used only for this post-processing analysis. It stabilizes the
measurement-space solve and controls how strongly weakly observed directions
are assigned to the null-like component. Smaller values produce a closer
approximation to the exact range--null decomposition but lead to a more
ill-conditioned linear system.

The component $x_{\mathrm{range}}$ lies in
$\mathcal R(A_\Omega^\top)$ because it has the form
$A_\Omega^\top q$ for some $q\in\mathbb R^m$. It therefore emphasizes
directions that affect the measured sinogram. The complementary component
$x_{\mathrm{null}}$ emphasizes directions that are weakly constrained or
unobservable under the acquisition geometry.

For finite $\gamma_{\mathrm{dec}}>0$, these operators are regularized
component operators rather than exact orthogonal projectors. Consequently,
$x_{\mathrm{null}}$ need not lie exactly in $\mathcal N(A_\Omega)$, and the
two components need not be orthogonal. As
$\gamma_{\mathrm{dec}}\to0$, the operators approach the orthogonal
projectors onto $\mathcal R(A_\Omega^\top)$ and
$\mathcal N(A_\Omega)$, respectively.

To make the regularization independent of the numerical scaling of
$A_\Omega$, we set
\begin{equation}
    \gamma_{\mathrm{dec}}
    \;=\;
    \eta\,
    \lambda_{\max}
    \left(
        A_\Omega A_\Omega^\top
    \right),
    \qquad
    \eta=10^{-4}.
    \label{eq:decomposition-gamma-scaling}
\end{equation}
The largest eigenvalue is estimated using 20 matrix-free power iterations.
For the 40-projection geometry used here, this gives
$\lambda_{\max}(A_\Omega A_\Omega^\top)\approx1.9606\times10^4$ and
$\gamma_{\mathrm{dec}}\approx1.9606$. The linear systems are solved by
conjugate gradient with relative tolerance $10^{-6}$.

We quantify the relative squared norms of the two components and their cross
term by
\begin{equation}
    r_{\mathrm{range}}
    \;=\;
    \frac{\|x_{\mathrm{range}}\|_2^2}{\|x\|_2^2},
    \qquad
    r_{\mathrm{null}}
    \;=\;
    \frac{\|x_{\mathrm{null}}\|_2^2}{\|x\|_2^2},
    \qquad
    c_{\mathrm{cross}}
    \;=\;
    \frac{
        2\langle x_{\mathrm{range}},x_{\mathrm{null}}\rangle
    }{
        \|x\|_2^2
    }.
    \label{eq:exp_single_energy}
\end{equation}
These quantities satisfy
\begin{equation}
    r_{\mathrm{range}}
    +
    r_{\mathrm{null}}
    +
    c_{\mathrm{cross}}
    =1.
\end{equation}
Because $c_{\mathrm{cross}}$ is generally nonzero, the first two quantities
are relative component-magnitude diagnostics rather than fractions of an
orthogonal energy decomposition.

We further measure the relative sinogram response of the null-like component:
\begin{equation}
    \rho_{\mathrm{null}}
    \;=\;
    \frac{
        \|A_\Omega x_{\mathrm{null}}\|_2
    }{
        \|A_\Omega x\|_2
    }.
    \label{eq:exp_single_null_visibility}
\end{equation}
A small value of $\rho_{\mathrm{null}}$ indicates that the extracted null-like
component only weakly affects the predicted sinogram.

\begin{figure}[htbp]
    \centering
    \includegraphics[width=0.7\linewidth]
    {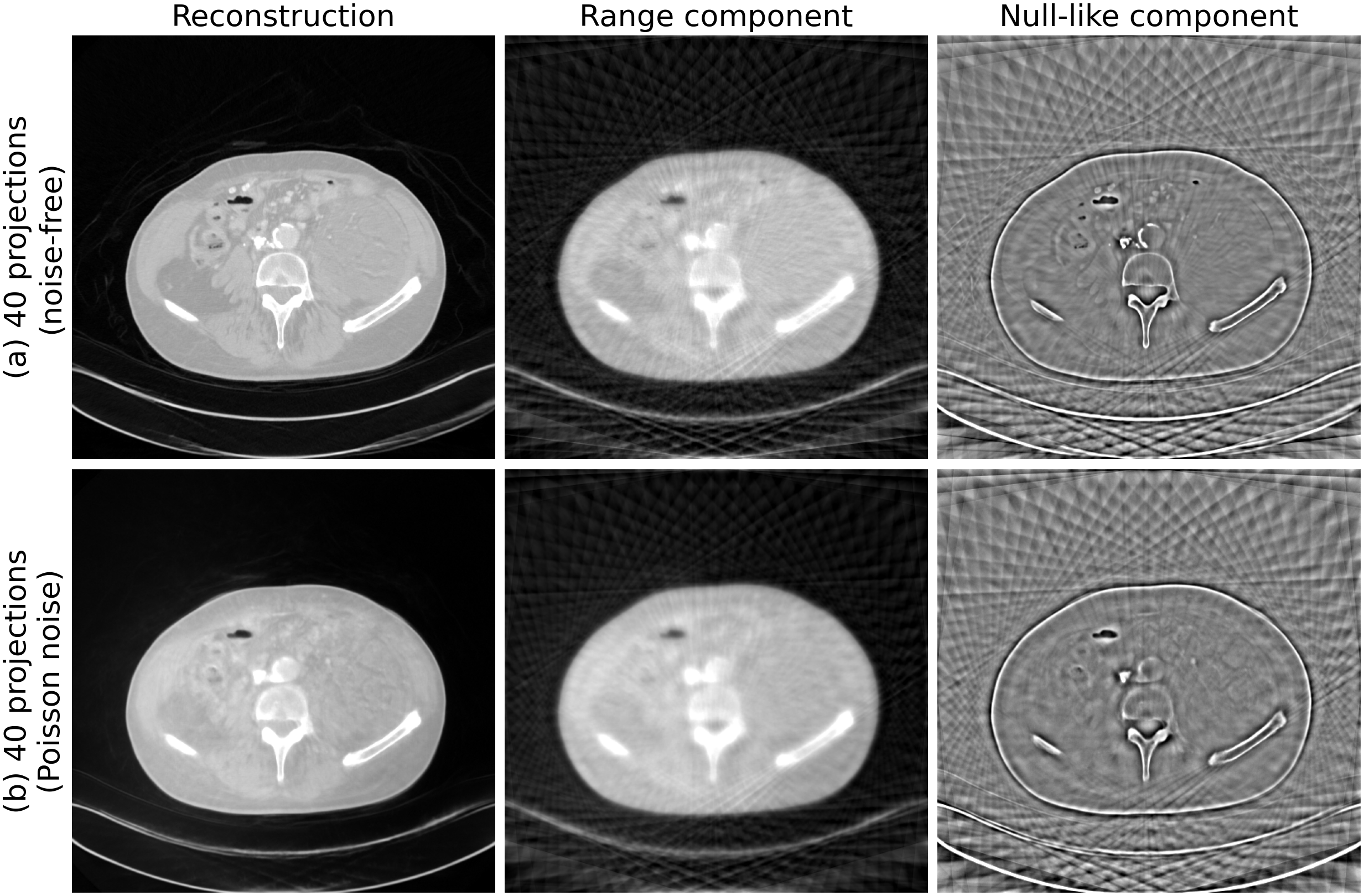}
    \caption{
    Regularized range--null component analysis of sample-mean
    reconstructions for a representative test image under
    (a) 40-projection noise-free and
    (b) 40-projection Poisson-noise conditions.
    Each row shows the sample mean obtained from $N=20$ stochastic
    reconstructions, its regularized range component, and its complementary
    null-like component. The component images are independently
    contrast-normalized for visualization; their displayed intensities
    therefore do not represent their relative magnitudes.
    }
    \label{fig:single_range_null_decomposition}
\end{figure}

\begin{table}[t]
\centering
\caption{
Regularized range--null diagnostics for posterior-mean reconstructions. The quantities $r_{\mathrm{range}}$ and $r_{\mathrm{null}}$ are relative squared norms, $c_{\mathrm{cross}}$ accounts for their nonorthogonality, and $\rho_{\mathrm{null}}$ measures the relative sinogram response of the null-like component.
}
\label{tab:single_range_null}
\begin{tabular}{lcccc}
\toprule
Setting
& $r_{\mathrm{range}}$
& $r_{\mathrm{null}}$
& $c_{\mathrm{cross}}$
& $\rho_{\mathrm{null}}$ \\
\midrule
40 projections (noise-free)
& 0.99631
& 0.00305
& $6.43\times10^{-4}$
& $1.93\times10^{-4}$ \\
40 projections (Poisson noise)
& 0.99729
& 0.00217
& $5.43\times10^{-4}$
& $1.77\times10^{-4}$ \\
\bottomrule
\end{tabular}
\end{table}

Figure~\ref{fig:single_range_null_decomposition} shows that the regularized range component retains most of the overall image magnitude and low-frequency anatomical content, whereas the null-like component emphasizes edge- and streak-like structures that are weakly constrained by the acquisition geometry. Quantitatively, the range-component squared norm is $99.631\%$ of the reconstruction squared norm in the noise-free case and
$99.729\%$ in the Poisson-noise case. The corresponding null-like-component ratios are $0.305\%$ and $0.217\%$, respectively.

In both settings, $\rho_{\mathrm{null}}$ is below
$2\times10^{-4}$, showing that the extracted null-like component has only a small effect on the sinogram. These results support its interpretation as an approximately measurement-invisible component under
the 40-projection geometry. Because this experiment analyses posterior means, the component ratios should not be interpreted as measures of posterior uncertainty; variability across stochastic reconstruction samples is examined
separately.
%%%%%%%%%%%%%%%%%%%%%%%%%%%%%%%%%%%%%%%%%%%%%%%%%%%%%%%%%%%
\FloatBarrier
\subsubsection{Geometry-Aware Posterior Variability Analysis}
\label{sub:Geometry-Aware Posterior}

We next analyse the geometry of the posterior variability produced by the proposed reconstruction method. For each acquisition setting, we first generate $K$ posterior reconstructions
$\{x^{(k)}\}_{k=1}^{K}$ of the same test image using different random seeds and the same measured sinogram. The posterior mean is then computed as $\bar{x}=\frac{1}{K}\sum_{k=1}^{K}x^{(k)}$. All diagnostics in this subsection are computed from these reconstructed stochastic reconstruction samples. Thus, the analysis does not study synthetic perturbations,
but the actual variability of the reconstructions produced by the sampler.

To investigate whether this posterior variability follows the measurement geometry, we perform two complementary analyses. First, we compare the variability of the posterior reconstructions in image space with their variability after applying the CT forward operator. Given the stochastic reconstruction samples $\{x^{(k)}\}_{k=1}^{K}$ and sample mean $\bar{x}$, we compute
\begin{equation}
\begin{aligned}
\sigma_x
&=
\sqrt{
\frac{1}{K}
\sum_{k=1}^{K}
\left(
x^{(k)}-\bar{x}
\right)^2
},
\qquad
\sigma_y
&=
\sqrt{
\frac{1}{K}
\sum_{k=1}^{K}
\left(
A_{\Omega}x^{(k)}
-
A_{\Omega}\bar{x}
\right)^2
}.
\end{aligned}
\label{eq:exp_std_maps}
\end{equation}
Here, $\sigma_x$ measures pixel-wise variability across reconstructed images, whereas $\sigma_y$ measures how much this variability remains visible in the corresponding sinograms.

To compare variability in the two domains, we normalise by the magnitude of the corresponding posterior mean:
%To compare the two domains, we normalize by the magnitude of the corresponding posterior mean,
\begin{equation}
\begin{aligned}
R_{\mathrm{img}}
&=
\frac{\|\sigma_x\|_2}
{\|\bar{x}\|_2},
\qquad
R_{\mathrm{sino}}
&=
\frac{\|\sigma_y\|_2}
{\|A_{\Omega}\bar{x}\|_2},
\qquad
R_{\mathrm{sino/img}}
&=
\frac{R_{\mathrm{sino}}}
{R_{\mathrm{img}}}.
\end{aligned}
\label{eq:exp_relative_variability_ratios}
\end{equation}
A small value of $R_{\mathrm{sino/img}}$ means that differences between
posterior reconstructions are strongly reduced after applying the forward operator. This indicates that much of the image-space variability lies in directions that are weakly visible, or approximately invisible, to the measurements.

Second, we decompose the deviation of each posterior reconstruction from the posterior mean,
\begin{equation}
v^{(k)}=x^{(k)}-\bar{x}
\label{eq:exp_sample_deviation}
\end{equation}
into a measurement-visible component and an approximately
measurement-invisible component:
\begin{equation}
v_{\mathrm{range}}^{(k)}
\;=\;
P_{R,\gamma_{\mathrm{dec}}}v^{(k)},
\qquad
v_{\mathrm{null}}^{(k)}
\;=\;
P_{N,\gamma_{\mathrm{dec}}}v^{(k)}.
\label{eq:exp_range_projection}
\end{equation}
The relative contributions of the two components are quantified by
\begin{equation}
\begin{aligned}
r_{\mathrm{range}}^{(k)}
&\;=\;
\frac{
\|v_{\mathrm{range}}^{(k)}\|_2^2
}{
\|v^{(k)}\|_2^2
},
\qquad
r_{\mathrm{null}}^{(k)}
&\;=\;
\frac{
\|v_{\mathrm{null}}^{(k)}\|_2^2
}{
\|v^{(k)}\|_2^2
}.
\end{aligned}
\label{eq:range_null_energy_ratios}
\end{equation}
The quantities $r_{\mathrm{range}}^{(k)}$ and $r_{\mathrm{null}}^{(k)}$
measure the relative energy of the measurement-visible and approximately
measurement-invisible components of each posterior deviation. Since the
projection is computed with a finite regularisation parameter $\gamma_{\mathrm{dec}}$, these values should be interpreted as diagnostic energy ratios rather than exact orthogonal fractions.

While $R_{\mathrm{sino/img}}$ measures the overall suppression of posterior variability after applying the forward operator, the quantities $r_{\mathrm{range}}$ and $r_{\mathrm{null}}$ describe how the image-space posterior deviations are distributed between measurement-visible and approximately measurement-invisible components. Together, these diagnostics show both how strongly the stochastic reconstruction samples remain constrained by the measurements and where the remaining variability is concentrated.

\begin{table}[t]
\centering
\caption{
Geometry-aware posterior variability diagnostics for three representative acquisition settings. A smaller $R_{\mathrm{sino/img}}$ indicates stronger suppression of posterior variability after applying the forward operator. The mean range--null diagnostic energy ratios, $\bar r_{\mathrm{range}}$ and $\bar r_{\mathrm{null}}$, describe how posterior deviations from the posterior mean are distributed between measurement-visible and approximately measurement-invisible components. These quantities are averaged over stochastic reconstruction samples.
}
\label{tab:geometry_variability_results}
\begin{tabular}{lcccccc}
\toprule
Setting & Samples & $R_{\mathrm{img}}$ & $R_{\mathrm{sino}}$ & $R_{\mathrm{sino/img}}$ & 
$\bar r_{\mathrm{range}}$ &
$\bar r_{\mathrm{null}}$ \\
\midrule
40 projections (noise-free) & 50 & 0.0094 & 0.0004 & 0.0399 & 0.184 & 1.046 \\
40 projections (Poisson noise) & 50 & 0.0386 & 0.0106 & 0.2737 & 0.422 & 0.618 \\
20 projections (Poisson noise) & 50 & 0.0437 & 0.0133 & 0.3043 & 0.313 & 0.697 \\
\bottomrule
\end{tabular}
\end{table}

Table~\ref{tab:geometry_variability_results} summarises the geometry-aware posterior variability diagnostics from two complementary perspectives. First, the ratios $R_{\mathrm{img}}$, $R_{\mathrm{sino}}$, and $R_{\mathrm{sino/img}}$ compare the variability of posterior reconstructions in image space with the variability that remains after applying the CT forward operator. Second, the range--null diagnostic ratios $\bar r_{\mathrm{range}}$ and $\bar r_{\mathrm{null}}$ describe whether the image-space posterior deviations are mainly measurement-visible or approximately measurement-invisible.

From the measurement-space perspective, all three settings show much smaller variability in sinogram space than in image space. This indicates that although the stochastic reconstruction samples differ visibly as reconstructed images, these differences are strongly reduced after applying the measurement operator. The strongest suppression occurs in the 40-projection noise-free setting, where $R_{\mathrm{sino/img}}=0.0399$, meaning that the relative variability in
measurement space is only about $4\%$ of that in image space. With Poisson noise, this ratio increases to $0.2737$ for 40 projections. With both Poisson noise and stronger undersampling, it further increases to $0.3043$ for 20 projections. This trend is consistent with the measurements becoming less informative: noise reduces the reliability of the measured sinogram, and fewer projection angles leave more image directions weakly constrained.

The range--null diagnostics provide a complementary view of the same behaviour. In the 40-projection noise-free setting, the posterior deviations are dominated by the approximately measurement-invisible component ($\bar r_{\mathrm{null}}=1.046$ versus $\bar r_{\mathrm{range}}=0.184$; their sum exceeds one because the cross-term $\bar c_{\mathrm{cross}}\approx-0.23$ is nonzero for finite $\gamma_{\mathrm{dec}}$). This indicates that, once measurement
consistency is enforced, the remaining posterior variability lies mainly in directions that have little effect on the measured sinogram. Under Poisson noise, the measurement-visible contribution increases
($\bar r_{\mathrm{range}}=0.422$), reflecting the fact that uncertainty in the measurements can also produce variability in directions visible to the forward operator. Nevertheless, the approximately null component remains larger ($\bar r_{\mathrm{null}}=0.618$). In the 20-projection noisy case, the null component again dominates ($\bar r_{\mathrm{null}}=0.697$), consistent with the larger null-space induced by stronger angular undersampling.

Overall, the two diagnostics give consistent evidence that the proposed sampler produces geometry-aware posterior variability. The stochastic reconstruction samples differ primarily in image directions that are weakly constrained by the acquisition, while remaining comparatively consistent in the measurement domain.

%%%%%%%%%%%%%%%%%%%% Effect of Pullback-metric %%%%%%%%%%%%%%%%%%%%%%%%
\FloatBarrier
\subsubsection{Effect of Pullback-Metric Preconditioning}
\label{sub:pullback_ablation}

To isolate the effect of pullback-metric preconditioning, we compare the proposed deterministic update with an identical Euclidean update. Within each acquisition geometry, the two variants use the same pretrained prior, data-fidelity term, initialization, reconstruction parameters, and regularized data-consistency schedule. Langevin noise and the additional approximately null-space-restricted perturbation are disabled in both cases. We consider two noise-free configurations with the same number of measurements: 40-view full-angle CT over $[0,\pi]$ and 40-view limited-angle CT over $[0,3\pi/4]$. The endpoint convention matches the DM4CT benchmark geometry.

At reconstruction iteration $k$, let $s^{(k)}$ denote the DDPM score estimate evaluated at the corresponding noised image state. We form the combined guidance direction
\begin{equation}
r^{(k)}
\;=\;
\lambda_k s^{(k)}
-\mu_k A_{\Omega}^{\top}\Sigma^{-1}
\left(A_{\Omega}x^{(k)}-y\right).
\label{eq:ablation_rhs}
\end{equation}
This is the practical annealed guidance direction used by the implementation, rather than the exact score of the clean-image posterior. The Euclidean variant applies
\begin{equation}
u_{\mathrm{Euc}}^{(k)}
\;=\;
r^{(k)},
\label{eq:ablation_euclidean_drift}
\end{equation}
whereas the pullback-preconditioned variant solves
\begin{equation}
\left(
A_{\Omega}^{\top}\Sigma^{-1}A_{\Omega}
+\gamma I_n
\right)
u_{\mathrm{geo}}^{(k)}
\;=\;
r^{(k)}.
\label{eq:ablation_pullback_drift}
\end{equation}
Thus, within each geometry, the variants differ only in whether the combined guidance is preconditioned by $G_{\gamma}=A_{\Omega}^{\top}\Sigma^{-1}A_{\Omega}+\gamma I_n$.

\begin{figure*}[t]
    \centering
    \includegraphics[width=0.7\textwidth]{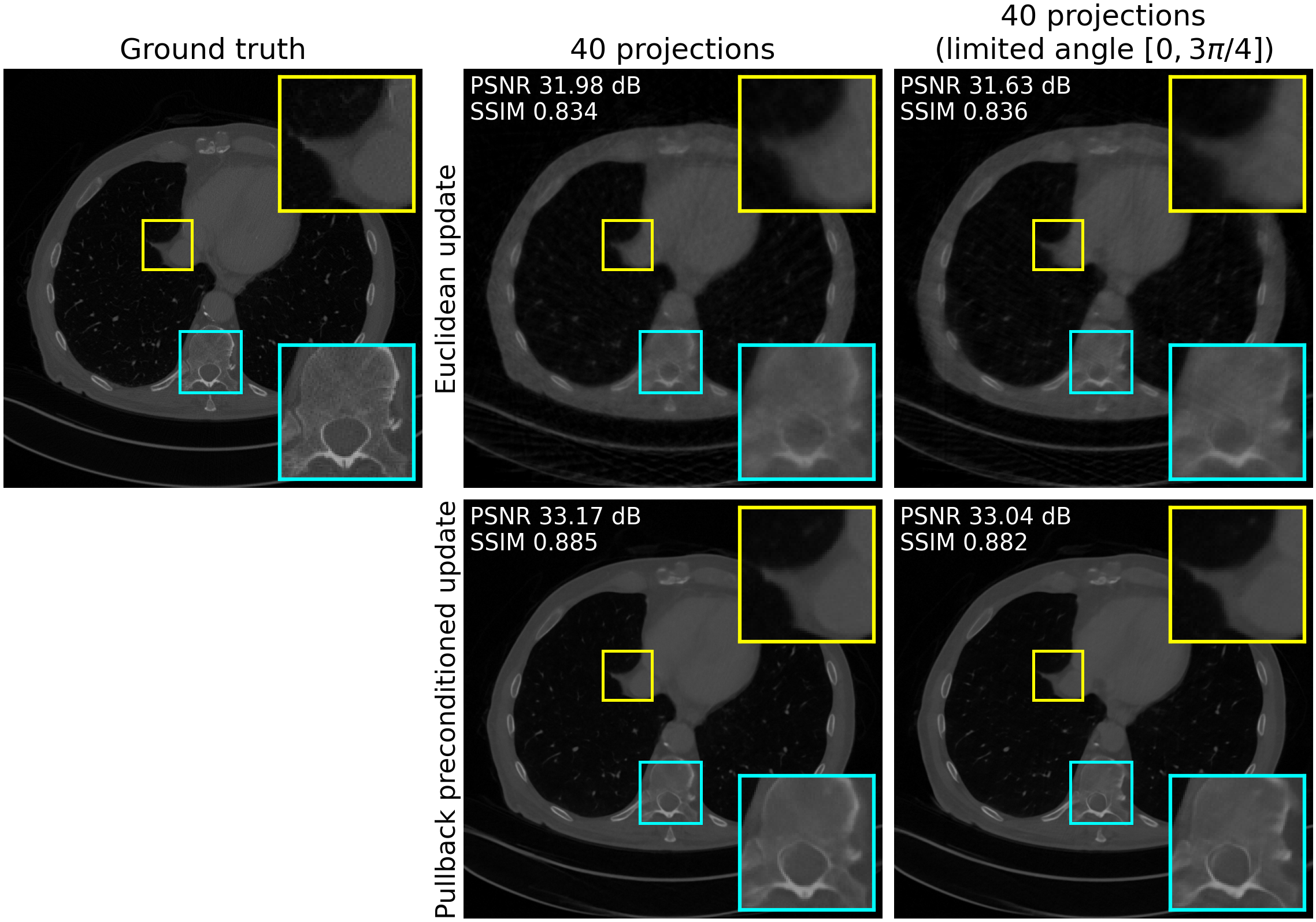}
    \caption{
    Reconstruction comparison between the Euclidean and
    pullback-preconditioned variants under two noise-free acquisition
    geometries. The columns show the ground truth, 40-view full-angle CT over $[0,\pi]$, and 40-view limited-angle CT over $[0,3\pi/4]$. The top and bottom reconstruction rows correspond to the Euclidean and pullback-preconditioned variants, respectively. Yellow and cyan boxes mark the lung--soft-tissue interface and vertebral region, with enlarged views  shown in the corresponding insets. PSNR and SSIM are displayed for each reconstruction. All images use the same display window.
    }
    \label{fig:pullback_recon_panel}
\end{figure*}

Figure~\ref{fig:pullback_recon_panel} shows that pullback preconditioning improves reconstruction quality in both geometries. For full-angle CT, PSNR increases from $31.98$ to $33.17$~dB and SSIM increases from $0.8339$ to $0.8852$. For limited-angle CT, PSNR increases from $31.63$ to $33.04$~dB and SSIM increases from $0.8359$ to $0.8816$. In the yellow region, the pullback-preconditioned reconstructions preserve the lung--soft-tissue interface more clearly. In the cyan region, they recover a sharper and more continuous vertebral boundary. Compared with the Euclidean variant, the
pullback-preconditioned results also exhibit less local blurring and fewer residual streak artifacts. These visual differences are consistent with the improvements in PSNR and SSIM.

We quantify agreement with the measured sinogram using
\begin{equation}
\mathrm{RelErr}_{y}
\;=\;
\frac{\|A_{\Omega}x-y\|_2}{\|y\|_2}.
\label{eq:relative_data_error}
\end{equation}
To examine the geometry of the applied drift, let
$\Delta x_{\mathrm{det}}^{(k)}$ denote the deterministic displacement before Langevin noise, null-space perturbation, and the subsequent data-consistency correction. We define
\begin{equation}
\rho_{\Delta x}^{(k)}
\;=\;
\frac{
\|A_{\Omega}\Delta x_{\mathrm{det}}^{(k)}\|_2
}{
\|\Delta x_{\mathrm{det}}^{(k)}\|_2
},
\label{eq:update_measurement_response}
\end{equation}
and report
\begin{equation}
\bar{\rho}_{\Delta x}
\;=\;
\frac{1}{T}
\sum_{k=0}^{T-1}\rho_{\Delta x}^{(k)}.
\label{eq:mean_update_measurement_response}
\end{equation}
This quantity measures the sinogram change per unit image-space displacement. A smaller value indicates a displacement with lower measurement sensitivity.

\begin{table}[t]
\centering
\small
\renewcommand{\arraystretch}{1.18}
\caption{
Effect of pullback-metric preconditioning under two noise-free 40-view
geometries. Within each geometry, Langevin noise and the additional null-space perturbation are disabled, and both variants share all settings except the drift preconditioner. Bold indicates the better reconstruction metric or the lower measurement error and update response within each geometry.
}
\label{tab:pullback_metric_ablation}
\begin{tabular}{llcccc}
\toprule
Geometry & Method & PSNR & SSIM & RelErr$_y$ & $\bar{\rho}_{\Delta x}$ \\
\midrule
\multirow{2}{*}{Full angle}
& Euclidean & 31.98 & 0.8339 & $\mathbf{1.13\times10^{-4}}$ & $1.04\times10^{2}$ \\
& Pullback  & \textbf{33.17} & \textbf{0.8852} & $2.24\times10^{-4}$ & $\mathbf{3.96\times10^{-1}}$ \\
\midrule
\multirow{2}{*}{Limited angle}
& Euclidean & 31.63 & 0.8359 & $\mathbf{1.21\times10^{-4}}$ & $1.10\times10^{2}$ \\
& Pullback  & \textbf{33.04} & \textbf{0.8816} & $2.08\times10^{-4}$ & $\mathbf{2.89\times10^{-1}}$ \\
\bottomrule
\end{tabular}
\end{table}

As shown in Table~\ref{tab:pullback_metric_ablation}, pullback
preconditioning improves PSNR by $1.19$~dB in the full-angle case and
$1.41$~dB in the limited-angle case. The Euclidean variant gives the smaller relative data error, but all four errors remain below $2.3\times10^{-4}$. Although the pullback variant does not attain the lowest measurement residual, it maintains strong measurement agreement, with $\mathrm{RelErr}_y<2.3\times10^{-4}$, while substantially improving PSNR and SSIM.

The update-response diagnostic changes much more strongly. In the full-angle case, $\bar{\rho}_{\Delta x}$ decreases from $1.04\times10^2$ to
$3.96\times10^{-1}$, a reduction by a factor of approximately $263$. In the limited-angle case, it decreases from $1.10\times10^2$ to
$2.89\times10^{-1}$, a reduction by a factor of approximately $381$. These results show that the pullback metric substantially reduces the measurement response of the applied deterministic displacement. Together with the range--null diagnostic below, this is consistent with shifting the update toward more weakly constrained image directions.

\begin{figure}[t]
    \centering
    \begin{subfigure}{0.496\linewidth}
        \centering
        \includegraphics[width=0.85\linewidth]{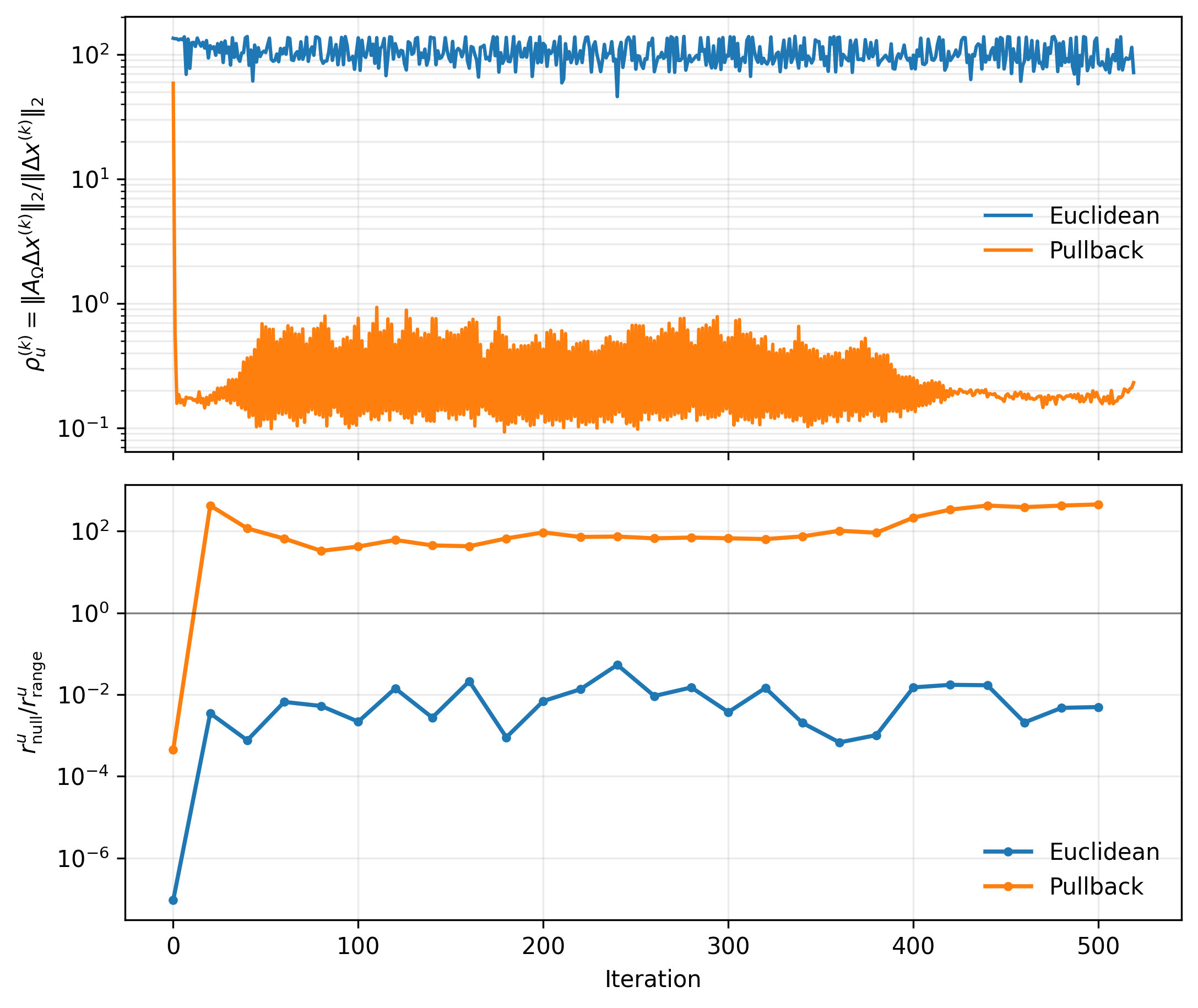}
        \caption{40-view full angle}
        \label{fig:pullback_40views}
    \end{subfigure}
    \hfill
    \begin{subfigure}{0.496\linewidth}
        \centering
        \includegraphics[width=0.85\linewidth]{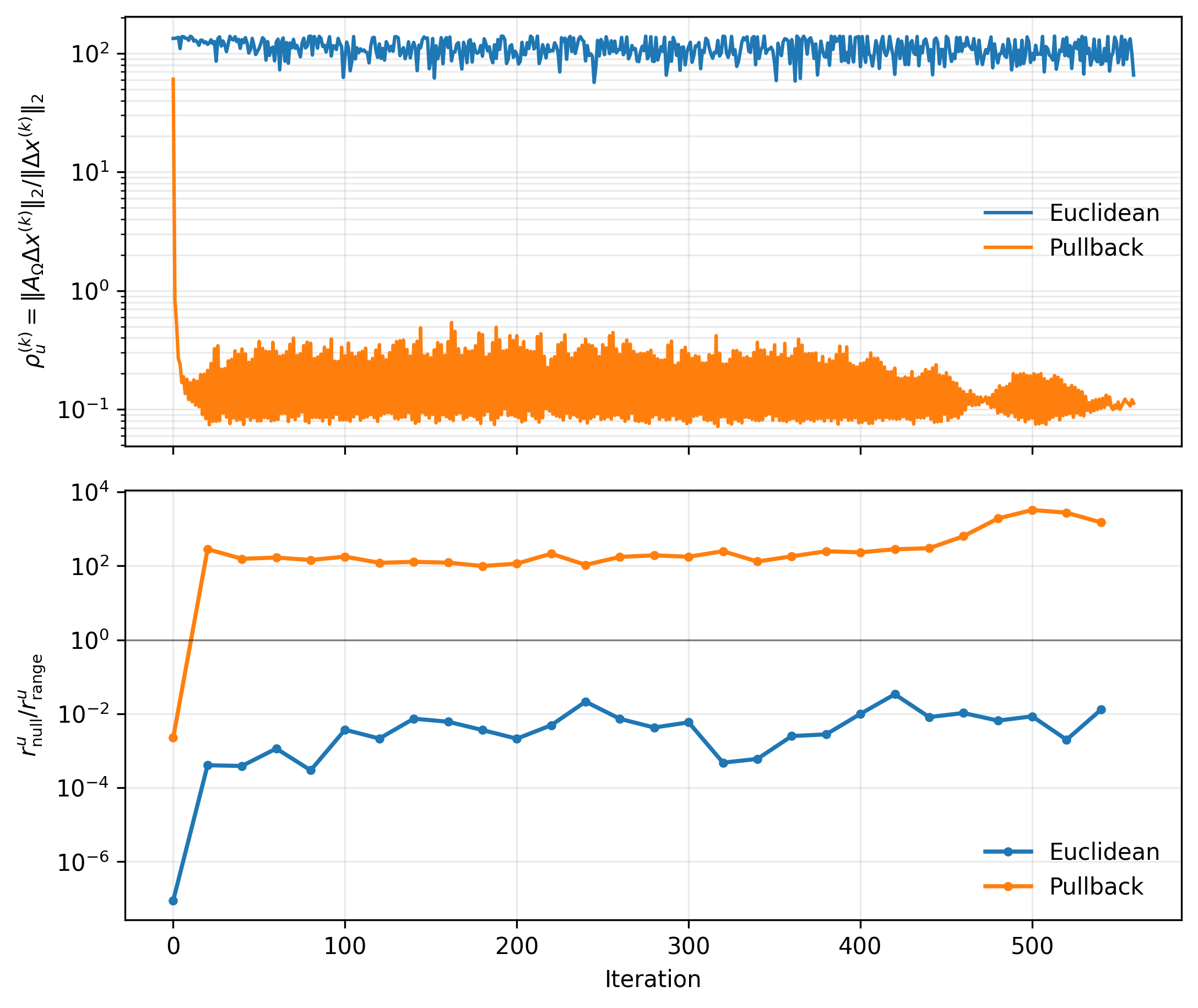}
        \caption{40-view limited angle}
        \label{fig:pullback_limited_40views}
    \end{subfigure}
    \caption{
    Effect of pullback-metric preconditioning on the applied deterministic
    displacement. The top panels show $\rho_{\Delta x}^{(k)}$, and the bottom
    panels show the regularized null-like-to-range component ratio.
    }
    \label{fig:pullback_metric_analysis}
\end{figure}

For the component diagnostic, define
\begin{equation}
B\;:=\;A_{\Omega}^{\top}\Sigma^{-1}A_{\Omega}.
\end{equation}
For each sampled displacement $\Delta x_{\mathrm{det}}^{(k)}$, we compute
\begin{equation}
\Delta x_{R}^{(k)}
\;=\;
\left(B+\gamma_{\mathrm{dec}}I_n\right)^{-1}
B\Delta x_{\mathrm{det}}^{(k)},
\qquad
\Delta x_{N}^{(k)}
\;=\;
\Delta x_{\mathrm{det}}^{(k)}-\Delta x_{R}^{(k)}.
\end{equation}
Here $\gamma_{\mathrm{dec}}>0$ regularizes the matrix-free CG solve. For finite $\gamma_{\mathrm{dec}}$, $\Delta x_N^{(k)}$ is a null-like component rather than an exact element of $\mathcal N(A_{\Omega})$, and the two components need not be orthogonal. We therefore use the squared component-magnitude ratio
\begin{equation}
\frac{r_{N}^{(k)}}{r_{R}^{(k)}}
\;=\;
\frac{\|\Delta x_N^{(k)}\|_2^2}
{\|\Delta x_R^{(k)}\|_2^2}
\end{equation}
as a diagnostic rather than as an exact partition of image energy.

Figure~\ref{fig:pullback_metric_analysis} shows the evolution of the update response and the regularized component ratio. The Euclidean displacement has a substantially larger measurement response, whereas pullback preconditioning keeps this response much smaller over the reconstruction trajectory. The component ratio provides a complementary interpretation: a larger null-like-to-range ratio indicates that relatively more of the applied displacement lies along weakly constrained or unobservable directions. 

Figure~\ref{fig:pullback_metric_analysis} shows that pullback preconditioning changes the geometry of the applied deterministic displacement throughout reconstruction. In both acquisition settings, the Euclidean displacement remains strongly measurement-sensitive, with $\rho_{\Delta x}^{(k)}$ generally of order $10^2$. After the initial iteration, pullback preconditioning reduces this response to approximately $10^{-1}$–$10^0$. The regularized component analysis gives a complementary result. The null-like-to-range ratio remains below one for the Euclidean variant, whereas it rises well above one for the pullback variant and remains large over most of the trajectory. These observations indicate that pullback preconditioning suppresses displacement components that strongly affect the measured sinogram while retaining relatively more movement along weakly constrained and unobservable directions.

Taken together, the reconstruction metrics, measurement residuals, and update diagnostics show that pullback preconditioning changes the deterministic update in a manner consistent with the measurement geometry. It reduces movement that strongly affects the measured sinogram and permits relatively larger movement along weakly constrained directions, where the learned prior can contribute anatomical information. In these experiments, this change is
associated with improved PSNR and SSIM while the final measurement residual remains small.

%%%%%%%%%%%%%%%%%%%%%%%%%%% Noise isolaltion %%%%%%%%%%%%%%%%%%%%%%%
\subsubsection{Effect of Signal-Dependent Precision Weighting}
\label{sub:precision_ablation}

We finally examine the effect of the signal-dependent precision weighting used in the proposed method. In the Poisson photon-counting model of CT, different sinogram entries have different noise levels: rays with lower photon counts are less reliable than rays with higher photon counts. Following the Gaussian approximation described in Section~\ref{sec:algorithm} and Appendix~\ref{app:gaussian-approx}, we model this using a signal-dependent covariance
\begin{equation}
\widehat{\Sigma}^{-1}
\;=\;\operatorname{diag}(w_1,\ldots,w_m),
\end{equation}
where $w_i > 0$ is a stabilized approximation of the estimated photon
count $\widehat{\lambda}_i = I_0 e^{-y_i}$ for the $i$-th sinogram entry;
the precise construction is given in Appendix~\ref{app:gaussian-approx}.

This signal-dependent covariance enters the proposed method in three places. First, it appears in the weighted data-fidelity term
\begin{equation}
\Phi(x)
\;=\;\frac{1}{2}
\left\|
\widehat{\Sigma}^{-1/2}
(A_{\Omega}x-y)
\right\|_2^2 .
\label{eq:precision_weighted_data_fidelity}
\end{equation}
Second, it appears in the pullback metric
\begin{equation}
G_{\gamma}
\;=\;A_{\Omega}^{\top}
\widehat{\Sigma}^{-1}
A_{\Omega}
+\gamma I .
\label{eq:precision_pullback_metric}
\end{equation}
Third, the same metric defines the geometry-aware Langevin perturbation,
\begin{equation}
\nu_G^{(k)}
\sim
\mathcal N\!\left(
0,2\alpha_kG_\gamma^{-1}
\right).
\label{eq:precision_pullback_noise}
\end{equation}
In practice, this noise sampler is implemented in a matrix-free way: $G_{\gamma}$ is never formed explicitly, and only matrix-vector products
\begin{equation}
G_{\gamma}v
\;=\;A_{\Omega}^{\top}
\widehat{\Sigma}^{-1}
A_{\Omega}v
+\gamma v
\label{eq:precision_pullback_matvec}
\end{equation}
are required.

To isolate the contribution of this signal-dependent weighting, we compare the full method against an ablation in which $\widehat{\Sigma}$ is replaced by the identity matrix. This unit-precision variant treats all sinogram bins as equally reliable.

Both variants use the same pretrained diffusion prior, pullback-metric preconditioning, projection step, initialization, and hyperparameters. The only difference is whether the method uses the physically motivated signal-dependent covariance $\widehat{\Sigma}$ or the unit covariance $I$. We evaluate the two variants under three 40-projection settings: no noise, Gaussian noise, and Poisson noise. The Poisson setting is the case where signal-dependent weighting is expected to matter most.

\begin{figure}[t]
    \centering
    \includegraphics[width=0.85\linewidth]{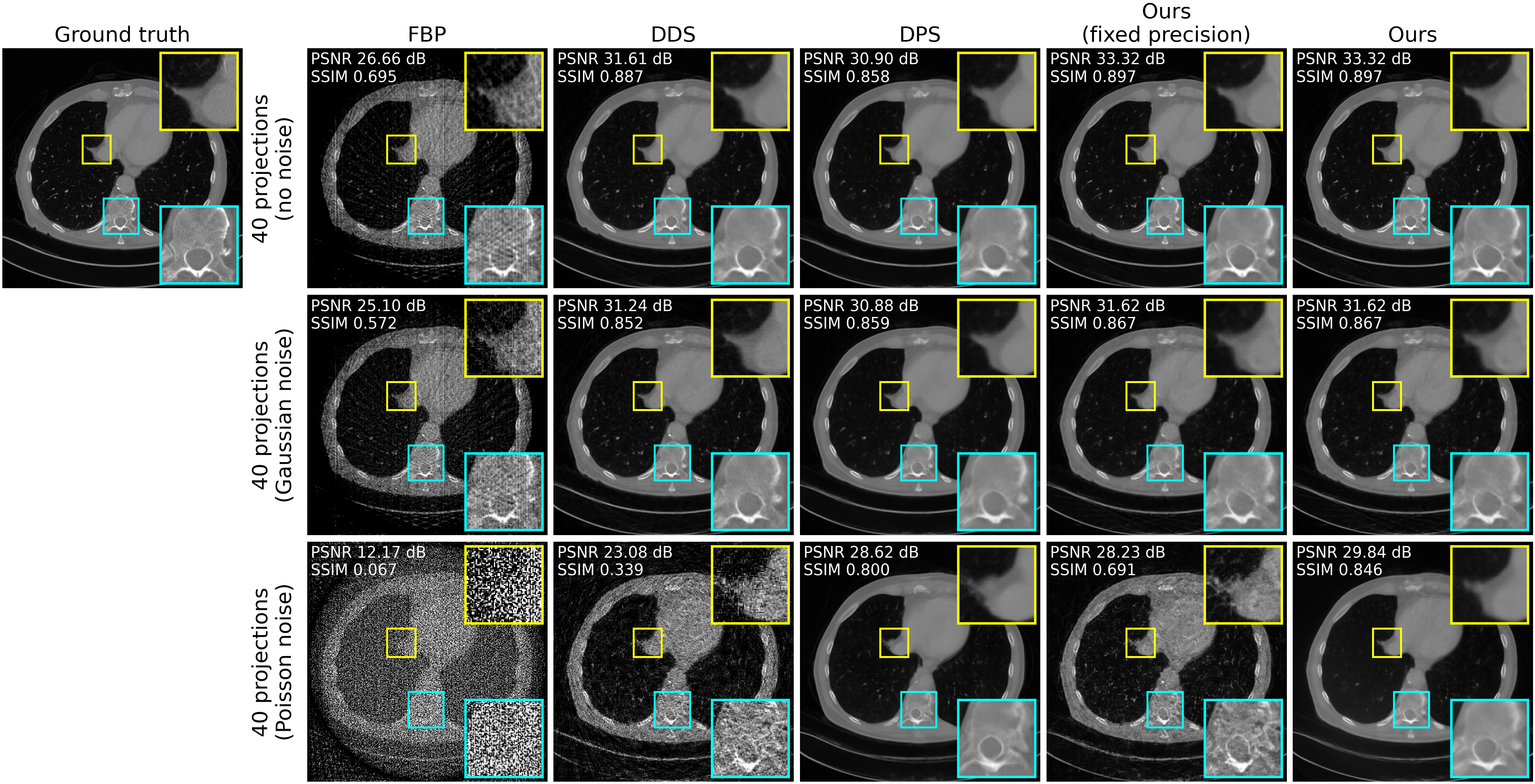}
    \caption{Effect of signal-dependent precision weighting under three noise conditions (40 projections). Columns show FBP, DDS, DPS, Ours with unit precision ($\Sigma = I$), and Ours with effective measurement-dependent precision
$\widehat{\Sigma}^{-1}=\operatorname{diag}(w_i)$. The two variants of our method are identical under noise-free and Gaussian noise settings, but the signal-dependent precision provides a substantial improvement under Poisson noise, particularly in high-attenuation regions (cyan box, vertebral body).}
    \label{fig:precision_ablation}
\end{figure}

\iffalse
\begin{table}[t]
\centering
\caption{Effect of signal-dependent precision weighting (40 projections, 3 noise conditions). PSNR~(dB) / SSIM reported. The two variants are identical under noiseless and Gaussian noise settings; the benefit of signal-dependent precision appears specifically under Poisson noise.}
\label{tab:precision_ablation}
\begin{tabular}{lccc}
\toprule
\textbf{Method} & \textbf{No noise} & \textbf{Gaussian noise} & \textbf{Poisson noise} \\
\midrule
FBP                      & 26.66 / 0.695 & 25.10 / 0.572 & 12.17 / 0.067 \\
DDS~\cite{Chung2024DDS}  & 31.61 / 0.887 & 31.24 / 0.852 & 23.08 / 0.339 \\
DPS~\cite{Chung2023DPS}  & 30.90 / 0.858 & 30.88 / 0.859 & 28.62 / 0.800 \\
\midrule
Ours (unit precision, $\Sigma=I$)    & 33.32 / 0.897 & 31.62 / 0.867 & 28.23 / 0.691 \\
Ours (signal-dependent $\widehat{\Sigma}$) & \textbf{33.32 / 0.897} & \textbf{31.62 / 0.867} & \textbf{29.84 / 0.846} \\
\bottomrule
\end{tabular}
\end{table}
\fi

Figure~\ref{fig:precision_ablation} shows that signal-dependent precision has little effect when the measurement noise is absent or spatially uniform. In the noise-free setting, the two variants of our method are identical, giving 33.32~dB/0.897. The same is true under Gaussian noise, where both variants achieve 31.62~dB/0.867. This is expected: when all sinogram entries have the same noise level, unit precision is an appropriate model.

The situation changes under Poisson noise. FBP is strongly corrupted by noise, and DDS also degrades substantially to 23.08~dB/0.339. DPS is more stable, reaching 28.62~dB/0.800. The unit-precision version of our method performs worse than DPS, with 28.23~dB/0.691. This shows that using the wrong noise model can reduce the benefit of the geometry-aware update: if all sinogram bins are treated as equally reliable, noisy low-photon-count measurements can be over-weighted.

With signal-dependent precision, the proposed method improves to 29.84~dB/0.846. This is a gain of 1.61~dB and 0.155 SSIM over the unit-precision variant, and it also outperforms DPS by 1.22~dB and 0.046 SSIM. The improvement is visible in the cyan zoomed region around the vertebral body. The unit-precision reconstruction contains stronger noise-driven artifacts and a less clear cortical bone boundary, whereas the signal-dependent reconstruction gives a cleaner and sharper structure.

This behaviour is consistent with the Poisson noise model. Rays passing through high-attenuation tissue, such as bone, produce fewer detected photons and therefore have higher measurement variance. The signal-dependent precision down-weights these less reliable measurements in both the data-fidelity term and the pullback metric. As a result, the reconstruction is less dominated by noisy sinogram entries and relies more appropriately on the learned prior in uncertain regions.

These results show that signal-dependent precision is important for realistic Poisson CT noise. At the same time, it is complementary to the geometry-aware part of the method. In the noise-free and Gaussian-noise settings, where signal-dependent weighting provides no additional advantage, the proposed method already performs strongly. Under Poisson noise, the weighting becomes necessary to preserve this advantage by matching the data-consistency term to the physical noise model.

%%%%%%%%%%%%%%%%%%%%%%%%%%%%%%%%%%%%%%%%%%%%%%%%%%%%%%%%%%%%%%%%%%%%% Conclusion and Future %%%%%%%%%%%%%%%%%%%%%%%%%%%%%%%%%%
\FloatBarrier
\section{Conclusion and Future Work}
\label{sec:lConclusion}
%%%%%%%%%%%%%%%%%%%%%%%%%%%%%%%%%%%%%%%%%%%%%%%%%%%%%%%%%%%%%%%%%%%
%%%%%%%%%%%%%%%%%%%%%%%%%%%%%%%%%%%%%%%%%%%%%%%%%%%%%%%%%%%%%%%%%%%

We presented a geometry-aware posterior sampling method for sparse-view and limited-angle CT reconstruction. The method combines a diffusion prior with data consistency and a pullback-metric update. Instead of applying the diffusion update directly in image space, we use the measurement-induced metric $G_\gamma$
% \begin{equation}
% G_\gamma
% =A_\Omega^\top \Sigma^{-1} A_\Omega
% +\gamma I
% \label{eq:conclusion_pullback_metric}
% \end{equation}
to adapt the update direction to the CT acquisition geometry. This reduces updates that would strongly change the measured sinogram, while allowing the learned prior to act in directions that are weakly constrained by the measurements.

The experiments show that this geometry-aware design improves reconstruction quality across sparse-view, noisy, and limited-angle CT settings. The proposed method reduces streak artifacts, suppresses noise amplification, and preserves anatomical boundaries more clearly than the posterior-sampling baselines considered. The pullback update also improves PSNR and SSIM compared with the Euclidean update, while keeping the measurement residual very small. The update-direction analysis further shows that the pullback metric reduces the measurement response of the deterministic update and redirects it toward approximately measurement-invisible directions.

The method also provides posterior uncertainty estimates. The posterior mean gives the final reconstruction, while the empirical sample standard-deviation map shows where the reconstruction is uncertain. In our experiments, Uncertainty is highest near anatomical boundaries, fine structures, and noise-sensitive regions, as well as structures weakly constrained by missing angular coverage in limited-angle CT. These regions often coincide with larger reconstruction errors, suggesting that the maps capture reconstruction ambiguity. By showing which structures are less reliably determined by the measured data, the maps improve the interpretability of learned reconstructions in ambiguous inverse problems.

Several directions remain for future work. First, although we test Poisson-corrupted measurements, the current posterior model uses a Gaussian approximation to the CT noise model. This approximation is reasonable when photon counts are not too low, but it may become inaccurate in very low-dose CT. A natural extension is to use the Poisson likelihood more directly. In that case, the data-fidelity term and the pullback metric would depend on the predicted photon counts at the current reconstruction. This could give a more accurate uncertainty model in strongly noise-dominated settings.

Second, the current method requires iterative linear solves to apply the inverse pullback metric during the reverse process. This makes the method more expensive than a Euclidean update. Future work should therefore develop faster matrix-free methods for applying $G_\gamma^{-1}$.
% \begin{equation}
% G_\gamma^{-1},
% \qquad
% G_\gamma
% =A_\Omega^\top \Sigma^{-1} A_\Omega
% +\gamma I .
% \label{eq:conclusion_inverse_pullback_metric}
% \end{equation}
Possible directions include preconditioned conjugate gradients, warm-starting the solver across reverse iterations, adaptive stopping criteria, multigrid methods, low-rank approximations, and learned approximations of the inverse metric. Faster solvers would make the method easier to scale to larger images and three-dimensional CT volumes. 

Finally, future work will consider latent-space versions of the method, to accelerate reconstruction speed, which would require a latent pullback metric. 

% Finally, future work will consider latent-space versions of the method. Instead of sampling directly in image space, we could write the reconstruction as $x=D(z)$, where $D$ is a learned decoder, and then sample over the latent variable $z$. The corresponding latent pullback metric would be
% \begin{equation}
% G_z
% =J_D(z)^\top
% A_\Omega^\top
% \Sigma^{-1}
% A_\Omega
% J_D(z)
% +\gamma I ,
% \label{eq:conclusion_latent_pullback_metric}
% \end{equation}
% where $J_D(z)$ is the decoder Jacobian. This could reduce the sampling dimension and make the updates more efficient. However, it also introduces new challenges, such as possible information loss in the decoder, limited latent-space expressivity, and the cost of approximating decoder Jacobians.
\FloatBarrier
\bibliographystyle{plain}
\bibliography{references}

%%%%%%%%%%%%%%%%%%%%%%%%%%%%%%%%%%%%%%%%%%%%%%%%%%%%%
%%%%%%%%%%%%%%%%%%%%%%%%%%%%%%%%%%%%%%%%%%%%%%%%%%%%%
%%%%%%%%%%%%%%%%%%%%%%%%%%%%%%%%%%%%%%%%%%%%%%%%%%%%%
%%%%%%%%%%%%%%%%%%%%%%%%%%%%%%%%%%%%%%%%%%%%%%%%%%%%%
%%%%%%%%%%%%%%%%%%%%%%%%%%%%%%%%%%%%%%%%%%%%%%%%%%%%%
%%%%%%%%%%%%%%%%%%%%%%%%%%%%%%%%%%%%%%%%%%%%%%%%%%%%%

\appendix
\section{Theoretical Support}
\label{appendix}

\subsection{Gaussian Approximation of CT Measurements}
\label{app:gaussian-approx}

This subsection motivates the approximate linear--Gaussian model
\begin{equation*}
    y=A_\Omega x+\varepsilon,
    \qquad
    \varepsilon\sim\mathcal N(0,\Sigma),
\end{equation*}
used in the main text. The derivation starts from an idealized Poisson
photon-counting model and explains the diagonal, signal-dependent covariance
of the post-log measurements. It also describes the fixed plug-in precision
used during reconstruction.

\subsubsection{Poisson photon-counting model}

Under a monoenergetic Beer--Lambert model, the expected photon count at
detector bin $i$ is~\cite{Whiting2006Sinogram}
\begin{equation}
    \lambda_i
    \;=\;
    I_0e^{-(A_\Omega x)_i},
    \label{eq:beer-lambert}
\end{equation}
where $I_0$ is the incident photon count, $x\in\mathbb R^n$ is the image of
linear-attenuation coefficients, and $(A_\Omega x)_i$ is the line integral
for ray $i$. The detected count is modeled as
\begin{equation}
    Y_i\sim\operatorname{Poisson}(\lambda_i),
    \qquad
    \mathbb E[Y_i]=\lambda_i,
    \qquad
    \operatorname{Var}(Y_i)=\lambda_i.
\end{equation}
Thus, the variance of the raw count depends on its expected value.

This model is an idealization. Real CT measurements may additionally include
electronic noise, scatter, beam hardening, detector correlations, and effects
introduced by calibration and preprocessing~\cite{Whiting2006Sinogram,xu2009electronic}.

\subsubsection{Gaussian approximation of Poisson counts}

For a sufficiently large expected transmitted count, a Poisson random
variable is well approximated by a Gaussian random variable with the same
mean and variance.
\vspace{0.5em}
\begin{proposition}[Gaussian limit of the Poisson distribution]
\label{prop:poisson-gaussian}
Let $Y_\lambda\sim\operatorname{Poisson}(\lambda)$ and define
\begin{equation*}
    Z_\lambda
    \;:=\;
    \frac{Y_\lambda-\lambda}{\sqrt{\lambda}}.
\end{equation*}
Then $Z_\lambda\xrightarrow{d}\mathcal N(0,1)$ as $\lambda\to\infty$.
\end{proposition}

\begin{proof}
The moment-generating function of $Y_\lambda$ is
\begin{align*}
    M_{Y_\lambda}(t)
    &\;=\;
    \mathbb E[e^{tY_\lambda}]\\
    &\;=\;
    \sum_{k=0}^{\infty}
    e^{tk}\frac{\lambda^k e^{-\lambda}}{k!}\\
    &\;=\;
    \exp\!\left(\lambda(e^t-1)\right).
\end{align*}
Consequently,
\begin{align*}
    M_{Z_\lambda}(t)
    &\;=\;
    e^{-t\sqrt{\lambda}}
    M_{Y_\lambda}\!\left(\frac{t}{\sqrt{\lambda}}\right)\\
    &\;=\;
    \exp\!\left[
        -t\sqrt{\lambda}
        +\lambda\left(e^{t/\sqrt{\lambda}}-1\right)
    \right].
\end{align*}
For fixed $t$,
\begin{equation*}
    e^{t/\sqrt{\lambda}}
    \;=\;
    1+\frac{t}{\sqrt{\lambda}}
    +\frac{t^2}{2\lambda}
    +O(\lambda^{-3/2}),
\end{equation*}
and therefore
\begin{equation*}
    M_{Z_\lambda}(t)
    \;=\;
    \exp\!\left(
        \frac{t^2}{2}+O(\lambda^{-1/2})
    \right)
    \longrightarrow
    \exp\!\left(\frac{t^2}{2}\right).
\end{equation*}
The limiting function is the moment-generating function of
$\mathcal N(0,1)$. The continuity theorem for moment-generating functions
then gives the convergence in distribution.
\end{proof}

Equivalently, for sufficiently large $\lambda$,
\begin{equation*}
    Y_\lambda
    \approx
    \lambda+\sqrt{\lambda}\,\eta,
    \qquad
    \eta\sim\mathcal N(0,1),
\end{equation*}
or $Y_\lambda\approx\mathcal N(\lambda,\lambda)$. The quality of this
approximation depends on the transmitted count $\lambda$, rather than only
on the incident count $I_0$. A highly attenuated ray can therefore remain in
the low-count regime even when $I_0$ is relatively large
\cite{Whiting2006Sinogram,Wang2017HybridPrePost}.

\subsubsection{Propagation through the logarithmic transform}

Let $\widetilde Y_i>0$ denote the detected count after zero-count
correction, and define the post-log measurement
\begin{equation}
    y_i
    \;:=\;
    -\log\!\left(\frac{\widetilde Y_i}{I_0}\right).
    \label{eq:log-transform}
\end{equation}
For example, zero counts may be replaced using
\begin{equation*}
    \widetilde Y_i
    \;=\;
    \max(Y_i,Y_{\min}),
    \qquad Y_{\min}>0.
\end{equation*}
This correction prevents infinite log-sinogram values. However, it modifies
the count distribution and may introduce additional bias when the detected
photon count is low~\cite{Wang2017HybridPrePost}.

To derive a local Gaussian approximation, we temporarily neglect count
clipping and assume $\widetilde Y_i=Y_i>0$. Define
\begin{equation*}
    f(u)
    \;=\;
    -\log(u/I_0),
    \qquad
    f'(u)
    \;=\;
    -\frac{1}{u},
    \qquad
    f''(u)
    \;=\;
    \frac{1}{u^2}.
\end{equation*}
Since
\begin{equation*}
    f(\lambda_i)
    \;=\;
    -\log\!\left(\frac{\lambda_i}{I_0}\right)
    \;=\;
    (A_\Omega x)_i,
\end{equation*}
a first-order Taylor expansion of $f(Y_i)$ around $\lambda_i$ gives
\begin{align*}
    y_i
    &\approx
    f(\lambda_i)
    +
    f'(\lambda_i)(Y_i-\lambda_i)\\
    &=
    (A_\Omega x)_i
    -
    \frac{Y_i-\lambda_i}{\lambda_i}.
\end{align*}
For sufficiently large $\lambda_i$, Proposition~\ref{prop:poisson-gaussian}
gives
\begin{equation*}
    Y_i-\lambda_i
    \approx
    \sqrt{\lambda_i}\,\eta_i,
    \qquad
    \eta_i\sim\mathcal N(0,1).
\end{equation*}
Therefore,
\begin{equation}
    y_i
    \approx
    (A_\Omega x)_i
    -
    \frac{1}{\sqrt{\lambda_i}}\eta_i.
    \label{eq:post-log-local-gaussian}
\end{equation}
Because $-\eta_i$ has the same distribution as $\eta_i$, the leading
post-log error is approximately Gaussian with zero mean and variance
$\lambda_i^{-1}$.

The first-order approximation does not capture the bias introduced by the
nonlinear logarithm. Retaining the quadratic term gives
\begin{align*}
    \mathbb E[y_i]
    &\approx
    f(\lambda_i)
    +
    \frac{1}{2}f''(\lambda_i)
    \operatorname{Var}(Y_i)\\
    &=
    (A_\Omega x)_i
    +
    \frac{1}{2\lambda_i},
\end{align*}
where $\operatorname{Var}(Y_i)=\lambda_i$. The post-log variance is approximately
\begin{equation*}
    \operatorname{Var}(y_i)
    \;\approx \;
    \bigl(f'(\lambda_i)\bigr)^2\operatorname{Var}(Y_i)
    \;=\;
    \frac{1}{\lambda_i}.
\end{equation*}
The zero-mean Gaussian model used below neglects the leading bias
$1/(2\lambda_i)$. This approximation is most accurate when the transmitted
photon count $\lambda_i$ is sufficiently large.

\subsubsection{Post-log covariance and likelihood}

Under the additional approximation that detector-bin counts are
conditionally independent, the post-log errors are approximately
independent. Stacking the measurements gives
\begin{equation}
    y
    \approx
    A_\Omega x+\varepsilon,
    \qquad
    \varepsilon\sim\mathcal N(0,\Sigma),
    \qquad
    \Sigma
    \;=\;
    \operatorname{diag}\!\left(
        \frac{1}{\lambda_1},\ldots,\frac{1}{\lambda_m}
    \right).
    \label{eq:gaussian-final}
\end{equation}
The term \emph{post-log} indicates that this Gaussian model is applied after
the logarithmic transformation of the photon counts
\cite{Whiting2006Sinogram,Wang2017HybridPrePost}.

Using~\eqref{eq:beer-lambert},
\begin{equation*}
    \operatorname{Var}(\varepsilon_i)
    \;\approx \;
    \frac{1}{\lambda_i}
    \;=\;
    \frac{\exp((A_\Omega x)_i)}{I_0}.
\end{equation*}
Thus, highly attenuated rays receive lower precision. This motivates the
post-log weighted least-squares functional
\cite{Wang2017HybridPrePost,Zeng2017NoiseWeighting}
\begin{equation*}
    \Phi(x)
    \;=\;
    \frac{1}{2}
    \left\|
        \Sigma^{-1/2}(A_\Omega x-y)
    \right\|_2^2.
\end{equation*}

\subsubsection{Fixed effective precision}

The ideal post-log precision depends on the unknown transmitted count
$\lambda_i=I_0\exp[-(A_\Omega x)_i]$. An ideal count-based plug-in estimate
would be
\begin{equation}
    \widehat{\lambda}_i
    \;=\;
    I_0e^{-y_i}.
\end{equation}
In the implementation, we instead use a stabilized measurement-dependent
precision proxy. Zero detected counts are replaced by one count before
applying the logarithm. From the resulting post-log measurement $y_i$, we
first compute
\begin{equation}
    w_i^{\mathrm{raw}}
    \;=\;
    \frac{1}{y_i+\epsilon}.
\end{equation}
The weights are then normalized, scaled, and clipped:
\begin{equation}
    w_i
    \;=\;
    \operatorname{clip}
    \left(
        s\,
        \frac{w_i^{\mathrm{raw}}}
        {\frac{1}{m}\sum_{j=1}^m w_j^{\mathrm{raw}}+\delta},
        w_{\min},
        w_{\max}
    \right),
    \label{eq:precision-plugin}
\end{equation}
where $\epsilon>0$ prevents division by small values, $\delta=10^{-8}$
stabilizes the mean normalization, $s>0$ controls the overall precision
scale, and $w_{\min}$ and $w_{\max}$ limit the dynamic range.

We define the fixed effective precision and covariance by
\begin{equation}
    \widehat{\Sigma}^{-1}
    \;:=\;
    \operatorname{diag}(w_1,\ldots,w_m),
    \qquad
    \widehat{\Sigma}
    \;:=\;
    \operatorname{diag}(w_1^{-1},\ldots,w_m^{-1}).
\end{equation}
Once computed from the measured sinogram, these matrices are held fixed
throughout reconstruction. Unless stated otherwise, subsequent occurrences
of $\Sigma$ denote this fixed effective covariance
$\widehat{\Sigma}$, including the normalization and clipping above.

The resulting weights should be interpreted as a stabilized
measurement-dependent precision approximation rather than an unbiased
estimator of the ideal Poisson precision.

\paragraph{Validity regime.}
The Gaussian approximation is most accurate when all relevant transmitted
counts are sufficiently large. At low counts, the normal approximation
deteriorates, the logarithmic bias becomes appreciable, and zero-count
handling becomes important. A pre-log Poisson model or a hybrid pre-/post-log
model can be preferable in that regime
\cite{Whiting2006Sinogram,Wang2017HybridPrePost}.

For completeness, the Fisher information of the ideal raw-count model with
respect to $x$ is
\begin{equation*}
    \mathcal I(x)
    \;=\;
    A_\Omega^\top
    \operatorname{diag}(\lambda_1(x),\ldots,\lambda_m(x))
    A_\Omega.
\end{equation*}
Indeed,
$\nabla_x\lambda_i(x)=-\lambda_i(x)a_i$, where $a_i^\top$ is row $i$ of
$A_\Omega$, and hence
\begin{equation*}
    \mathcal I(x)
    \;=\;
    \sum_{i=1}^m
    \frac{1}{\lambda_i(x)}
    \nabla_x\lambda_i(x)\nabla_x\lambda_i(x)^\top
    \;=\;
    \sum_{i=1}^m\lambda_i(x)a_i a_i^\top.
\end{equation*}
A metric based directly on this quantity would depend on $x$. Such a
position-dependent metric generally requires additional drift terms in a
Riemannian Langevin construction; these terms vanish for the fixed metric
used in this work~\cite{Girolami2011Riemannian}.

\subsection{Orthogonal Decomposition of the Image Space}
\label{app:orthogonal-decomposition}

This subsection proves the decomposition stated in
Section~\ref{sub:direct_sum}.

\begin{proof}[Proof of the orthogonal direct-sum decomposition]
Let $r\in\mathcal R(A_\Omega^\top)$ and
$v\in\mathcal N(A_\Omega)$. There exists $z\in\mathbb R^m$ such that
$r=A_\Omega^\top z$, and therefore
\begin{equation*}
    r^\top v
    \;=\;
    z^\top A_\Omega v
    \;=\;
    0.
\end{equation*}
Thus,
$\mathcal R(A_\Omega^\top)\perp\mathcal N(A_\Omega)$. Moreover,
\begin{align*}
    \dim\mathcal R(A_\Omega^\top)
    +\dim\mathcal N(A_\Omega)
    &\;=\;
    \operatorname{rank}(A_\Omega)
    +n-\operatorname{rank}(A_\Omega)\\
    &\;=\;n.
\end{align*}
The two orthogonal subspaces therefore span $\mathbb R^n$, giving
\begin{equation*}
    \mathbb R^n
    \;=\;
    \mathcal R(A_\Omega^\top)
    \oplus
    \mathcal N(A_\Omega),
    \qquad
    \mathcal R(A_\Omega^\top)
    \;=\;
    \mathcal N(A_\Omega)^\perp.
\end{equation*}
\end{proof}

\subsection{Posterior Geometry Under Incomplete-View CT}
\label{app:posterior-geometry}

Consider
\begin{equation*}
    y\;=\;A_\Omega x+\varepsilon,
    \qquad
    \varepsilon\sim\mathcal N(0,\Sigma),
    \qquad
    \Sigma\succ0,
\end{equation*}
and let $p_0(x)$ denote the clean-image prior density. The posterior is
\begin{equation}
    p(x\mid y)
    \propto
    p_0(x)
    \exp\!\left[
        -\frac{1}{2}
        \left\|
            \Sigma^{-1/2}(A_\Omega x-y)
        \right\|_2^2
    \right],
\end{equation}
which is the standard Bayesian formulation of a linear inverse problem with
Gaussian noise~\cite{Stuart2010BayesianInverse}.

Assume that $\log p_0(x)$ is twice differentiable in a neighborhood of a posterior mode $x^\star$, and define the local prior curvature by
\begin{equation}
    H_{\mathrm{prior}}(x^\star)
    \;:=\;
    -\left.
    \nabla_x^2\log p_0(x)
    \right|_{x=x^\star}.
    \label{eq:prior-hessian}
\end{equation}
Up to an additive constant, the negative log-posterior is
\begin{equation*}
    -\log p(x\mid y)
    \;=\;
    -\log p_0(x)
    +
    \frac{1}{2}
    (A_\Omega x-y)^\top
    \Sigma^{-1}
    (A_\Omega x-y).
\end{equation*}
Since $A_\Omega$ is linear and $\Sigma$ is fixed, taking two derivatives
with respect to $x$ gives
\begin{equation}
    H_{\mathrm{post}}(x^\star)
    \;=\;
    H_{\mathrm{prior}}(x^\star)
    +
    A_\Omega^\top\Sigma^{-1}A_\Omega,
    \label{eq:posterior-hessian}
\end{equation}
where
\begin{equation*}
    H_{\mathrm{post}}(x^\star)
    \;:=\;
    -\left.
    \nabla_x^2\log p(x\mid y)
    \right|_{x=x^\star}.
\end{equation*}

Assume that $x^\star$ is a nondegenerate local posterior mode, so that
\begin{equation*}
    H_{\mathrm{post}}(x^\star)\succ0.
\end{equation*}
To obtain the local Laplace approximation, define the negative
log-posterior
\begin{equation*}
    \mathcal L(x)
    \;:=\;
    -\log p(x\mid y).
\end{equation*}
A second-order Taylor expansion around $x^\star$ gives
\begin{align*}
    \mathcal L(x)
    &\; \approx \;
    \mathcal L(x^\star)
    +
    \nabla\mathcal L(x^\star)^\top(x-x^\star)\\
    &\quad+
    \frac{1}{2}
    (x-x^\star)^\top
    H_{\mathrm{post}}(x^\star)
    (x-x^\star).
\end{align*}
Because $x^\star$ is a local posterior mode,
$\nabla\mathcal L(x^\star)=0$. Therefore,
\begin{equation*}
    p(x\mid y)
    \propto
    \exp\!\left[
        -\frac{1}{2}
        (x-x^\star)^\top
        H_{\mathrm{post}}(x^\star)
        (x-x^\star)
    \right]
\end{equation*}
locally around $x^\star$. This is the kernel of a Gaussian distribution,
and hence
\begin{equation}
    p(x\mid y)
    \;\approx \;
    \mathcal N\!\left(
        x^\star,
        H_{\mathrm{post}}(x^\star)^{-1}
    \right)
    \label{eq:local-laplace-approximation}
\end{equation}
in a neighborhood of the mode~\cite{TierneyKadane1986Laplace}.

\vspace{0.5em}
\begin{proposition}[Anisotropic posterior curvature]
\label{prop:anisotropic-curvature}
Suppose $\Sigma\succ0$. Then:
\begin{enumerate}
    \item for every $v\in\mathcal N(A_\Omega)$,
    \begin{equation*}
        A_\Omega^\top\Sigma^{-1}A_\Omega v=0;
    \end{equation*}

    \item for every $v\in\mathcal N(A_\Omega)$,
    \begin{equation*}
        v^\top H_{\mathrm{post}}(x^\star)v
        \;=\;
        v^\top H_{\mathrm{prior}}(x^\star)v;
    \end{equation*}

    \item for every nonzero $u\in\mathcal R(A_\Omega^\top)$,
    \begin{equation*}
        u^\top A_\Omega^\top\Sigma^{-1}A_\Omega u>0.
    \end{equation*}
\end{enumerate}
\end{proposition}

\begin{proof}
If $v\in\mathcal N(A_\Omega)$, then $A_\Omega v=0$. The first statement
follows immediately, and substitution into
\eqref{eq:posterior-hessian} proves the second.

For the third statement, let
$0\neq u\in\mathcal R(A_\Omega^\top)$. By
$\mathcal R(A_\Omega^\top)=\mathcal N(A_\Omega)^\perp$, we have
$u\notin\mathcal N(A_\Omega)$ and hence $A_\Omega u\neq0$. Therefore,
\begin{equation*}
    u^\top A_\Omega^\top\Sigma^{-1}A_\Omega u
    \;=\;
    \left\|\Sigma^{-1/2}A_\Omega u\right\|_2^2
    >0.
\end{equation*}
This argument does not require $A_\Omega$ to have full row rank.
\end{proof}

The proposition shows that the likelihood adds no curvature along exact
null-space directions and positive curvature along every nonzero direction
in $\mathcal R(A_\Omega^\top)$. The positive curvature can nevertheless be
very small along directions associated with small singular values of
$\Sigma^{-1/2}A_\Omega$; these are the weakly constrained, or effective
null-space, directions discussed in the main text. The resulting posterior
covariance also depends on the prior curvature and its cross-directional
couplings. It is therefore not generally correct to claim that every
null-space direction has larger posterior variance than every range-space
direction.

\subsection{Projection Operators and Conjugate Gradient Computation}
\label{app:projection-operators}

\begin{proposition}[Exact orthogonal projectors]
\label{prop:projection-operators}
Define
\begin{equation*}
    P_R
    \;:=\;
    A_\Omega^\top
    (A_\Omega A_\Omega^\top)^\dagger
    A_\Omega,
    \qquad
    P_N\;:=\;I_n-P_R,
\end{equation*}
where ${}^\dagger$ denotes the Moore--Penrose pseudoinverse. Then $P_R$
and $P_N$ are the Euclidean orthogonal projectors onto
$\mathcal R(A_\Omega^\top)$ and $\mathcal N(A_\Omega)$, respectively.
Consequently,
\begin{equation*}
    P_R^2\;=\;P_R,
    \quad
    P_N^2\;=\;P_N,
    \quad
    P_R^\top\;=\;P_R,
    \quad
    P_N^\top\;=\;P_N,
    \quad
    P_R P_N\;=\;P_N P_R\;=\;0,
    \quad
    P_R+P_N\;=\;I_n.
\end{equation*}
\end{proposition}

\begin{proof}
Let
$A_\Omega=U_r S_r V_r^\top$ be a compact singular-value decomposition,
where $S_r$ contains the positive singular values. Then
\begin{equation*}
    (A_\Omega A_\Omega^\top)^\dagger
    \;=\;
    U_r S_r^{-2}U_r^\top,
\end{equation*}
and hence
\begin{equation*}
    P_R
    \;=\;
    V_r V_r^\top.
\end{equation*}
The columns of $V_r$ form an orthonormal basis of
$\mathcal R(A_\Omega^\top)$, so $P_R$ is its Euclidean orthogonal
projector. The orthogonal complement is
$\mathcal N(A_\Omega)$, and therefore $P_N=I_n-P_R$ is the corresponding
null-space projector. All stated identities follow directly.
\end{proof}

\paragraph{Regularized range and null filters.}
For $\gamma_\star>0$, define
\begin{equation}
    P_{R,\gamma_\star}^{\Sigma}
    \;:=\;
    A_\Omega^\top
    (A_\Omega A_\Omega^\top+\gamma_\star\Sigma)^{-1}
    A_\Omega.
    \label{eq:regularized-range-filter}
\end{equation}
Let
$B:=A_\Omega^\top\Sigma^{-1}A_\Omega$. Direct multiplication shows that
\begin{equation*}
    (B+\gamma_\star I_n)^{-1}
    A_\Omega^\top\Sigma^{-1}
    \;=\;
    A_\Omega^\top
    (A_\Omega A_\Omega^\top+\gamma_\star\Sigma)^{-1}.
\end{equation*}
It follows that
\begin{equation*}
    P_{R,\gamma_\star}^{\Sigma}
    \;=\;
    (B+\gamma_\star I_n)^{-1}B,
\end{equation*}
and its complementary filter is
\begin{equation*}
    P_{N,\gamma_\star}^{\Sigma}
    \;:=\;
    I_n-P_{R,\gamma_\star}^{\Sigma}
    \;=\;
    \gamma_\star(B+\gamma_\star I_n)^{-1}.
\end{equation*}
For finite $\gamma_\star$, these matrices are symmetric regularized
filters, not exact projectors: their nontrivial eigenvalues lie strictly
between zero and one. As $\gamma_\star\downarrow0$,
\begin{equation*}
    P_{R,\gamma_\star}^{\Sigma}\longrightarrow P_R,
    \qquad
    P_{N,\gamma_\star}^{\Sigma}\longrightarrow P_N.
\end{equation*}

\subsubsection{Matrix-free computation via conjugate gradient}
\label{app:CG}

The method uses two types of regularized linear systems. Both coefficient
matrices are symmetric positive definite and can be solved by conjugate
gradient~\cite{Hestenes1952CG,Hansen2018AIRToolsII} without forming dense matrices.

\paragraph{Image-space system.}
Pullback-metric preconditioning uses
\begin{equation}
    (A_\Omega^\top\Sigma^{-1}A_\Omega+\gamma I_n)u\;=\;r.
    \label{eq:cg-image-system}
\end{equation}
For $p\in\mathbb R^n$, its matrix--vector product is
\begin{equation}
    (A_\Omega^\top\Sigma^{-1}A_\Omega+\gamma I_n)p
    \;=\;
    A_\Omega^\top[\Sigma^{-1}(A_\Omega p)]+\gamma p.
    \label{eq:cg-image-matvec}
\end{equation}
Thus, each CG iteration uses one forward projection, one application of the
diagonal precision, and one backprojection.

\paragraph{Measurement-space system.}
The approximately null-space-restricted perturbation and the regularized
data-consistency correction use
\begin{equation}
    (A_\Omega A_\Omega^\top+\gamma_\star\Sigma)h=b,
    \qquad
    \gamma_\star>0,
    \label{eq:cg-measurement-system}
\end{equation}
where $\gamma_\star=\gamma_{\mathrm{null}}$ for the null-space operation
and $\gamma_\star=\gamma_{\mathrm{proj}}$ for the data-consistency
correction. For $p\in\mathbb R^m$,
\begin{equation}
    (A_\Omega A_\Omega^\top+\gamma_\star\Sigma)p
    \;=\;
    A_\Omega(A_\Omega^\top p)+\gamma_\star\Sigma p.
    \label{eq:cg-measurement-matvec}
\end{equation}
For every nonzero $p$,
\begin{equation*}
    p^\top(A_\Omega A_\Omega^\top+\gamma_\star\Sigma)p
    \;=\;
    \|A_\Omega^\top p\|_2^2
    +\gamma_\star p^\top\Sigma p
    >0,
\end{equation*}
so the coefficient matrix is symmetric positive definite even when
$A_\Omega$ is rank deficient.

\vspace{0.5em}
\begin{remark}
Premultiplying~\eqref{eq:cg-measurement-system} by $\Sigma^{-1}$ gives the
algebraically equivalent system
\begin{equation*}
    (\Sigma^{-1}A_\Omega A_\Omega^\top+\gamma_\star I_m)h
    \;=\;
    \Sigma^{-1}b.
\end{equation*}
However, its coefficient matrix is generally not symmetric in the Euclidean
inner product when $\Sigma$ is non-scalar. Ordinary conjugate gradient
should therefore be applied to the symmetric system
\eqref{eq:cg-measurement-system}, not directly to this premultiplied form.
\end{remark}

Algorithm~\ref{alg:CG} states the measurement-space recurrence. The same
recurrence applies to~\eqref{eq:cg-image-system} after replacing the
matrix--vector product by~\eqref{eq:cg-image-matvec}.

\begin{algorithm}[t]
\caption{Conjugate Gradient for the Noise-Weighted Measurement-Space System}
\label{alg:CG}
\begin{algorithmic}[1]
\Require Right-hand side $b\in\mathbb R^m$, operator $A_\Omega$,
         covariance $\Sigma$, regularization $\gamma_\star>0$,
         relative tolerance $\tau$, maximum iterations $K$
\Ensure Approximate solution
        $h\approx(A_\Omega A_\Omega^\top+\gamma_\star\Sigma)^{-1}b$

\State $h^{(0)}\gets0$
\State $r^{(0)}\gets b$
\State $p^{(0)}\gets r^{(0)}$

\If{$\|b\|_2=0$}
    \State \Return $h^{(0)}$
\EndIf

\For{$j=0,1,\ldots,K-1$}
    \State $q^{(j)}\gets
    A_\Omega(A_\Omega^\top p^{(j)})
    +\gamma_\star\Sigma p^{(j)}$

    \State $a_j\gets
    \langle r^{(j)},r^{(j)}\rangle/
    \langle p^{(j)},q^{(j)}\rangle$

    \State $h^{(j+1)}\gets h^{(j)}+a_j p^{(j)}$
    \State $r^{(j+1)}\gets r^{(j)}-a_j q^{(j)}$

    \If{$\|r^{(j+1)}\|_2\leq\tau\|b\|_2$}
        \State \Return $h^{(j+1)}$
    \EndIf

    \State $\beta_j\gets
    \langle r^{(j+1)},r^{(j+1)}\rangle/
    \langle r^{(j)},r^{(j)}\rangle$

    \State $p^{(j+1)}\gets r^{(j+1)}+\beta_j p^{(j)}$
\EndFor

\State \Return $h^{(K)}$
\end{algorithmic}
\end{algorithm}

\subsection{Pullback Metric and Preconditioned Langevin Dynamics}
\label{app:fp-proof}

This subsection concerns the ideal continuous-time process with an exact,
time-independent target score. It does not by itself prove that the
finite-step practical algorithm, which uses a DDPM score approximation and
additional scheduled corrections, produces exact stochastic reconstruction samples.

Let
\begin{equation*}
    \pi(x)
    \propto
    p_0(x)e^{-\Phi(x)},
    \qquad
    \Phi(x)
    \;=\;
    \frac{1}{2}
    \|\Sigma^{-1/2}(A_\Omega x-y)\|_2^2,
\end{equation*}
and define
\begin{equation*}
    G_\gamma
    \;=\;
    A_\Omega^\top\Sigma^{-1}A_\Omega+\gamma I_n,
    \qquad
    \gamma>0.
\end{equation*}
Because $G_\gamma$ is constant and positive definite, so is
$M:=G_\gamma^{-1}$. Consider
\begin{equation}
    dX_s
    \;=\;
    M\nabla\log\pi(X_s)\,ds
    +
    \sqrt{2M}\,dW_s,
    \label{eq:fp-sde}
\end{equation}
where $W_s$ is an $n$-dimensional standard Brownian motion and
$\sqrt{2M}$ denotes any matrix $L$ satisfying $LL^\top=2M$. The variable
$s$ denotes continuous Langevin time and is distinct from a DDPM timestep.

Since $A_\Omega$ and the fixed effective covariance $\Sigma$ do not depend
on $x$, the metric is independent of position. The metric-derivative terms
that arise for a position-dependent Riemannian metric therefore vanish
\cite{Girolami2011Riemannian}.

\subsubsection{Fokker--Planck equation and invariance}

Let $\rho(x,s)$ be the density of $X_s$. Under standard regularity and
nonexplosion assumptions, its Fokker--Planck equation is
\begin{equation}
    \frac{\partial\rho}{\partial s}
    \;=\;
    -\nabla\cdot(M\nabla\log\pi\,\rho)
    +
    \nabla\cdot(M\nabla\rho).
    \label{eq:fp-our}
\end{equation}

\begin{lemma}[Diffusion covariance for the fixed drift]
\label{lem:noise-cov-unique}
Consider the constant-covariance family
\begin{equation}
    dX_s
    \;=\;
    M\nabla\log\pi(X_s)\,ds
    +
    L\,dW_s,
    \qquad
    D\;:=\;LL^\top\succ0.
    \label{eq:sde-family}
\end{equation}
The density $\pi$ is invariant for every smooth target in this family if
and only if $D=2M$.
\end{lemma}

\begin{proof}
If $D=2M$, substitution of $\rho=\pi$ into the Fokker--Planck equation
gives zero because $\pi\nabla\log\pi=\nabla\pi$.

Conversely, suppose invariance holds for every smooth target, and set
$C:=D/2-M$. Stationarity requires
$\nabla\cdot(C\nabla\pi)=0$ for every such $\pi$. Choose the standard
Gaussian density
$\pi(x)\propto\exp(-\|x\|_2^2/2)$. Then
\begin{equation*}
    \nabla\cdot(C\nabla\pi)
    \;=\;
    (x^\top Cx-\operatorname{tr}C)\pi(x).
\end{equation*}
This expression can vanish for every $x$ only if $C=0$, because $C$ is
symmetric. Hence $D=2M$.
\end{proof}

\begin{proposition}[Invariance and reversibility]
\label{prop:fp-invariance}
Assume that $\pi$ is a smooth, strictly positive probability density,
$\nabla\log\pi$ satisfies conditions ensuring a nonexplosive solution of
\eqref{eq:fp-sde}, and the required boundary terms vanish at infinity.
Then $\pi$ is a reversible invariant density of~\eqref{eq:fp-sde}.
\end{proposition}

\begin{proof}
Substituting $\rho=\pi$ into~\eqref{eq:fp-our} gives
\begin{align*}
    \left.\frac{\partial\rho}{\partial s}\right|_{\rho=\pi}
    &\;=\;
    -\nabla\cdot(M\nabla\log\pi\,\pi)
    +\nabla\cdot(M\nabla\pi)\\
    &\;=\;0.
\end{align*}
Thus, $\pi$ is invariant. Moreover, the stationary probability current is
\begin{equation*}
    J(x)
    \;:=\;
    M\nabla\log\pi(x)\,\pi(x)-M\nabla\pi(x)
    \;=\;0.
\end{equation*}
The vanishing current gives detailed balance and hence reversibility under
the stated boundary assumptions.
\end{proof}

\iffalse
\subsubsection{Convergence of the ideal continuous process}

Invariance alone does not guarantee convergence from every initial law.
One sufficient setting is as follows. Suppose $\pi$ satisfies a Poincar\'e
inequality, and let $h_s:=\rho(\cdot,s)/\pi$ with
$h_0-1\in L^2(\pi)$. Uniform positive definiteness of $M$ then yields
exponential $L^2(\pi)$ convergence, and consequently
\begin{equation*}
    \|\rho(\cdot,s)-\pi\|_{\mathrm{TV}}
    \leq
    \frac{1}{2}\|h_s-1\|_{L^2(\pi)}
    \leq
    C e^{-cs}
\end{equation*}
for constants $C,c>0$ determined by the initial density, the Poincar\'e
constant, and $M$; see~\cite{BakryGentilLedoux2014,Pavliotis2014}.

\begin{remark}[Role of $\gamma>0$]
At $\gamma=0$,
$G_0=A_\Omega^\top\Sigma^{-1}A_\Omega$ is singular whenever
$\mathcal N(A_\Omega)$ is nontrivial, so $G_0^{-1}$ is not defined. If one
replaced the inverse by the Moore--Penrose pseudoinverse, both the
preconditioned drift and diffusion would be annihilated on
$\mathcal N(A_\Omega)$, preventing exploration in those directions. Taking
$\gamma>0$ makes $G_\gamma$ positive definite. Along an exact null-space
direction, $G_\gamma^{-1}$ has eigenvalue $1/\gamma$, so decreasing
$\gamma$ increases rather than decreases the scale of the ideal Langevin
update in that direction.
\end{remark}
\fi

\subsubsection{Discrete-time approximation and matrix-free noise}

For an ideal fixed target and a step size $a_k>0$, Euler--Maruyama gives
\begin{equation}
    X^{(k+1)}
    \;=\;
    X^{(k)}
    +a_k G_\gamma^{-1}\nabla\log\pi(X^{(k)})
    +\sqrt{2a_k}\,Z^{(k)},
    \qquad
    Z^{(k)}\sim\mathcal N(0,G_\gamma^{-1}).
    \label{eq:em-discrete}
\end{equation}
The following construction samples $Z^{(k)}$ without forming a matrix square root.

\vspace{0.5em}
\begin{lemma}[Matrix-free noise sampler]
\label{lem:matrix-free-sampler}
Let $\xi\sim\mathcal N(0,I_m)$ and
$c\sim\mathcal N(0,I_n)$ be independent, and define
\begin{equation*}
    w
    \;=\;
    A_\Omega^\top\Sigma^{-1/2}\xi+\sqrt{\gamma}\,c,
    \qquad
    z
    \;=\;
    G_\gamma^{-1}w.
\end{equation*}
Then
$w\sim\mathcal N(0,G_\gamma)$ and
$z\sim\mathcal N(0,G_\gamma^{-1})$.
\end{lemma}

\begin{proof}
The vector $w$ is zero-mean Gaussian and
\begin{equation*}
    \operatorname{Cov}(w)
    \;=\;
    A_\Omega^\top\Sigma^{-1}A_\Omega+\gamma I_n
    \;=\;
    G_\gamma.
\end{equation*}
Since $G_\gamma^{-1}$ is symmetric,
\begin{equation*}
    \operatorname{Cov}(z)
    \;=\;
    G_\gamma^{-1}G_\gamma G_\gamma^{-1}
    \;=\;
    G_\gamma^{-1}.
\end{equation*}
\end{proof}

Drawing $\xi\sim\mathcal N(0,I_n)$ and using
$G_\gamma^{-1}\xi$ would instead produce covariance $G_\gamma^{-2}$ and
is therefore not equivalent.

Euler--Maruyama has weak order one over a fixed finite time interval under
standard smoothness and moment assumptions~\cite{KloedenPlaten1992}. This
fact alone does not prove an $O(a_k)$ bias for the invariant distribution.
Such stationary-bias or nonasymptotic sampling bounds require additional
assumptions, such as smoothness and strong log-concavity
\cite{Dalalyan2017}. Moreover, the practical method uses a time-dependent
DDPM score approximation, finite CG solves, scheduled data-consistency
corrections, and optional additional perturbations.
Proposition~\ref{prop:fp-invariance} therefore justifies the covariance matching in the
ideal continuous model; it does not establish exact posterior invariance of
the complete finite-step algorithm.

\section{Extra Experimental Results}

\subsection{Single-Image Quantitative Comparison}
\label{app:single_image_comparison}

Table~\ref{tab:single_image} presents a supplementary comparison on a
representative test image. For each stochastic diffusion-based method, we
generate $N=20$ reconstructions and compute PSNR and SSIM from their sample
mean; FBP produces a single deterministic reconstruction. Because these
results are based on only one image, they are intended to illustrate
method-specific behavior rather than establish dataset-level performance.
The overall observations are broadly consistent with the full
526-image evaluation in Table~\ref{tab:dm4ct_full}, although the ranking of
individual methods may vary between configurations.

\begin{table}[t]
\centering
\caption{Single-image PSNR (dB) and SSIM under the five DM4CT benchmark
configurations. Config~i: 40 views without noise; Config~ii: 20 views with
mild noise; Config~iii: 80 views with higher noise; Config~iv: 80 views with
noise and ring artifacts; Config~v: 40 views over the limited angular range
$[0,\tfrac{3}{4}\pi]$. For each stochastic method, the metrics are computed
on the mean of $N=20$ reconstructions. FBP is deterministic. Bold and
underlining indicate the best and second-best values, respectively, among
the diffusion-based methods, evaluated separately for each metric. Tied
values at the displayed precision receive the same formatting.}
\label{tab:single_image}
\begin{tabular}{lccccc}
\toprule
Method & Config~i & Config~ii & Config~iii & Config~iv & Config~v \\
\midrule
FBP
  & 26.70\,/\,0.70
  & 9.52\,/\,0.03
  & 12.45\,/\,0.08
  & 13.88\,/\,0.11
  & 25.33\,/\,0.70 \\
\midrule
DDS
  & 31.64\,/\,\underline{0.89}
  & 25.11\,/\,0.43
  & 18.25\,/\,0.18
  & 20.29\,/\,0.23
  & \underline{31.34}\,/\,\underline{0.88} \\
DPS
  & 30.90\,/\,0.87
  & 27.92\,/\,\textbf{0.78}
  & 28.73\,/\,\underline{0.80}
  & 28.80\,/\,\underline{0.81}
  & 30.60\,/\,0.86 \\
PGDM
  & 30.97\,/\,0.86
  & 27.28\,/\,\underline{0.77}
  & 28.59\,/\,\underline{0.80}
  & 28.70\,/\,\underline{0.81}
  & 30.43\,/\,0.85 \\
ReSample
  & \underline{32.37}\,/\,0.88
  & \underline{28.39}\,/\,\underline{0.77}
  & \underline{29.21}\,/\,\underline{0.80}
  & \textbf{29.55}\,/\,\underline{0.81}
  & 28.89\,/\,0.82 \\
\textbf{Ours}
  & \textbf{33.33}\,/\,\textbf{0.90}
  & \textbf{28.47}\,/\,\textbf{0.78}
  & \textbf{29.24}\,/\,\textbf{0.81}
  & \underline{29.40}\,/\,\textbf{0.83}
  & \textbf{33.18}\,/\,\textbf{0.90} \\
\bottomrule
\end{tabular}
\end{table}

In Config~i, the proposed method obtains 33.33~dB/0.90, compared with 32.37~dB/0.88 for ReSample. On this image, the corresponding improvements are 0.96~dB in PSNR and 0.02 in SSIM. This result is consistent with the
full-dataset trend reported in Table~\ref{tab:dm4ct_full}.

In Config~ii, the proposed method achieves the highest PSNR (28.47~dB), exceeding ReSample by 0.08~dB, and ties with DPS for the highest SSIM at the displayed precision (0.78). In Config~iii, it obtains 29.24~dB/0.81, narrowly exceeding ReSample in PSNR (29.21~dB) and
achieving a higher SSIM than DPS, PGDM, and ReSample (0.80). DDS performs substantially worse in Configs~ii--iv than in Config~i. Although this observation is consistent with sensitivity to measurement noise, a
single-image comparison cannot determine the underlying cause.

In Config~iv, ReSample achieves the highest PSNR (29.55~dB), whereas the proposed method achieves the second-highest PSNR (29.40~dB) and the highest SSIM (0.83). One possible explanation is that a strong
data-consistency operation may retain components associated with the ring artifact. Nevertheless, the different PSNR and SSIM rankings only demonstrate a trade-off between reference-based pixel error and
structural similarity on this image; they do not by themselves establish the underlying mechanism.

In Config~v, the proposed method obtains 33.18~dB/0.90. It exceeds DDS by 1.84~dB in PSNR and 0.02 in SSIM, and ReSample by 4.29~dB and 0.08, respectively. This result is consistent with the intended benefit of adapting the reconstruction update to the directional sensitivity of the limited-angle acquisition geometry. Any causal interpretation should, however, be supported by the geometry ablation and full-dataset evaluation rather than by this single-image comparison alone.

\subsection{Sample-Variability Maps for Baseline Methods}
\label{app:baseline_uncertainty}

Figure~\ref{fig:baseline_uncertainty} compares the pixel-wise sample standard deviation computed from $N=20$ stochastic reconstructions for each method under the five controlled acquisition settings introduced in
Section~\ref{subsec:exp_visual_comparison}. The row corresponding to the proposed method is included for direct visual comparison with the diffusion-based baselines. All maps are displayed using the same intensity scale.

\begin{figure}[htbp]
    \centering
    \includegraphics[width=0.95\linewidth]{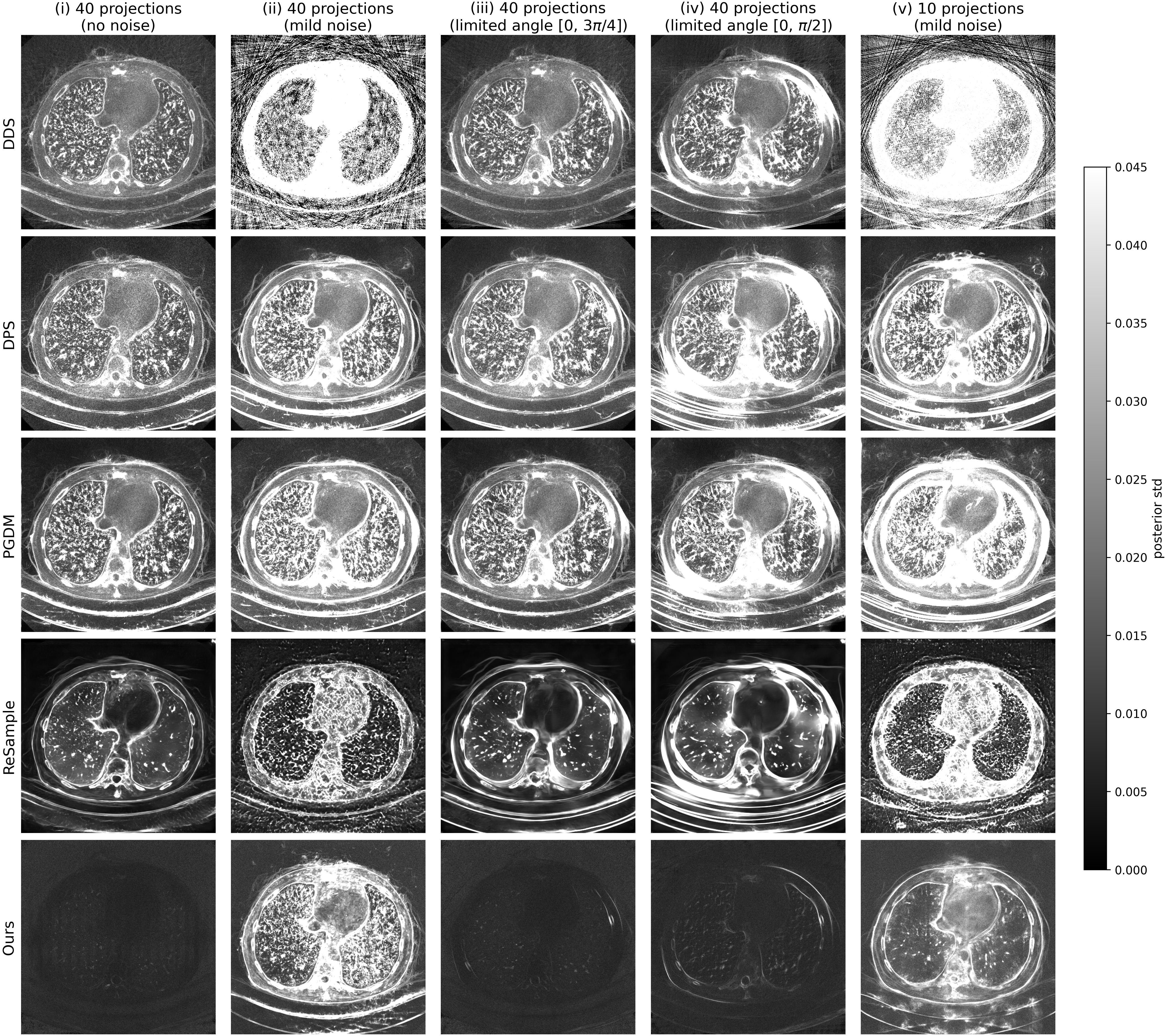}
    \caption{Pixel-wise sample standard deviation over $N=20$
    reconstructions for DDS, DPS, PGDM, and ReSample under the five
    controlled acquisition settings. Rows correspond to methods and columns
    to settings (i--v). Darker shading denotes lower empirical sample
    variability. These maps describe the variability generated by each
    sampler; they do not, by themselves, establish uncertainty calibration.}
    \label{fig:baseline_uncertainty}
\end{figure}

As summarized in Table~\ref{tab:posterior_uncertainty}, the proposed method has a lower mean pixel-wise sample standard deviation than ReSample in Settings~i, iii, iv, and v, whereas ReSample has the lower value in Setting~ii. DDS, DPS, and PGDM have higher mean sample variability than the proposed method in all five settings. Spatially, the baseline maps commonly show increased variability around anatomical boundaries and
undersampling-related streaks. The proposed method produces comparatively concentrated variability in Settings~i, iii, and iv, while its variability
increases under noisy or more severely undersampled acquisition, as seen in Settings~ii and v.

These results should be interpreted together with reconstruction quality and measurement residuals. High sample variability may indicate measurement-induced ambiguity, but it may also arise from sampling
instability. Conversely, low variability may indicate either a concentrated conditional distribution or insufficient exploration by the sampler. Accordingly, the maps should be viewed as empirical sample-variability
diagnostics rather than as evidence of calibrated posterior uncertainty. Their spatial patterns are most informative when considered jointly with the posterior-mean and error maps in Figure~\ref{fig:ours_uncertainty}.

\end{document}